\documentclass{article}
\usepackage{geometry}
\usepackage{url}
\usepackage{bm}
\usepackage{amsthm}
\usepackage{cite}
\usepackage{mathtools}
\usepackage{mathstyle}
\usepackage{color}
\usepackage{amsmath}
\usepackage{footnote}

\usepackage{amssymb}
\usepackage{color}
\usepackage[ruled,linesnumbered]{algorithm2e}
\usepackage{float} 
\newtheorem{property}{Proposition}[section]
\newtheorem{definition}{Definition}[section]
\newtheorem{remark}{Remark}[section]
\newtheorem{lemma}{Lemma}[section]
\usepackage{graphicx}
\usepackage{subcaption}
\usepackage{subcaption}
\usepackage[flushleft]{threeparttable}

\usepackage{newtxtext}
\usepackage{tikz}
\usetikzlibrary{arrows.meta,
                chains,
                positioning}
                
\usepackage{caption}
\usepackage{tabularx}
\usepackage{graphicx}
\usepackage{booktabs}

\usepackage{listings}

\usepackage{multirow}

\usepackage{footnote}

\newtheorem{corollary}{Corollary}[section]

\newcommand{\Probability}{\mathbb{P}}
\newtheorem{theorem}{Theorem}[section]
\newcommand{\R}{\mathbb{R}}

\newcommand{\mw}{\bm{w}}

\newcommand{\mg}{\bm{g}}

\newcommand{\mC}{\mathbb{C}}
\newcommand{\bomega}{\bm{\omega}}

\usepackage{enumitem}

\newcommand{\dynx}{\text{dyn}(\x)}

\newcommand{\grid}{\text{grid}}

\newcommand{\mr}{\bm{r}}

\newcommand{\mB}{\mathcal{B}}

\newcommand{\mT}{\mathcal{T}}
\newcommand{\mI}{\bm{I}}
\newcommand{\Crad}{C_{\text{rad}}}

\newcommand{\W}{\bm{F}}

\newcommand{\mP}{\bm{P}}

\newcommand{\mS}{\mathcal{S}}

\newcommand{\x}{{\bm x}}
\newcommand{\y}{{\bm y}}
\newcommand{\mZ}{{\bm Z}}

\newcommand{\mv}{{\bm v}}
\newcommand{\mz}{{\bm z}}

\newcommand{\e}{{\bm e}}

\newcommand{\mL}{{\mathcal L}}
\newcommand{\btau}{{\bm \tau}}

\newcommand{\K}{\mathcal K}

\newcommand{\w}{{\bm f}}

\title{Performance Analysis of OMP in Super-Resolution}
\author{Yuxuan Han\footnote{Department of Mathematics, The Hong Kong University of Science and Technology(yhanat@connect.ust.hk)},\hspace{0.1cm} Zhiyi Huang\footnote{Department of Computer Science
, University of Hong Kong(zhiyi@cs.hku.hk)},\hspace{0.1cm} Yang Wang\footnote{Department of Mathematics, The Hong Kong University of Science and Technology(yangwang@ust.hk)},\hspace{0.1cm} Rui Zhang\footnote{Theory Lab, Central Research Institute, 2012 Labs, Huawei Technologies Co. Ltd., Hong Kong SAR, China(zhangrui191@huawei.com)} }
\date{}
\begin{document}
\maketitle
\begin{abstract}
Given a spectrally sparse signal $\bm y = \sum_{i=1}^s x_i\w(\tau_i) \in \mathbb{C}^{2n+1}$ consisting of $s$ complex sinusoids, we consider the super-resolution problem, which is about estimating frequency components $\{\tau_i\}_{i=1}^s$ of $\bm y$. We consider the OMP-type algorithms for super-resolution, which is more efficient than other approaches based on Semi-Definite-Programming. Our analysis shows that a two-stage algorithm with OMP initialization can recover frequency components under the separation condition $n\Delta \gtrsim \dynx$ and the dependency on $\dynx$ is inevitable for vanilla OMP algorithm. We further show that the Sliding-OMP algorithm, a variant of the OMP algorithm with an additional refinement step at each iteration, is provable to recover $\{\tau_i\}_{i=1}^s$ if $n\Delta \geq c$.  Moreover, our result can be extended to an incomplete measurement model with $O( s^2\log n)$ measurements.
\end{abstract}

\section{Introduction}
\subsection{Super Resolution}
One fundamental problem  in many industrial applications is  estimating
the modulation parameters (e.g. locations, time delays, etc) from (incomplete) measurements \cite{mccutchen1967superresolution,greenspan2009super,heckel2016super,cai2018spectral}. 
Being limited by sensing or imaging devices, such as
the sampling rate of an analog-to-digital converter, the low temporal
or spatial resolution of the signal is the bottleneck of improving
the performance of denoising or inference. 

In this paper, we consider the super-resolution problem for spectrally sparse signals, which involves extrapolating its frequency information from the low-resolution observation.  
To be precise, our observation is a mixture of $s$ complex sinusoids 
\[
y_{i}=\sum_{k=1}^{s}x_{k}e^{j2\pi\tau_{k}i},-n\leq i\leq n
\]
with unknown frequencies $\mathcal{T}=\{\tau_{1},\cdots,\tau_{s}\}\subset[0,1)$
. In a compact form, we write 
\[
\boldsymbol{y}=\W\boldsymbol{x}
\]
 where 
$\bm F\in\mathbb{C}^{(2n+1)\times s}$ and $x\in\mathbb{C}^{s}$ ($n\gg s$).
Specifically, $\bm F$ has the shape
\[
\W=\left(\begin{array}{cccc}
e^{-j2\pi n\tau_{1}} & e^{-j2\pi n\tau_{2}} & \cdots & e^{-j2\pi n\tau_{s}}\\
e^{-j2\pi(n-1)\tau_{1}} & e^{-j2\pi(n-1)\tau_{2}} & \cdots & e^{-j2\pi (n-1)\tau_{s}}\\
\vdots & \vdots & \ddots & \vdots\\
e^{j2\pi n\tau_{1}} & e^{j2\pi n\tau_{2}} & \cdots & e^{j2\pi n\tau_{s}}
\end{array}\right) = \bigg( \w(\tau_1),\dots,\w(\tau_s)\bigg),
\] 
where $\w(\tau_i)$ is the $i$-th column of $\bm F$. Our goal is to recover $\mT$ from $\bm y.$

\subsection{OMP for sparse representation}

\begin{algorithm}
	\KwIn{$\y$, Stopping threshold $\gamma$, Dictionary $D$}
	\textbf{Initialization:} Setting $\mr_0 = \y, \mathcal{D}_0  = \emptyset, t = 0.$ \\
	\While{$\max_{\w \in D } \lvert \w ^* \mr_{t-1}   \rvert > \gamma $  }
	{
		$t \leftarrow t+ 1. $	
		 $\w_{t} \leftarrow  \text{argmax}_{\w\in D} \lvert \w^* \mr_{t-1}\rvert $.\\
		 $\mathcal{D}_{{t}}\leftarrow \mathcal{D}_{t-1}\cup \{ \bm f_{t} \} $ \\ 
		 $\mr_{t} \leftarrow \big (\mI - {\mP}( \mathcal{D})  \big ) \mr_{t-1} $.\\
			
	}
	\textbf{return}  $\mathcal{D}_{t} .$
	\caption{Orthogonal Matching Pursuit over Dictionary $D$}
	\label{alg-COMP}
\end{algorithm}
One observation is that underlying frequencies $\bm{\tau}\in [0,1)^s$ is the minimizer of the following loss over $[0,1) ^s:$  
\begin{align}\label{eq-global-loss-intro}
\mathcal{L}(\bomega):= \min_{\bm{x}\in \mathbb{C}^s }\text{}  \dfrac{1}{2} \|\bm{y}- \W(\bm{\bomega})\bm{x}\|_{2}^2.
\end{align}

Despite the difficulty of non-convexity and unknown $s$, a line of recent works \cite{Eftekhari2015, Traonmilin2020, traonmilin2020projected, Benard2022} attempted to solve the problem by developing two-stage algorithms based on \eqref{eq-global-loss-intro}: The first stage of the algorithm estimate the spike number $\hat s \in \mathbb{Z}_+$ and an initialization $\hat\bomega $,  and in the second stage, they try to solve $\eqref{eq-global-loss-intro}$ with $\hat{s}$ and $\hat\bomega$ using various optimization methods.

While these two-stage algorithms are numerically efficient, their theoretical guarantees in the $n\Delta \asymp 1$ regime are not well-understood. The theory developed for the algorithm in \cite{Eftekhari2015} only works when $n\Delta > \log n$.   \cite{Traonmilin2020} analyzed the non-convex landscape of \eqref{eq-global-loss-intro} when $\hat{s}=s$ and showed strongly convexity of \eqref{alg-COMP} when  $\lVert \bomega - \btau \rVert_2 \lesssim \dfrac{1}{n\dynx} $, such result is used to design an efficient two-stage algorithm in its follow-up work \cite{traonmilin2020projected}. However, the algorithm in \cite{traonmilin2020projected}  only works in $n\Delta \gtrsim \dynx$ regime due to the dependency on $\dynx$ in the strong convexity result.

 In recent work, \cite{Benard2022}  considered solving \eqref{eq-global-loss-intro} by projected gradient descent with Orthogonal Matching Pursuit (OMP, Algorithm~\ref{alg-COMP}) initialization and provided a promising empirical study. However, no theoretical results are presented for the algorithm.

Besides \cite{Benard2022}, various OMP-based algorithms have  been studied  in super-resolution and DOA estimation literatures \cite{mamandipoor2016newtonized, aich2017grid,emadi2018omp, ganguly2019compressive} due to its numerical efficiency. 
While the OMP algorithm for both discrete and continuous dictionaries is investigated in many previous works, its theoretical guarantee for the super-resolution problem is left open.

In the discrete dictionary setting, the theoretical guarantee of OMP  has been well-studied in previous works \cite{Cai2011, Tropp2004, Tropp2007}. All these works required a low correlation condition between different dictionary entries to show the theoretical success of OMP. However  the analysis is incompatible with the super-resolution scenario since  $D = \{\w(\tau),\tau \in [0,1) \}$ is a continuous dictionary whose entries  have arbitrary large correlation.  Although it is possible to discretize $D$ to convert the problem into the discrete setting \cite{fannjiang2010compressed,aich2017grid}, balancing the trade-off between the model misspecification error \cite{chi2011sensitivity,herman2010general, duarte2013spectral}(which encourages the smaller grid distance) and the correlation condition(which require the large grid distance) is still a challenging problem.

 Previously, the exact recovery guarantee of OMP under continuous dictionaries has been investigated in \cite{Elvira2019, Elvira2021}. 
 \cite{Elvira2021} gives the exact recovery guarantee over the completely monotone function(CMF) dictionatry. 
 However the complex sinusoids dictionary is not included in CMF class, therefore the exact recovery theory in \cite{Elvira2021} cannot be employed. 

\subsection{Contributions}

In this paper, we first propose and analyze the exact recovery guarantee of the Sliding-OMP algorithm (Algorithm~\ref{alg-SOMP}), which is a variant of the continuous OMP algorithm, then we establish the guarantee in the incomplete-measurement setting and discuss its implementation via grid-discretization. Our analysis also sheds light on the continuous OMP algorithm for the super-resolution problem.  
\subsubsection{The Sliding-OMP Algorithm}
The  Sliding-OMP algorithm is illustrated in Algorithm~\ref{alg-SOMP}, where we denote $\bm P(\bomega)$ as the orthogonal projection operator into the column space of $\bm F(\bomega)$. 

\vspace{0.5cm}
\begin{algorithm}[H]
	\textbf{Input:} $\y,$ dictionary $\w (\cdot)$, non-negative preconditioner $\{\sigma_\ell\}_{-n}^n$, sliding stepsize $\eta$, stopping threshold $\gamma$, sliding iteration number $T.$ \\
	\textbf{Initialization:}  $\y, \w (\cdot) \leftarrow \text{PreCondition} \big(\y, \w(\cdot  ),\{\sigma_\ell\} \big)$ . Setting $\mr_0 = \y , t= 0.$  \\
	\While{$\max_{\tau \in [0,1)} \lvert \w (\tau)^* \mr_{t}   \rvert > \gamma  $  }
	{
		 $ \hat\omega_{t+1} \leftarrow  \text{argmax}_{\tau \in [0,1)} \lvert \w (\tau)^* \mr_{t}\rvert $.\\
		 $\hat\bomega_{\leq {t+1}}\leftarrow (\bomega_{\leq t}, \hat\omega_{t+1} )$ \\ 
		$ \bomega_{\leq t+1}\leftarrow$ Sliding($\hat \bomega_{\leq t+1}, \y , \eta, T$)    \\
		 $\mr_{t+1} \leftarrow \big (\mI - {\mP}( \bomega_{\leq {t+1}}) \big ) \mr_{t} $.\\
		$t+1 \leftarrow t. $		
	}
	\textbf{return} $\bomega_{\leq t}.$
	\caption{Sliding Orthogonal Matching Pursuit}
	\label{alg-SOMP}
\end{algorithm}
\vspace{0.5cm}

Compared with continuous OMP,  our algorithm includes an additional pre-conditioning step and a local optimization procedure at each iteration.

\paragraph{Pre-conditioning Operation:} 
The pre-conditioning operation is a standard technique in Fourier edge-detection literatures \cite{Tadmor2007,Gelb1999,Cochran2013} that can help enforce the concentration phenomenon. 
There are numerous choices of possible preconditioners $\bm \sigma$ \cite{Tadmor2007}.  
In our work, we specify $\bm\sigma $ in Algorithm~\ref{alg-precondition operation}  as  \begin{align}\label{eq-particular-preconditioner}
	\sigma_\ell = \dfrac{1}{\lfloor n/2\rfloor} \sum_{k = \max (\ell - \lfloor n/2\rfloor, - \lfloor n/2\rfloor )}^{\min (\ell + \lfloor n/2\rfloor, \lfloor n/2\rfloor)} (1- \lvert \dfrac{k}{\lfloor n/2\rfloor}\rvert )  \big(1 - \lvert \dfrac{\ell}{\lfloor n/2\rfloor}- \dfrac{k}{\lfloor n/2\rfloor}\rvert \big),  \quad -n\leq  \ell \leq n.
\end{align}
In section~\ref{sec-improved-kernel}, we will explain the current selection on $\bm \sigma$, discuss the relationship between our work and Fourier edge detection, and explore other possible preconditioners.

\begin{algorithm}[H]
	\textbf{Input:}  $\bm y, \w(\cdot),\text{non-negative } \bm \sigma:= \{\sigma_\ell\} $\\
	$\y \leftarrow  \sqrt{\bm\sigma} \odot \y, \w(\cdot) \leftarrow \sqrt{\bm\sigma} \odot \w(\cdot). $ \\ 
	\textbf{return} $\y, \w(\cdot).$
	\caption{PreCondition}
	\label{alg-precondition operation}
\end{algorithm}

\paragraph{Sliding Operation:} 

In continuous OMP, the frequency $\omega_t$ obtained in the $t$-th round will keep unchanged in subsequent iterations. 
For Sliding-OMP, instead, we try to improve all previous found frequencies $\bomega_{\leq t}:= (\omega_1,\dots,\omega_t)$   by adding the sliding operation(Algorithm~\ref{alg-sliding operation}) at the end of $t$-th round. 
Similar refining operations were also developed in the Sliding-Frank-Wolfe(SFW) algorithm\cite{Denoyelle2020}. We note the following distinctions between their algorithm and ours:
Firstly, the SFW is proposed to solve the BLASSO problem, a provable optimization program whose solution is guaranteed to be the true frequency, and the authors mainly focus on the convergence of SFW to the BLASSO solution.  
Secondly, while the SFW has been proven to stop after a finite number of iterations, the $s$-step stopping guarantee in \cite{Denoyelle2020} is only illustrated empirically for specific circumstances, whereas our approach leads to an exact $s$-step stopping guarantee.
Finally, the sliding loss employed in the two algorithms differs: the iteration in Algorithm~\ref{alg-sliding operation} is equivalent to minimizing the loss \begin{equation}\label{eq-SOMP-slidingloss}
	\mL_t( \bomega): =  \dfrac{1}{2} \lVert \sum_{i=1}^t a_i\w(\omega_i) - \y \rVert_2^2 
\end{equation}  while the loss $\tilde{\mL}_t$ employed in SFW, has an additional sparsity-induced penalty: \begin{equation}\label{eq-SFW-slidingloss}
	\tilde{\mL}_t(\bomega):= \dfrac{1}{2} \lVert \sum_{i=1}^t a_i\w(\omega_i) - \y \rVert_2^2  + \lambda \lVert \bm a \rVert_1.  
\end{equation}

No convergence guarantee of minimizing  \eqref{eq-SFW-slidingloss}  is provided due to the difficulty of its non-convexity, while we develop a theoretical guarantee for Algorithm~\ref{alg-sliding operation} even when the loss \eqref{eq-SOMP-slidingloss} is also non-convex.

\begin{algorithm}[H]
	\textbf{Input:} $\bomega^{0},$ $\bm y$,  stepsize $\eta$, maximal iteration number $T$.    \\
	\textbf{Initialization: } $k = 0, \bm{w} =   \text{argmin}_{\bm a\in \mC^t} \lVert \sum_{i=1}^t a_i\w(\omega_i^0) - \y \rVert_2^2     $ \\
	\While{ $k\leq T$ }
	{
		$\bm{g}_k\leftarrow -  \dfrac{1}{\lvert \bm w \rvert^2}   \odot \nabla_{\bomega} \y ^*  P(\bomega^{k}) \y   $\\
		$\bomega^{k+1} \leftarrow \bomega^k - \eta   \bm{g}_k $ \\

		$k \leftarrow k+1$
	}
	\textbf{return} $\bomega^{k}.$
	\caption{Sliding}
	\label{alg-sliding operation}
\end{algorithm}

\subsubsection{Theoretical Results}
Now we present our main theoretical results of the recovery guarantee of the Sliding-OMP Algorithm~\ref{alg-SOMP}.

\begin{theorem}\label{thm-sliding-guarantee-intro} Suppose $n\Delta > c$ for some absolute constant $c$ and denote  $\bomega_{\leq t} = (\omega_1,\dots,\omega_t) $ the output of $t$-th iteration in Algorithm~\ref{alg-SOMP}. There exists a permutation $T(\cdot)$ over $[s]$ so that  for $\mT_t: = \{ T(1),\dots, T(t) \}$,  the weighted estimation error satisfies \begin{align*}
		\max_{i\leq t} \lvert x_{T(i)} (\tau_{T(i)} - \omega_i) \rvert   \lesssim  \dfrac{1}{n (n\Delta)^4} \lVert \x(\mT_t^c) \rVert_\infty, 
	\end{align*}
	where $\x(S):= (x_{s_1},\dots,x_{s_m}) $ for any $S = \{s_1,\dots,s_m\}\subset [s]$ and 
 we denote $ \x([s]^c) = 0$ for simplicity.  
\end{theorem}

As a corollary, we have
 \begin{corollary}\label{thm-sliding-OMP-guarantee} As long as $ n\Delta>c  $, Sliding-OMP will  recover $\{\tau_{i}\}_{i=1}^s$ exactly after $s$-iterations.	
\end{corollary}

We will prove the Theorem~\ref{thm-sliding-OMP-guarantee} in section~\ref{sec-proof-main-results}, Figure~\ref{fig-proof-sketch} shows the relation between the main results and several key intermediate results.

\vspace{0.5cm}

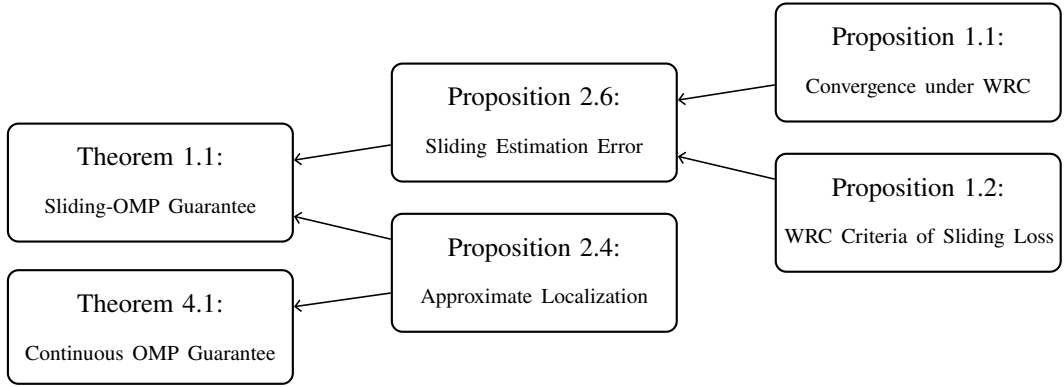
\begin{figure}[H]
\qquad \qquad     \begin{tikzpicture}[
node distance = 4mm and 8mm,
  start chain = going below,
   arr/.style = {{Straight Barb[scale=0.8]}-, rounded corners=1ex, semithick},
     N/.style = {draw, rounded corners, thick, fill=#1,
                 text width=10em, align=center, inner ysep=2ex},
   N/.default = white,
every edge/.style = {draw, arr}
                        ]   
\node (n11) [N=white!30] {Theorem~\ref{thm-sliding-OMP-guarantee}:\\ \footnotesize Sliding-OMP Guarantee}; 
\node (n12)  [N=white!30, below=10pt of n11]  { Theorem~\ref{thm-omp-naive}:\\{\footnotesize Continuous OMP Guarantee} };  
\node (n21) [N=white!30, right=90pt of n11, above] {Proposition~\ref{thm-sliding-est-error}:\\\footnotesize{Sliding Estimation Error} } ;
\node (n22) [N=white!30, below=of n21]  {Proposition~\ref{prop-improved-localization}:\\ \footnotesize{Approximate Localization} };
\node (n31) [N=white!30, right=90pt of n21,above] {Proposition~\ref{prop-regularity}:\\
\footnotesize\text{Convergence under WRC} }  ;
\node (n32) [N=white!30, below=of n31] {Proposition~\ref{prop-sliding-informal}:\\
\text{\footnotesize WRC Criteria of Sliding Loss}};

\draw[arr] (n11)  -- (n21);
\draw[arr] (n11) -- (n22);
\draw[arr] (n12) -- (n22);
\draw[arr] (n21) -- (n32);
\draw[arr] (n21) -- (n31);
    \end{tikzpicture}
    
\caption{The proof sketch of main results} \label{fig-proof-sketch}   
\end{figure}

 Our analysis consists of two key components: 
\paragraph{Approximate Localization:} 
The first contribution of our work is a framework for analyzing the continuous OMP algorithm.  In a discrete setting, the analysis of OMP makes use of low correlation across bases in the dictionary. While the correlation between two bases in a continuous dictionary can be arbitrary close to $1$, violating the requirements in \cite{Tropp2004, Cai2011}.
Indeed, as shown in \cite{Elvira2021}, for the dictionary in the super-resolution problem, one cannot expect that OMP algorithm chooses the true basis of the dictionary in any iteration, resulting in a basis mismatch in every iteration of continuous OMP, which will spread along subsequent iterations and may ruin the performance of continuous OMP.  

Instead of pursuing exact recovery guarantee, we establish the estimation error guarantee of the continuous OMP in  Proposition~\ref{prop-improved-localization}: At $t+1$-th step, for the weighted estimation error $$\varepsilon_{\x,t} := \max_{i\leq t} \lvert x_{T(i)} (\omega_i - \tau_{T_{(i)}}) \rvert $$ and the correlation maximizer $\hat{\omega}_{t+1} =  \text{argmax}_{\tau\in [0,1)} \lvert \w(\tau)^*\bm r_t \rvert$ we have $$\lvert  \hat \omega_{t+1}-\tau_{T(t+1)}\rvert \leq \dfrac{1}{n(n\Delta)^2}  \sqrt{\dfrac{n \varepsilon_{\x,t}  }{\lVert \x (\mT_{t}^c)\rVert_\infty  } }.$$

While such estimation error result helps establish the theoretical guarantee of continuous OMP(Algorithm~\ref{alg-COMP})  in $n\Delta \gtrsim \dynx$ regime, it is not sufficient to provide the guarantee of continuous OMP in $n\Delta \asymp c$ regime: 
Denoting \begin{align*}
	{\varepsilon}_t := \max_{i\leq t}  \lvert \omega_i - \tau_i\rvert  , \quad \dynx = \dfrac{\max_i \lvert x_i \rvert}{\min_i \lvert x_i \rvert}.
\end{align*}
In the worst case, the ratio $\frac{n\varepsilon_{\x,t}}{\lVert \x(\mT_t^c)\rVert_\infty}$ turns to $\dynx\cdot n\varepsilon_{t},$ and applying the result iteratively for every $t$ leads to $\varepsilon_{t} \leq \dfrac{1}{n(n\Delta)^2}\cdot\big(\frac{\sqrt{\dynx}}{n\Delta} \big)^t$, which will tend to infinity when $\frac{\dynx}{n\Delta} \to \infty.$  
Indeed, we show in section~\ref{sec-byproduct-OMP} that the dependency of $n\Delta$ on $ \dynx$ is necessary.

\paragraph{Generalized Basin of Attraction:}

The error results of continuous OMP in the $n\Delta\asymp c$ regime motivates us to add the sliding operation in the algorithm and derive improved estimation error results.

However, the analysis of sliding operation is non-trivial from two perspectives:
Firstly, the loss $\mL _t$ is nonconvex, which makes it difficult to establish the convergence guarantee.
Secondly, the sliding loss is biased when $t < s$, as illustrated in Figure~\ref{fig-sliding-analysis}(a). True frequencies are not the minimizer of $\mL_t$ in general, so traditional optimization analyses would fail in this case.
The above two challenges also prevent us from getting the estimation error guarantee of sliding loss from the established basin of attraction results when $t = s$ in \cite{Eftekhari2015, Traonmilin2020}.
The development of new approaches to studying the estimation of sliding error is the second contribution of our work.

To address these issues, we propose the \textit{weak regularity condition} in section \ref{sec-weak-regularity-condition} as a new criterion of convergence under \textit{weighted $\ell_\infty$ norm} ($\lVert \bm v \rVert_{\bm w,\infty}:= \max_i \lvert v_iw_i\rvert, \forall \bm v,\bm w \in \R^d $). 
 The weak regularity condition is more flexible than the classical one in our setting.
This new criterion allows us to figure out a general region, the \textit{generalized basin of attraction}   $\mB_t$ of $\mL_t$ for each $t$, in which the weighted estimation error $\lVert \bomega_{\leq t} - \btau (\mT_t)  \rVert_{\bm x(\mT_t),\infty} $ is guaranteed to contract under gradient descent iterations.
More precisely, we show that $\mB_t$ has the shape of a concentric weighted $\ell_\infty$-circle: \begin{align*}
		\mathcal B_t: = \{ \bomega \in \R^{t} : \dfrac{\lVert \x({\mT_{t}^c}) \rVert_\infty}{n(n\Delta)^4}  \lesssim \lVert \bomega - \btau(\mT_t)\rVert_{\x(\mT_t) , \infty} \lesssim \dfrac{\lVert \x(\mT_{t-1}^c)\rVert_\infty}{n}  \}. 
\end{align*}

When $t<s$, $\mathcal{B}_t$ helps us establish the estimation error guarantee of sliding operation: $\varepsilon_{\x,t} = O( \dfrac{\lVert \x(\mT_t^c)\rVert_\infty}{n(n\Delta)^4})$, which improves the previous approximate localization result. Indeed, at $t+1$-th step, we can combine such improved estimation error for $\bomega_{t}$ and the approximate localization result for $\hat{\omega}_{t+1}$ to show that $\hat{\bomega}_{t+1} \in \mathcal{B}_{t+1}$, so the Sliding-OMP can successively decrease the estimation error despite the bias of sliding loss. We illustrate the sliding effect in Figure~\ref{fig-sliding-analysis}(b) for two-dimensional projection of $\mathcal{B}_t.$

When $t = s$, $\mB_t$ degenerates to the classical basin of attraction, and the corresponding estimation error guarantee would be a local convergence result to $\btau$.
Furthermore, the convergence theory is based on a weighted $\ell_\infty$ geometry, which is better suited to super-resolution problems, so the generalized basin of attraction determined by our approach is broader than that in previous works \cite{Eftekhari2015, Traonmilin2020} under $\ell_2$ geometry. 
 \begin{figure}[H]
 \centering
\begin{subfigure}[b]{0.4\textwidth}\label{fig-sliding-bias}
 \centering
 	\includegraphics[width=0.8\textwidth]{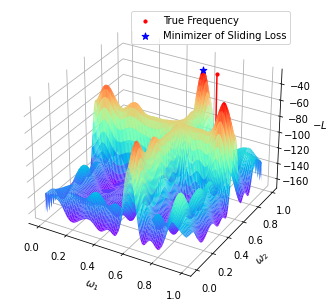}
 \caption{The bias of sliding loss landscape. Here we draw the negative sliding loss at the second step of sliding-OMP for a $4$-spike signal $\bm y = 3\w(0.6)-2\w(0.9)+1.5\w(0.1)-\w(0.3) $ }
 \end{subfigure}
 \qquad 
  \begin{subfigure}[b]{0.4\textwidth}
 \centering
 \hspace{0cm}\vspace{0.1cm}	\includegraphics[width=\textwidth]{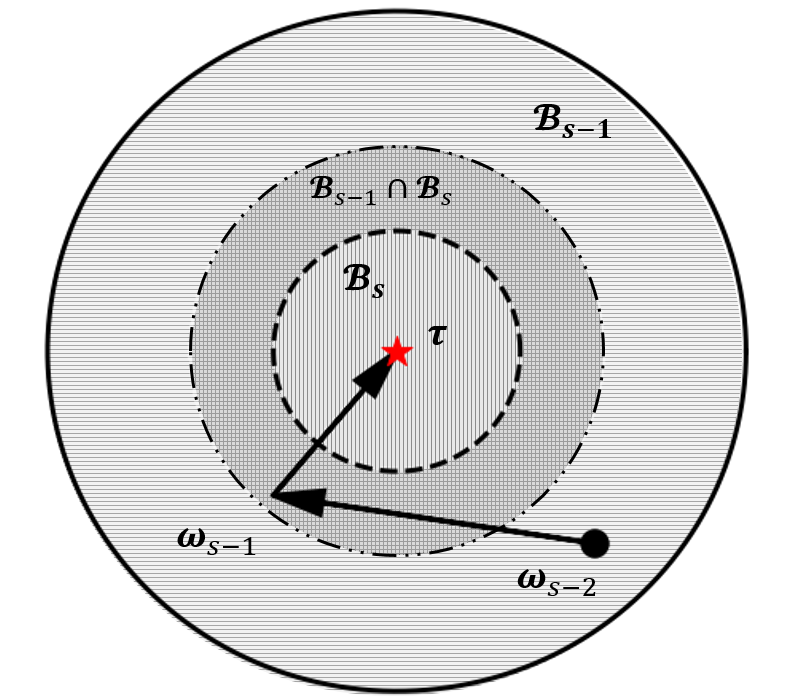}
 \caption{ {The illustration of sliding iteration  and the two-dimensional projection of $\mathcal{B}_t$ and $\bomega_t$ in last two iterations}}
 \end{subfigure}
 \caption{Illustration of Sliding Analysis}\label{fig-sliding-analysis}
 \end{figure}

\subsection{Weak Regularity Condition}\label{sec-weak-regularity-condition}

In this section, we introduce the \textit{weak regularity criterion}  in a general scenario, which is a key component when analyzing sliding operations.

Given a vector $\bm w \in \R_{++}^d$,  we denote $\lVert \bm v \rVert_{\bm w,\infty}:= \max_i \lvert v_iw_i\rvert, \forall \bm v \in \R^d$. We are interested in  the behaviour of the weighted distance $\lVert  \bomega^{k} - \btau  \rVert_{\bm w, \infty} $ for some $\btau \in \R^d$ when $\bomega^{k}$ is generated by the  following iteration formula \begin{align*}
	\bomega^{k+1} = \bomega^{k} - \bm g( \bomega^{k}),\quad k\in \mathbb{N},
\end{align*}
where $\bm g(\cdot)$ is some specified direction function.

In our setting, $\bm g$ is also a variant of the gradient of $\mL$, while $\btau$ is not the minimizer of $\mL$. Moreover, $\mL$ is not convex in general.

Furthermore,  the convergence under $\ell_{\bm w,\infty}$ (short for weighted $\ell _\infty$) geometry is different from that in  $\ell_2$ space. 
All of the aforementioned challenges force us to develop new convergence criteria. To make this quantitative, We introduce the following definition.

\begin{definition}[$\ell^\infty_\mw$-Weak Regularity Condition]  \label{def-weak-regularity} Given $\bm{\tau} \in \R^d$ and $\bm{w} \in \R^d_{++}$ ,	we say a direction function $\mg(\cdot)$ satisfies the  {WRC}$_{\mw}(\bm{\lambda},\alpha,\beta )$\big($\ell^\infty_{\mw}$-weak regularity condition\big) for some $ \bm{\lambda}\in \R^d_{++}$ and $\alpha,\beta>0$   over a set  $B\subset \R^d$, if  
    $\forall \bomega \in B$, we have 
\begin{align}\label{eq-regularity-contraction-requirement} 
	2 g(\bomega)_i( \omega_i - \tau_i ) &\geq   g_i(\bomega)^2 + \lambda_i  ( \omega_i - \tau_i)^2, \quad \forall i \in \mathcal{I}_{\alpha},\\
	 \lvert w_i g_i(\bomega)\rvert &\leq \beta  \lVert \bomega - \btau \rVert_{\bm{w},\infty}\quad \forall i \in \mathcal{I}_\alpha^c, \label{eq-regularity-small-requirement}
\end{align}
where  $\mathcal{I}_\alpha(\bomega): =\{ i\in[d]: \lvert w_i(\omega_i - \tau_i) \rvert \geq \alpha \lVert \bomega - \btau\rVert_{\bm{w},\infty}  \}.  $
\end{definition}

In the definition, we divide the coordinates into ``boundary compoments'' $\mathcal I_\alpha$ and ``interior componments'' $\mathcal I_\alpha^c$  according to their distance to the boundary.
For the boundary component $i\in \mathcal I_\alpha$ , \eqref{eq-regularity-contraction-requirement}  ensures a contraction along $i$-th coordinate after iteration; 
for the interior component $i\in \mathcal I_\alpha^c$, \eqref{eq-regularity-small-requirement} ensures $i$-th coordinate will not move too far and remain within the contracted boundary. 

The number $\alpha$ serves as a threshold for dividing coordinates into boundary components or interior components.
The number $\beta$ indicates how far interior components move after iteration.

The weak regularity criterion leads to linear convergence of $\lVert \bomega^k - \bm g\rVert_{\bm w,\infty}$, according to the following proposition.
\begin{property} \label{prop-regularity}
	If a direction function $\mg(\cdot)$ satisfies $\text{WRC}_{\mw}( \bm{\lambda}, \alpha , \beta )$  over a set $B$ for some $\bm\lambda  , \alpha,\beta$, then  the sequence $\{\bomega^k\}_{k\geq 0}$ given by the iteration formula \begin{equation}
		\bomega^{k+1} = \bomega^k - \mg(\bomega^k)
	\end{equation}
	satisfies \begin{align*}
		 \lVert \bomega^{k+1} - \btau \rVert_{\bm{w},\infty} \leq \max\{\alpha+\beta , \sqrt{1 -\min_i \lvert  \lambda_i \rvert } \} \lVert \bomega^k-\btau \rVert_{\bm{w},\infty}
	\end{align*}
	as long as $\{\bomega^k\}_{k \geq 0}\subset  B.$
\end{property}

\begin{proof}
At $k$-th iteration, for $i \in \mathcal I_\alpha,$ we have \begin{align*}
	\lvert  \omega_i^{k+1} - \tau_i\rvert^2 & = \lvert  \omega_i^{k} -  g_i(\bomega^k) - \tau_i\rvert^2 \\
	& =\lvert \omega_i^k-\tau_i\rvert^2 -2 g_i(\bomega^k)(\omega_i^k-\tau_i) +  g_i(\bomega^k)^2\\
	& \leq \lvert \omega_i^k-\tau_i\rvert^2 -  g_i(\bomega^k)^2- \lambda_i(\omega_i^k-\tau_i)^2+  g_i(\bomega^k)^2\\
	&\leq (1- \lambda_i) \lvert \omega_i^k - \tau_i\rvert^2.
\end{align*}
Multiplying $w_i^2$ in both sides lead to \begin{align*}
	\lvert w_i(\omega_i^{k+1}-\tau_i)\rvert \leq \sqrt{1 - \min_i   \lambda_i } \cdot \lvert w_i(\omega_i^{k}-\tau_i)\rvert.
\end{align*}
For $i\in \mathcal I_{\alpha}^c,$ we have \begin{align*}
	\lvert w_i( \omega_i^{k+1} - \tau_i)\rvert & \leq  \lvert w_i( \omega_i^{k}  - \tau_i)\rvert + \lvert w_i   g_i(\bomega^k)\rvert \\
	&\leq  \lvert w_i( \omega_i^{k}  - \tau_i)\rvert +\beta \lVert \bomega^k - \btau\rVert_{\bm w, \infty}\\
	&\leq (\alpha+\beta) \lVert \bomega^k - \btau\rVert_{\bm w, \infty}.
\end{align*}
Then the claim holds by combining the bounds for $i\in \mathcal I_\alpha$ and $i\in \mathcal I_\alpha^c$ together.  
\end{proof}

We provide the following conclusion informally to illustrate the role of the weak regularity criteria in our analysis of Sliding-OMP: \begin{property}\label{prop-sliding-informal}
	 There exists some absolute constant $c$ so that at the $t$-th step of Algorithm~\ref{alg-SOMP}, the weighted gradient direction  $\bm  g:= \dfrac{1}{n^2}\cdot  \dfrac{1}{ \lvert \bm{w} \rvert^2} \odot \nabla \mL_t $  satisfies the $\text{WRC}_{\x(\mT_t)} ( \bm \lambda,\alpha,\beta) $ with $\min_i \lvert \lambda_i \rvert >c, \alpha +\beta < 1-c$ over the region \begin{align*}
	 	\mathcal B_t: = \{ \bomega \in \R^{t} : \dfrac{\lVert \x(\mT_t^c)\rVert_\infty}{n(n\Delta)^4}  \lesssim \lVert \bomega - \btau_{\leq t}\rVert_{\x(\mT_t), \infty} \lesssim {\dfrac{\lVert \x(\mT_{t-1}^c)\rVert_\infty}{n} }  \}.
	 \end{align*}  
	 Moreover, the initialization given in  Algorithm~\ref{alg-SOMP} is guaranteed to lie in $\mB_t.$ 
\end{property}
Indeed, Proposition~\ref{prop-sliding-informal}  implies that the Sliding Algorithm~\ref{alg-sliding operation} reduces the estimation error into $O( \dfrac{\lVert \x_{>t}\rVert_\infty}{n(n\Delta)^4})$ after logarithmic iterations, along with  an estimation error guarantee of sliding operation.  We will provide a formal and more precise statement of Proposition~\ref{prop-sliding-informal} in section~\ref{sec-proof-sliding-lanscape}.

\subsection{Additional Results}
\subsubsection{Super-Resolution from Incomplete Measurements} 
 Our analysis can be easily extended from full-measurement analysis to incomplete measurement settings. 
We investigate the following symmetric Bernoulli$(p)$-observation model of the signal $\bm y$ and transfer full-measurement analyses to the incomplete measurement scenario: Every index $y_i$ with $i \in [n]$ is observed with probability $p$, and $y_{-i}$ is observed if and only if $y_i$ is observed.
The extended result states that $p \gtrsim s^2\dfrac{\log n}{n}$ is sufficient to ensure the similar exact recovery result as the full-measurement setting.

In practice, the complexity of the gradient descent operation in sliding operation is proportional to the measurement number, so subsampling can dramatically reduce the computational complexity.
Our discussion in section~\ref{sec-discuss} shows that only $O\big(sn\log n + \text{poly}(s,\log n) \big)$ FLOPs are needed in our algorithm.

We summarize the result formally as the following theorem: 
\begin{theorem}\label{thm-subsampling-result}
	When $\y$ is drawn from the symmetric $p$-Bernoulli observation model, with $np \gtrsim s^2\log(n)$, the claims in Theorem~1.1 and Corollary~1.1 still hold with probability at least $1-1/n^2$. 
\end{theorem}
We leave the proof of Theorem~1.2 in section~\ref{sec-analysis-incomplete-sample}.

\subsubsection{Results for Continuous OMP}\label{sec-byproduct-OMP-intro}
As a corollary of our analysis in approximate localization,  we provide the following estimation error guarantee of continuous OMP without sliding.
\begin{theorem}
	Suppose $n\Delta \geq \zeta \dynx$ for $\zeta \geq C$ with some absolute constant $C$, then there exists two absolute constants $c_1<c_2$ so that the Algorithm~\ref{alg-COMP} with the preconditioner  and stopping threshold $c_1  \min_i\lvert x_i\rvert <\gamma <c_2 \min_i\lvert x_i\rvert $ will stop exactly after finding $s$ frequencies $\omega_1,\dots,\omega_s$. Moreover, there exists unique $T(1),\dots, T(t) \in [s]$ so that $\tau_{T(i)} =\text{argmin}_{\tau \in \mT} \lvert \tau - \omega_i\rvert$ 
	and   \begin{align*}
		\lvert \omega_{i}-\tau_{T(i)}\rvert \lesssim \dfrac{1}{n\cdot \zeta^2\dynx^2 }. 	\end{align*}

\end{theorem}

This result gives theoretical guarantee to the algorithm proposed in \cite{Benard2022},  and it is comparable to another greedy-based initialization algorithm proposed in \cite{Eftekhari2015} (note that the condition $\liminf \frac{log n}{n\Delta}  > c$ implies our separation gap condition).  
Finally, the separation condition $n\Delta > c \dynx$ for continuous OMP has an additional term $\dynx$, comparing to the state-of-the-art separation gap condition $n\Delta > c$.
The dependency on $\dynx $ is fundamental:  for any constant $c>0$, we construct an instance in Theorem~\ref{thm-impossible-OMP}  with large enough $\dynx$  so that $n\Delta >c$ while the continuous OMP algorithm fails to recover true frequencies. 

We will summarize all aforementioned results for continuous OMP in section~\ref{sec-byproduct-OMP}.

\subsection{Organization of the Paper}
The paper is organized as follows. In section~\ref{sec-proof-main-results} we prove our main results . In section~\ref{sec-analysis-incomplete-sample} we analyze the incomplete measurement setting. 
In Section~\ref{sec-byproduct-OMP}, we show results for continuous OMP. In section~\ref{sec-experiment} we provide an efficient implementation of the algorithm and numerical experiments. In section~\ref{sec-discuss} we discuss other related works and some future directions.

\section{Proof of Main Results}\label{sec-proof-main-results}
\paragraph{Notation:} Throughout the paper, we use $c, C$ to denote positive absolute constants, whose value may change from one line to the next. We denote $a = O(b)$ or $a\lesssim b$( $a = \Omega(b)$ or $a \gtrsim b$) if there exists some  $c>0$ so that $\lvert a \rvert \leq c \cdot \lvert b \rvert$( $\lvert a \rvert \geq c\cdot \lvert b \rvert $). For any $\tau_1,\tau_2\in [0,1),$ we set $\tau_1\pm \tau_2$ as $\tau_1\pm \tau_2 \mod  1$ . For $\tau \in [0,1)$, we denote $\mS(\tau): = [\tau-\frac{1}{2n+4},\tau + \frac{1}{2n+4} ].$  

\subsection{Warm-Up: Analysis in the first step for un-preconditioned setting}\label{sec-warmup}
We briefly illustrate the ideas and motivate the necessity of pre-conditioning by providing the first-step analysis of \textbf{un-preconditioned} OMP algorithm in this section. 
We would note that  the first steps of Algorithm~\ref{alg-COMP} and Algorithm~\ref{alg-SOMP} are coincident because the sliding loss $\mL_1(\bomega)$  is equivalent to the negative correlation $-\lvert \w(\omega_1)^* \y \rvert$. 

\begin{property}\label{prop-first-step-unprecondition}
	 As long as $n\Delta > C\log s,$ we have the first step output $\omega_1$ of Algorithm~\ref{alg-COMP} and Algorithm~\ref{alg-SOMP} without preconditioning satisfies $\omega_1\in \mS(\tau_{T(1)})$ for some $\tau_{T(1)} \in \mT$ and $\lvert x_{T(1)}\rvert \geq (1-  \frac{c}{n\Delta }\log s ) \lVert \x \rVert_\infty$.
\end{property}

\begin{proof}
We have \begin{align*}
	\w(\tau)^*\mr_{0} = \w(\tau)^*\y  = \sum_{i=1}^s x_i (\sum_{\ell = -n}^n e^{2\pi j \ell (\tau- \tau_i)})  = \sum_{i=1}^{s} x_i D_n(\tau-\tau_i).
\end{align*} 
Thus the value of the inner product at $\tau$ can be interpreted as the value of suppression of Dirichlet kernels at $\tau. $  
Due to the tail effect of the Dirichlet kernel, its value near every true frquency $\tau_i$ will be perturbed by a noise signal $ \sum_{k\neq i} x_k D_n(\tau_i- \tau_k) $. 
We can further bound such perturbation magnitude by the decay-rate bound for $D_n:$ \begin{align*}
	\lvert D_n(\tau) \rvert =\lvert \dfrac{\sin (2n\pi \tau )}{2n \sin (\pi \tau)}\rvert \leq \dfrac{1}{2n\lvert \tau \rvert}, \quad \tau \in [-\dfrac{1}{2},\dfrac{1}{2}].
\end{align*}  
Thus \begin{align*}
	\lvert \sum_{k\neq i} x_k D_n(\tau_i -\tau_k) \rvert &\leq \dfrac{\lVert \x_{-i}\rVert_{\infty}}{2n} \sum_{k\neq i}  \dfrac{1}{\lvert \tau_i - \tau_k\rvert}\leq \dfrac{\lVert \x_{-i}\rVert_\infty}{2n}\sum_{k = 1}^{\lceil s/2\rceil } \dfrac{1}{\Delta k}   \leq c'\dfrac{\lVert \x_{-i}\rVert_\infty}{n\Delta} \cdot \log s. 
\end{align*}
As a result, for $i_1 = \text{argmax}_{i}\lvert x_i\rvert,$ we have \begin{align*}
	\lvert\w (\tau_{i_1})^* \y \rvert \geq \lvert x_{i_1} \rvert \big(1-\dfrac{c'\log s}{n\Delta} \big). 
\end{align*}
On the other hand, for every $\tau \notin \mS(\mT),$ and $i(\tau): = \text{argmin}_{i\in [s]} \lvert \tau_i- \tau\rvert,$ then \begin{align*}
	\lvert \w(\tau)^* \y\rvert &\leq \lvert x_{i(\tau)} \cdot D_n(\tau - \tau_{i(\tau)}) \rvert + \sum_{k\neq i(\tau)} \lvert x_k D_n(\tau - \tau_{k})\rvert \\
	&\leq \lvert x_{i_1}\rvert \big(\lvert D_n(\tau - \tau_{i(\tau)}) \rvert + \sum_{k\neq i(\tau)} \dfrac{1}{2n \lvert \tau-\tau_k\rvert}  \big)\\
	&\leq \lvert x_{i_1}\rvert \big(\lvert D_n(\tau - \tau_{i(\tau)}) \rvert +  \dfrac{2c'\log s}{n\Delta} \big)
\end{align*}
where we have used $\lvert \tau-\tau_k \rvert \geq  \frac{\Delta}{2}$ for $k\neq i(\tau)$ in last inequality. Now by $\tau \notin \mS(\tau_{i(\tau)}),$ we have $\lvert D_n(\tau-\tau_{i(\tau)})\rvert < 0.5,$ that gives the upper bound on $\lvert\w (\tau)^*\y \rvert$ when $\tau \notin \mS(\mT)$. Combining such upper bound and the lower bound on $\lvert \w(\tau_{i_1})^*\y \rvert,$ we now derive a sufficient condition for  approximate localization at the first step: \begin{align*}
	\dfrac{3c'\log s}{n\Delta}< 0.5 \implies  \omega_1 \in \mS(\mT). 
	\end{align*}  
	We can further argue that $\omega_1\in \mS(\tau_1)$ for some $\tau_1$ with sufficiently large magnitude via similar argument.  
Suppose  $\mS(\tau_{T(1)})$ , then we have  
\begin{align*} 
	\lvert \w(\omega_1)^*\y \rvert > \lvert \w(\tau_{i_1})^*\y \rvert &\implies  \lvert x_{T(1)} \rvert  + \lvert x_{i_1} \rvert \dfrac{c\log s}{n\Delta} > \lvert x_{i_1} \rvert(1- \dfrac{c\log s}{ n\Delta}),
\end{align*}
which leads to \begin{align*}
	\lvert x_{T(1)} \rvert \geq \big( 1-\dfrac{2c' \log s}{n\Delta } \big) \lvert x_{i_1} \rvert . \end{align*}
That finishes the proof with $C = 6c',c = 2c'$ .
\end{proof}

\subsection{Enforcing the Concentration via Preconditioning}\label{sec-improved-kernel}
 We find that the analysis in section~\ref{sec-warmup} leads to a  $\log s$ term due to the summation of the tail value of $s$ Dirichlet kernels, and the  $\dfrac{1}{nt}$ decay rate of $D_n.$ And the preconditioning procedure turns $\w(\tau)^*\mr_0$  to  $
	 \sum_{i=1}^s x_i K_n(\tau-\tau_i)$ 
with more general concentration kernels $K_n$. Suppose $\lvert K_n(\tau) \rvert\leq \frac{1}{(n\tau )^{1+\alpha}}$ for some $\alpha >0$, then the $\frac{\log s}{n\Delta}$ term will turn to $\frac{c_\alpha}{(n\Delta)^{1+\alpha}}$ for some constant $c_\alpha$ depending only on $\alpha.$ 

For a non-negative vector  $\sigma\in \R^{2n+1},$ denote $$K_{n,\sigma}(\tau ): = \sum_{\ell = -n}^n \sigma_\ell e^{2\pi j \ell \tau  } ,  $$ 
 if we take the pre-conditioning operation  $\w(\tau)\leftarrow \sqrt \sigma\odot \w(\tau)$  and $\y \leftarrow \sqrt{\sigma}\odot \y,$ then we have \begin{align}\label{eq-squared-Fejer}
 	\w(\tau)^*  \y  = \sum_{i=1}^{s} x_i K_{n,\sigma}(\tau-\tau_i).
 \end{align}
thus all analysis involving $D_n$ can be replaced by $K_{n,\sigma. } $ Now it remains to select the preconditioner $\{\sigma_\ell\}$. Such a topic has been well-studied in Fourier edge detection area  \cite{Gelb1999, Tadmor2007}.  Among large class of possible conditioners, our choice in \eqref{eq-particular-preconditioner}  leads to the squared Fej\'er kernel $$K_{n,\sigma}(t) =  \bigg[ \dfrac{\sin\big((\lfloor \frac{n}{2}\rfloor +1)\pi t\big) }{\big(\lfloor \frac{n}{2}\rfloor+1 \big) \sin(\pi t) } \bigg]^4, $$  a well-concentrated  kernel used to construct sharply-peaked trigonometric polynomials in \cite{Candes2014,Tang2013}.   For simplicity of notation, we assume W.L.O.G. that $n$ is even, thus $\lfloor n/2\rfloor = n/2.$ When $n$ is odd, all theoretical analysis and results still hold by replacing $n$ with $n-1$.   
 \begin{property}\label{prop-squared-fejer-kernel}
	For 
\begin{align*}
	K_{n,\sigma }(t) = \bigg[ \dfrac{\sin\big((\frac{n}{2}+1)\pi t\big) }{\big(\frac{n}{2}+1 \big) \sin(\pi t) } \bigg]^4.
\end{align*}
We have \begin{align*}
		 &\lvert	K_n (\tau) \rvert  \leq  \min\{ 0.7, \dfrac{1}{\big ((n+2)\tau \big )^4} \},\quad \forall \dfrac{1}{2n+4} \leq  \lvert  \tau \rvert \leq 0.5  \\
		   & \dfrac{\pi^2}{6} (n+2)^2 \tau^2  \geq \lvert 1 - K_n(\tau) \rvert \geq  (n+2)^2  \tau^2 \quad \forall  \lvert  \tau \rvert   \leq \frac{1}{2n+4}. \\
		  &\lvert K_n'(\tau) \rvert \leq \left\{\begin{matrix}
		  	\frac{\pi^2}{3} (n+2), & \lvert \tau \rvert \leq \frac{1}{2n+4} \\
		  	\dfrac{\pi^2}{(n+2)^3 \lvert \tau \rvert^4}  \} , &\frac{1}{2n+4} \leq \lvert \tau \rvert \leq \frac{1}{2}
		  \end{matrix}\right.  \\
		& \frac{\pi^2}{3} (n+2)^2\tau^2 \geq -K_n'(\tau) \tau  \geq 1.9(n+2)^2 \tau^2, \quad  \forall  \lvert \tau \rvert\leq\dfrac{1}{2n+4} \\
		&\lvert K_n{''}(\tau) \rvert \leq \left\{\begin{matrix}
		  	\frac{\pi^2}{3} (n+2)^2, & \lvert \tau \rvert \leq \frac{1}{2n+4} \\
		  	\dfrac{4\pi^4}{(n+2)^2t^4} , &\frac{1}{2n+4} \leq \lvert \tau \rvert \leq \frac{1}{2}
		  \end{matrix}\right.  \\
	\end{align*}
\end{property} 
\begin{remark}
\cite{Eftekhari2015} also incorporates the preconditioning technique into the design of their algorithm. They choose the \textit{discrete prolate spheroidal wave function}(DPSWF) as the preconditioner and  assert several asymptotic properties of DPSWF that similar to our Proposition~\ref{prop-squared-fejer-kernel}. However, they leave the verification of their assertion as an open problem while we provide a rigorous proof of the Proposition~\ref{prop-squared-fejer-kernel} in Appendix~\ref{appendix-kernel-inequality}.
\end{remark}
\begin{remark}
	In the remaining part of the context, we replace the $n+2$ term with the $n$ term in the upper and lower bounds of $K_n, K_n', K_n''$ for simplicity of notation. Such replacement only leads to a minor change by multiplying the constant in all upper bounds by a scalar that is close to $1$.
\end{remark}

With such $K_n$, we can get the similar approximate localization guarantee for the first step by an easy modification of arguments in section~\ref{sec-warmup}: \begin{property}~\label{prop-first-step-fejer}  We have the first step output $\omega_1$ of Algorithm~\ref{alg-COMP} and Algorithm~\ref{alg-SOMP} with preconditioner $\{\sigma_\ell\}$ specified in \eqref{eq-particular-preconditioner}  satisfies $\omega_1\in \mS(\tau_{T(1)})$ with $\tau_{T(1)}\in \mT$ and  $\lvert x_{T(1)}\rvert \geq (1 - \frac{c}{n^4\Delta^4} ) \lVert \x \rVert_\infty$ as long as  $n\Delta > C$.
\end{property} 

\subsection{Approximate Localization of Continuous OMP}\label{sec-approximate-localization}

\subsubsection{Approximate Localization Lemma}

\begin{lemma}\label{lem-error-formula}  Supposing $\varepsilon_{t}:=\lVert \bomega_{\leq t} - \btau(\mT_t)\rVert_\infty  < \dfrac{1}{2n+4}$. Then at time-step $t+1$, as long as we have $\mathcal C_t>0$ with  \begin{align}\label{eq-error-accumulation}  \mathcal C_t = 0.3 - O\big(\dfrac{1}{(n\Delta)^4} +\dfrac{n\varepsilon_{\x,t} }{\lVert \x(\mT_t^c) \rVert_\infty} \big) ,  \end{align}
we have then $\omega_{t+1}$ is guaranteed to locate at $\mS(\tau_{T(t+1)})$ with  $T(t+1) \in \mT_{t}^c$ and  $$\lvert x_{T(t+1)} \rvert \geq (1-\lambda_{t+1}) \lVert \x(\mT_t^c) \rVert_{\infty}, $$
 with some $\lambda_t$ such that 
\begin{align*}
	\lambda_{t} = \left\{\begin{matrix} O(\dfrac{n\varepsilon_{\x,t}}{\lVert \x(\mT_t^c) \rVert_\infty (n\Delta)^4})   & \text{, }t\geq 1 \\
	 O( {1}/{(n\Delta)^4} )& \text{, }t=0 
	\end{matrix}\right.
\end{align*}  
\end{lemma}

\begin{proof}[Proof of the Lemma] \quad \\
For $t = 0$, the claim is covered by Proposition~\ref{prop-first-step-fejer}.\\
 For $t\geq 1$, we have suppose W.L.O.G. that $x_{t+1} \in \mT_{t}^c$ and $\lvert x_{t+1}\rvert = \lVert \x(\mT_t^c)\rVert_\infty,$ then for any $\tau \in [0,1),$ denoting $\bm P_t: = \bm P(\btau(\mT_t)), \tilde{\bm P}_t : = \bm P(\bomega_t),$ \begin{align*}
	&\w(\tau)^* (\bm{I}- \tilde{\bm{P}}_t) \y \\
	 =& \w(\tau)^* (\mI - \tilde{\bm{P}}_t) [\w(\mT_t)\x(\mT_t)+\w(\mT_t^c)\x(\mT_t^c) ]\\
	 =& \underbrace{\w(\tau)^* (\mI - \tilde{\bm{P}}_t) \w(\mT_t)\x(\mT_t)}_{J_{1,t}(\tau)} + \underbrace{\w(\tau)^* (\mI - {\bm{P}}_t)\w(\mT_t^c)\x(\mT_t^c)}_{J_{2,t}(\tau)} +\underbrace{\w(\tau)^* ({\bm{P}}_t- \tilde{\mP}_t )\w(\mT_t^c)\x(\mT_t^c)}_{J_{3,t}(\tau)} .
	 \end{align*}
As a result, we have
 \begin{equation}\label{eq-condition-hold} \max_{i\in \mT_t^c} \lvert J_{2,t}(\tau)\rvert  - \sup_{\tau\notin \mS(\mT_t^c)} \lvert J_{2,t}(\tau) \rvert  > 2 \big( \sup_\tau  \lvert J_{1,t}(\tau)\rvert + \sup_\tau  \lvert J_{3,t}(\tau)\rvert \big) \implies  \omega_{t+1}\in \mS(\mT_t^c),\end{equation} 
To guarantee the left-hand side inequality in \eqref{eq-condition-hold}, we build the following bounds on $J_{1,t},J_{2,t},J_{3,t}$:	 
 \begin{lemma}\label{lem-J-bounds}
	As long as $(n+2)\Delta > C$ and $\lvert \omega_i - \tau_{T(i)}\rvert \leq \frac{1}{2n+4},$ we have  
	\begin{align*}
&\max_{i\in \mT_t^c}  \lvert J_{2,t}(\tau) \rvert  \geq (1-\frac{c}{(n+2)^4\Delta^4} )\lvert x_{t+1
}\rvert, \\
&\sup_{\tau \notin \mS(\mT_t^c) }	\lvert J_{2,t}(\tau) \rvert \leq   \lvert x_{t+1} \rvert \bigg( 0.7 +\nu_2  \bigg) , \\
&\sup_{-1/2\leq \tau \leq 1/2 }	\lvert J_{1,t}(\tau) \rvert \leq   n\varepsilon_{\x,t} \cdot \nu_1 ,  \\
&\sup_{-1/2\leq \tau \leq 1/2 }	\lvert J_{3,t}(\tau) \rvert \leq \lvert x_{t+1}\rvert  \cdot \nu_3 .
	\end{align*}
	Where $\nu_1 =O(1), \nu_2 , \nu_3 = O(\dfrac{1}{(n\Delta)^4}). $ \end{lemma}  
Now we have for  $\tau \notin \mS(\mT_t^c), $   \begin{align*}
	\lvert \w(\tau)^*(\mI-\tilde{\mP}_t  ) \y \rvert 
	& \leq \sup_{\tau\notin \mS(\mT_t^c)} \lvert J_{2,t}(\tau) \rvert+ \sup_{\tau }\lvert J_{1,t}(\tau)\rvert+\sup_{\tau}\lvert J_{3,t}(\tau)\rvert\\
	&\leq  \lvert x_{t+1}\rvert (0.7+\nu_2+\nu_3) + n\varepsilon_{\x, t}\cdot \nu_1 . 
\end{align*}
And by \begin{align*}
	\max_{i\in \mT_t^c} \lvert \w(\tau_i)^*(\mI-\tilde{\mP}_t)\y \rvert &\geq \max_{i\in \mT_t^c} \lvert J_{2,t}(\tau_i) \rvert - \sup_{\tau} \lvert J_{1,t}(\tau) \rvert - \sup_{\tau}\lvert J_{3,t}(\tau) \rvert \\
	& \geq \big(1-\dfrac{c_1}{(n+2)^4 \Delta^4} - \nu_3  \big)\lvert x_{t+1}\rvert -  n\varepsilon_{\x,t} \cdot \nu_1  
\end{align*}
Thus  \eqref{eq-condition-hold} 
is implied by \begin{align*}
 \big(1-\dfrac{c_1}{(n+2)^4 \Delta^4} - \nu_3  \big)\lvert x_{t+1}\rvert -  n\varepsilon_{\x,t} \cdot \nu_1   	&>     \lvert x_{t+1}\rvert (0.7+\nu_2+\nu_3) + n\varepsilon_{\x, t}\cdot \nu_1 .      
\end{align*} 
i.e. \begin{align}
	0.3 >  \nu_2+2\nu_3+ \dfrac{c_1}{(n+2)^4\Delta^4} + 2\dfrac{n\varepsilon_{\x, t}}{\lvert x_{t+1} \rvert } \cdot \nu_1,
\end{align}
thus if we define \begin{equation}
	\mathcal{C}_t = 0.3 - ( \nu_2+2\nu_3+ \dfrac{c_1}{(n+2)^4\Delta^4} + 2\dfrac{n\varepsilon_{\x, t}}{\lvert x_{t+1} \rvert } \cdot \nu_1) ,
\end{equation}
then the claim holds. \\
Finally, suppose $\hat\omega_{t+1} \in \mS(x_{T(t+1)})$ for some $T(t+1) \in \mT_t^c,$ we have then \begin{align*}
		\lvert \w(\hat\omega_{t+1})^*(\mI-\tilde{\mP}_t ) \y \rvert  > \lvert \w(\tau_{t+1})^*(\mI-\tilde{\mP}_t ) \y \rvert 
\end{align*} 
implies \begin{align*}
	\lvert J_{2,t} (\hat\omega_{t+1}) \rvert - \lvert J_{2,t}(\tau_{t+1}) \rvert \geq  2\sup_{\tau \in \mS(\mT^c_t)} \big( \lvert J_{1,t}(\tau )\rvert + \lvert J_{3,t}(\tau) \rvert \big)
\end{align*}
When $\tau \in \mS(\mT_t^c)$, we have the following improved bounds for $J_{1,t}$ and $J_{3,t}$ :
\begin{lemma}\label{lem-J-small-bound} We have \begin{align*}
	\sup_{\tau \in \mS(\mT_t^c)} \lvert J_{1,t}(\tau) \rvert &\leq n\varepsilon_{\x,t} {\kappa_1}, \\
	\sup_{\tau \in \mS(\mT_t^c)} \lvert J_{3,t}(\tau) \rvert &\leq  \lvert x_{t+1}\rvert   {\kappa_3} .
\end{align*}
	And \begin{align*}
		\sup_{\omega\in \mS(\tau_i)} \lvert  J_{2,t}(\omega) \rvert &\leq \lvert x_i \rvert + \lvert x_{t+1} \rvert {\kappa_{21}}, \\ 
		 \lvert J_{2,t}(\tau_i) \rvert &\geq  \lvert x_{i}\rvert -\lvert x_{t+1}\rvert {\kappa_{22}} .
	\end{align*}
	when $\tau_i \in  \mT_t^c.$ Where $\kappa_{21},\kappa_{22},\kappa_{1},\kappa_3 = O( \dfrac{1}{(n\Delta)^4}). $ 
\end{lemma}
\noindent  As a result, we get \begin{align*}
	 \lvert x_{T(t+1)}\rvert \geq \bigg[1- \big ( 2\dfrac{n\varepsilon_{\x,t}}{\lvert x_{t+1}\rvert}   \kappa_1 + \kappa_{21}+\kappa_{22}+2\kappa_{3}\big ) \bigg] \lvert x_{t+1}\rvert, 
\end{align*}
and the second claim holds by plugging the order of $\kappa_{1},\kappa_{21},\kappa_{22},\kappa_{3}$ into the formula. 
\end{proof}

\subsubsection{Improved Localization Error}\label{sec-improved-estimation}
In the previous section, we give a criterion on the approximate localization guarantee with accuracy $\frac{1}{2n+4}$. We will show that the accuracy can be further improved to   $\varepsilon \ll \frac{1}{2n+4}.$

\begin{property}\label{prop-improved-localization} At $t$-th iteration, denote $T(t): =\text{argmin}_{i\in [s]} \lvert \tau_i - \omega_t\rvert, $ we have then $$\omega_{t}\in \mS(\tau_{T(t)} )\implies \lvert \omega_{t}-\tau_{T(t)}\rvert \leq \dfrac{1}{n}\sqrt{\dfrac{\lambda_t}{1.19(1-\lambda_t)} }  $$  for $\lambda_t$ defined in Lemma~\ref{lem-error-formula}.
\end{property}
\begin{proof}\quad \\
\textbf{The first iteration: } For the first step, suppose W.L.O.G. $\lvert x_1 \rvert = \max_{i\in [s]}\lvert x_i\rvert,$  we have by $\hat\omega_1 \in \mS(\tau_{T(1)}),$  \begin{align*}
	&\lvert  \w(\hat\omega_{1})^*\y \rvert \geq \lvert \w(\tau_{T(1)})^*\y \rvert \\
	\implies & \lvert \sum_{i=1}^s x_i K_n(\hat\omega_1 -\tau_i)\rvert \geq  \lvert \sum_{i=1}^s x_i K_n(\tau_{T(1)} -\tau_i)\rvert\\
	\implies & \lvert x_{T(1)}\rvert (1- K_n(\hat\omega_1-\tau_{T(1)}))  \leq \lvert x_1 \rvert \dfrac{2 c_1}{(n+2)^4\Delta^4}\\
	\implies &1-K_n  (\hat\omega_1-\tau_{T(1)}) \leq   \dfrac{2c_1 \lvert x_1\rvert}{\lvert x_{T(1)}\rvert (n+2)^4\Delta^4} \\
	\implies& -\int_{0}^{\hat\omega_1-\tau_{T(1)}} \int_0^u K_n''(v)  dvdu \leq   \dfrac{2c_1 \lvert x_1\rvert}{\lvert x_{T(1)}\rvert (n+2)^4\Delta^4}
\end{align*}
In particular, we have  \begin{align*}
	\int_{0}^{\hat\omega_1-\tau_{T(1)}} \int_0^u K_n''(v)  dvdu &= \int_{0}^{\lvert \hat\omega_1-\tau_{T(1)}\rvert } \int_0^{ u } K_n''(v)  dvdu\\
	& \leq  \int_{0}^{\lvert \hat\omega_1-\tau_{T(1)}\rvert } \int_0^{ u } -\dfrac{\pi^2}{3}n(n+4) +\dfrac{\pi^4}{6}(n+2)^4 v^2  dvdu\\
	&\leq -\dfrac{\pi^2}{6}n(n+4)\lvert \hat\omega_1-\tau_{T(1)}\rvert^2  + \dfrac{\pi^4}{6}(n+2)^4\dfrac{1}{12}\lvert \omega_1-\tau_{T(1)}\rvert^4 \\
	&\leq -\big[\dfrac{\pi^2}{6} - \dfrac{\pi^4}{72\cdot 3}     \big] n^2 \lvert \hat\omega_1-\tau_{T(1)}\rvert^2 \\
	&\leq -1.19 n^2\lvert \hat\omega_1-\tau_{T(1)}\rvert^2.
\end{align*}
As a result, we get \begin{align*}
	\lvert \hat\omega_1-\tau_{T(1)}\rvert \leq \dfrac{1}{\sqrt{1.19}n} \big[(1-\dfrac{2c_1}{n^4\Delta^4})^{-1} \cdot \dfrac{2c_1}{n^4\Delta^4} \big]^{1/2}
\end{align*}

Thus the claim holds for the first iteration.
\\
\textbf{The $t+1 $-th iteration: } Suppose W.L.O.G.$\lvert x_{t+1}\rvert = \lVert \x(\mT_{t}^c)\rVert_\infty$, then by $  \omega_{t+1} \in \mS(\tau_{T(t+1)}) $ and  \begin{align}\label{eq-OMP-stept-basic-inequality}
	 &\lvert \w(\hat\omega_{t+1})^*(\mI - \tilde{\mP}_t) \y \rvert \geq \lvert \w(\tau_{T(t+1)})^* (\mI - \tilde{\mP}_t)\y\rvert . 	
\end{align}
Then   \begin{align*}
	\lvert x_{T(t+1)}\rvert (1-K_n(\hat\omega_{t+1}-\tau_{T(t+1)})) \leq \lvert x_{t+1}\rvert  \underbrace{\big(\dfrac{2n\varepsilon_{\x,t}}{\lvert x_{t+1}\rvert}   \kappa_1 + \kappa_{21}+\kappa_{22}+2\kappa_{3}  \big)} _{\lambda_t},
\end{align*} 
now as argued in step one, we get \begin{align*}
	\lvert \hat\omega_{t+1}-\tau_{T(t+1)}\rvert \leq \dfrac{1}{\sqrt{1.19}n} \big[(1-\lambda_{t+1})^{-1}\lambda_{t+1} \big]^{1/2},
\end{align*}
as desired. 

\end{proof}

\subsection{Analysis of Sliding operation}\label{sec-proof-sliding-lanscape}

We first provide a formal version of Proposition~\ref{prop-sliding-informal}, which describes the region of $\bomega$ satisfying the weak-regularity condition $\mL_t$:
\begin{property} \label{thm-sliding-regularity}
For any $t\leq s,$  supposing $\min_{i\in \mT_t}\lvert \x_i\rvert > \frac{1}{2}\lVert \x(\mT_t^c)\rVert_\infty,$ there exists some absolute constant $\eta >0$ and $\Crad$ independent of $t$  so that for $$\bm{w}: =   [\W (\bomega^0)^* \W (\bomega^0)]^{-1} \W (\bomega^0) ^* \bm y ,$$ the weighted gradient function $\bm g(\bomega):=
 	\dfrac{\eta}{\lvert  \bm{w}\rvert^2  n^2}  \odot \nabla \mL_{t}(\bomega) $  satisfies WRC$_{\x(\mT_t)}(\bm{\lambda}, \alpha,\beta)$ with $\min_i \lambda_i  >0$ and $0<\alpha+\beta<1$ over the region \begin{align*}
  		\mB_t: =
  	\{\mathcal{R}_t	\dfrac{\lVert \x(\mT_t^c)\rVert_\infty }{n} \leq \lVert \bomega - \btau_{\mT_t}\rVert_{t,\infty}\leq \dfrac{\Crad }{n}\lVert  \x(\mT_{t-1}^c) \rVert_\infty  \}
  \end{align*}
  for some  $\mathcal R_t = \dfrac{C}{(n\Delta)^4 }$
  as long as $\mathcal{CR}_t >0$ 
  with \begin{align*}
  	 \mathcal{CR}_t =(1-\dfrac{81}{4\pi^4}) \dfrac{1}{16} - c\big(\dfrac{1}{(n\Delta)^4}+\dfrac{n\varepsilon_{\x,t}}{\min_{i\in \mT_t} \lvert x_i\rvert} \big)^2.
  \end{align*}
  Moreover, we have when $\lVert\bomega-\btau_{\mT_t} \rVert_{t,\infty} \leq \dfrac{\Crad}{n}\lVert \x(\mT_{t-1}^c)\rVert_\infty, $  \begin{align*}
  	\lvert  x_m g_m \rvert \leq \phi_1 \varepsilon_{\x,t}+\phi_2\frac{\lVert \x_{>t}\rVert_\infty}{n} 
  \end{align*}
  with $\phi_1,\phi_2\lesssim \frac{1}{n^4\Delta^4}$.
  \end{property}
\begin{remark}
Comparing with previous results \cite{Eftekhari2015,Traonmilin2020} on the basin of attraction of \eqref{eq-global-loss-intro}, the Proposition~\ref{thm-sliding-regularity} holds for more general sliding losses at every $1\leq t\leq s$. We would stress here that even when $t=s$, our result has a significant difference between previous results: while previous results describe the convergence region under $\ell_2$ geometry, our result gives the criterion of convergence under weighted $\ell_\infty$ norm, which is more suitable when analyzing the greedy-type algorithms. 
  \end{remark}
  \begin{property}[Estimation Error]\label{thm-sliding-est-error}
  Under the same assumption and notation in Proposition~\ref{thm-sliding-regularity}, supposing in addition that $$\lVert \bomega^0 - \btau_{\leq t}\rVert_{t,\infty} \leq \frac{\Crad}{n}\lVert \x(\mT_{t-1}^c)\rVert_\infty,$$ then for any $\epsilon>0,$ we have the weighted gradient descent iteration sequence generated by $$\bomega^{k+1} = \bomega^k -  \bm g(\bomega^k) $$  
    will satisfy $
   \lVert	\bomega^{k} -\btau_{\leq t}\rVert_{t,\infty}  \leq \bar{\epsilon}:= \min\{ \epsilon, \mathcal{E}_t  \dfrac{ \lVert \x_{>t}\rVert_\infty}{n}    \}
  $
  with $$\mathcal E_t = \phi_1 \mathcal{R}_t+\phi_2 = O(\dfrac{1}{(n\Delta)^4} )$$ 
  after at most $O\big(\log( \dfrac{1}{\bar{\epsilon}})\big)$ iterations.
 \end{property}

\begin{proof}[Proof of Proposition~\ref{thm-sliding-est-error}]
 Supposing W.L.O.G. $T(i) = i,$ then for any integer $k\geq 0$, if $$\mathcal{R}_t\dfrac{\lVert \x_{>t}\rVert_\infty}{n}  \leq \lVert\bomega^k-\btau_{\leq t} \rVert_{t,\infty} \leq \dfrac{\Crad}{n}\lVert \x_{\geq t}\rVert_\infty , $$ we have by Proposition~\ref{thm-sliding-regularity},  \begin{align*}
 	\lVert \bomega^{k+1} -\btau_{\leq t}\rVert_{t,\infty}\leq (1- c)\lVert \bomega^k - \btau_{\leq t}\rVert_{t,\infty}.
 \end{align*}
 As a result, if we denote $$k_0:= \inf\{k\in \mathbb{Z}_{+}, \lVert \bomega^{k} -\btau_{\leq t}\rVert_{t,\infty} < \mathcal{R}_t\dfrac{\lVert \x_{>t}\rVert_\infty}{n}   \},$$ 
 we have $k_0 \lesssim -\log \big(\mathcal{R}_t\frac{\lVert \x_{>t}\rVert_\infty}{n} \big)$, and we would finish the proof by arguing $\lVert \bomega^{k}-\btau_{\leq t}\rVert_{t,\infty}\leq \dfrac{\mathcal{E}_t\lVert \x_{>t}\rVert_\infty}{n}$ for all $k > k_0.$

 When $\lVert \bomega^{k} -\btau_{\leq t}\rVert_{t,\infty} < \mathcal{R}_t\frac{\lVert \x_{>t}\rVert_\infty}{n}$, we have for every coordinate $1\leq m \leq t$, 
\begin{align*}
 		\lvert  x_{m} (\omega_m^{k+1} -\tau_m) \rvert \leq   \lVert \bomega^k-\btau_{\leq t} \rVert_{t,\infty}  + \lvert   x_m g_m(\bomega^k)  \rvert 
 	\end{align*}
 	By Proposition~\ref{thm-sliding-regularity},  \begin{align*}
 		 \lvert  x_m g_m (\bomega^k)  \rvert  &\leq  \phi_1  \varepsilon_{\x,t} + \dfrac{\phi_2}{n}\lVert \x_{>t}\rVert_\infty ,  
 	\end{align*}
 	thus \begin{align*}
 		\lVert \bomega^{k+1}-\btau_{\leq t} \rVert_{t,\infty} &<\underbrace{ \big( {\mathcal{R}_t[1+\phi_1}]+\phi_2  \big)}_{\mathcal{E}_t}\dfrac{\lVert \x_{>t}\rVert_\infty }{n}.
 		  	\end{align*}
 Then by $\mathcal{E}_t \lVert \x_{>t}\rVert_\infty \leq \Crad \lVert \x_{\geq t}\rVert_\infty,$ we have either $\bomega^{k+1}\in \mB_t$ or $\lVert \bomega^{k+1} -\btau_{\leq t}\rVert_{t,\infty} < \mathcal{R}_t\frac{\lVert \x_{>t}\rVert_\infty}{n}$, in both cases, previous arguements imply $\lVert \bomega^{k+2} -\btau_{\leq t}\rVert_{t,\infty}< \mathcal{E}_t\frac{\lVert \x_{>t}\rVert_\infty}{n},$ thus the claim holds.

\end{proof}

\subsection{Proof of Theorem~\ref{thm-sliding-OMP-guarantee} }

\begin{proof}\textbf{At the first iteration: }\\
By Proposition~\ref{prop-first-step-fejer} we have $\hat\omega_1 \in \mathcal{S}(\tau_{T(1)})$ for some $\tau_{T(1)} \in \mathcal{T}$ and $\lvert x_{T(1)}\rvert \geq (1- O( \frac{1}{n^4\Delta^4}) ) \lVert \x \rVert_\infty.$ Then by Proposition~\ref{prop-improved-localization}  , we have  there exists some absolute constant $a_0>0$ so that when $n\Delta > a_0,$  $ \varepsilon_1 = \lvert  \hat \omega_1 - \tau_{T(1)}\rvert\lesssim \dfrac{1}{n(n\Delta)^2}.     $
 Then \begin{align*}
	 \dfrac{\hat\varepsilon_{\x,1}}{\lVert \x \rVert_\infty} \lesssim \dfrac{1}{n(n\Delta)^2}, \quad  \frac{\hat\varepsilon_{\x,1}}{\min_i \lvert  \x(\mT_1)_i \rvert } = \varepsilon_1 \lesssim \frac{1}{n(n \Delta)^2},  
\end{align*}
thus when $n\Delta > C$ for some $C$, the requirement of Proposition~\ref{thm-sliding-est-error} is satisfied, then it holds that ${\varepsilon}_{\x,t}\lesssim \frac{\lVert \x(\mT_1^c)\rVert_\infty}{n(n\Delta)^4}.$ 

\textbf{At the $t+1$-th iteration:} When $t+1\leq s$, supposing  by strong induction that there exists a large enough absolute constant $c'$ so that  $\varepsilon_{\x,k} \leq c'  \frac{\lVert \x(\mT_k^c) \rVert_{\infty} }{(n\Delta)^4n} $ and $\lvert x_{T(k)} \rvert > \dfrac{1}{2} \lVert \x(\mT_k^c) \rVert _\infty  $  holds for all $k \leq t$ and W.L.O.G. $T(i) = i$ for $i\leq t$.  Then we have there exists some large enough absolute constant $q\geq 1$ independent of $t$  so that when $n\Delta>q c'$, $$\mathcal C_{t} = 0.3 - O\big(\dfrac{1}{(n\Delta)^4} + \dfrac{1}{q} \big) >0.$$ By Lemma~\ref{lem-error-formula} and Proposition~\ref{prop-improved-localization} , we get  $\hat{\varepsilon}_{t+1} \lesssim \dfrac{\sqrt{c'}}{n(n\Delta)^2} $ and 
$$\lvert x_{T(t+1)}\rvert \geq (1 - O(\frac{1}{q} )) \lVert \x(\mT_t^c)\rVert_\infty \geq \dfrac{1}{2}\lVert \x(\mT_t^c)\rVert_\infty.$$
As a result of induction assumption, we have \begin{align*}
	&\hat{\varepsilon}_{\x,t+1} = \max \{\varepsilon_{\x,t}, \lvert x_{T(t+1) }(\hat{\omega}_{t+1} - \tau_{T(t+1)}\rvert  \} \lesssim \max\{\dfrac{c'\lVert \x(\mT_t^c)\rVert_\infty}{n(n\Delta)^4},\dfrac{ \sqrt{c'}\lVert \x(\mT_t^c)\rVert_\infty}{n(n\Delta)^2}\}\leq \dfrac{\lVert \x(\mT_t^c)\rVert_\infty }{n \sqrt{ q}} , \\
	&\dfrac{\hat{\varepsilon}_{\x,t+1}}{\min_{i\in \mT_{t+1}}\lvert x_i \rvert}  \lesssim  \dfrac{\lVert \x(\mT_t^c)\rVert_\infty}{n \sqrt{q} \min_{i\in \mT_{t+1}}\lvert x_i \rvert }  \leq \dfrac{2}{n\sqrt{q}}  .
\end{align*}
Thus the condition in Proposition~\ref{thm-sliding-est-error} is satisfied, then it holds that $\varepsilon_{\x,t} \lesssim \dfrac{\lVert \x(\mT_t^c)\rVert_\infty}{n (n\Delta)^4}$. In particular, when $a_2$ is large enough, we get $\varepsilon_{\x,t} \leq \dfrac{\lVert \x(\mT_t^c)\rVert_\infty}{n(n\Delta)^4}$.  

In conclusion, our argument shows that there exist absolute constants $q,c' >0$ so that when $n\Delta > q c',$  the induction result holds with $c'$. That finishes the proof of Theorem~\ref{thm-sliding-OMP-guarantee}.
 \end{proof}

\section{Analysis with Incomplete Measurements}\label{sec-analysis-incomplete-sample}

\subsection{Uniform Concentration of $K_n$}
While all of our conclusions presented in the previous section are under the full-sample setting, we will show it is painless to convert the result into the subsampling case by bounding the considered concentration kernels and their derivatives uniformly over $\tau\in  [-\frac{1}{2},\frac{1}{2}]$: 

\begin{lemma}\label{lem-uniform-concentration}
	For Squared Fejer kernel $K_n(\tau),$ and $0\leq \delta \leq \frac{1}{2}$,  as long as $np \gtrsim \sqrt{\log(1/\delta) }$, we have with probability at least $1-\delta,$ \begin{align*}
		\sup_{\tau\in [-\frac{1}{2},\frac{1}{2}]} \dfrac{1}{n^q} \lvert  \tilde{K}_n^{(q)}(\tau) -p \K_n^{(q)}(\tau) \rvert  \lesssim  \sqrt{\dfrac{p}{n}(\log n/\delta)},\quad \forall q = 0,1,2.
	\end{align*}
\end{lemma}
\begin{proof}
	By  Bernstein's inequality, we have  as long as $np  >  t^2 ,$  for any fixed $\tau \in [-\frac{1}{2},\frac{1}{2}]$ we have \begin{align*}
	\Probability\big( \lvert \sum_{\ell = -n}^n  (B_\ell \sigma_\ell - p \sigma_\ell ) e^{2\pi j \ell \tau } \rvert > \sqrt{np}\lVert \sigma\rVert_\infty t\big) \leq \exp(-Ct^2 )
\end{align*}
Consider a uniform $1/N$-net $\mathcal{N}$ of $[-\frac{1}{2},\frac{1}{2}]$, we then get by Bernstein's inequality, with probability at least $1- N \exp(-t^2),$ \begin{align*}
	 \lvert \sum_{\ell = -n}^n  (B_\ell \sigma_\ell - p \sigma_\ell ) e^{2\pi j \ell \tau } \rvert \leq  \sqrt{np}\lVert \sigma\rVert_\infty t, \quad \forall \tau \in \mathcal{N}, 
	\end{align*}
then by for any $\tau\in [-\frac{1}{2},\frac{1}{2}]$, consider its best approximation $\tilde{\tau}\in \mathcal{N}$, we have by \begin{align*}
	 \lvert \sum_{\ell = -n}^n  (B_\ell \sigma_\ell - p \sigma_\ell ) e^{2\pi j \ell \tau } - \sum_{\ell = -n}^n  (B_\ell \sigma_\ell - p \sigma_\ell ) e^{2\pi j \ell \tilde{\tau} } \rvert &\leq  (1+p) 2\pi \dfrac{n^2}{N}  \lVert  \sigma\rVert_\infty ,\end{align*} 
\begin{align*}
	\Probability \big( \sup_{\tau \in [-\frac{1}{2},\frac{1}{2}]}  \lvert \sum_{\ell = -n}^n  (B_\ell \sigma_\ell - p \sigma_\ell ) e^{2\pi j \ell \tau } \rvert \gtrsim  \lVert \sigma\rVert_\infty \sqrt{np}t + \dfrac{n^2\lVert \sigma\rVert_\infty }{N}    \big) \lesssim N \exp(-Ct^2).
\end{align*}
Thus for any $0<u<\sqrt{np}$, selecting $N = \dfrac{n^{3/2} }{\sqrt{p}u},$ we get then  \begin{align*}
	\Probability \big( \sup_{\tau \in [-\frac{1}{2},\frac{1}{2}]}  \lvert \sum_{\ell = -n}^n  (B_\ell \sigma_\ell - p \sigma_\ell ) e^{2\pi j \ell \tau } \rvert \gtrsim \lVert \sigma\rVert_\infty \sqrt{np} u  \big) \lesssim    \dfrac{n^{3/2} }{\sqrt{p}u} \exp(-Cu^2) .
\end{align*}
Selecting $u = C^{-1} \sqrt{\log( \dfrac{n^{3/2}}{\delta\sqrt{p} })}$, we have then when $np \gtrsim \sqrt{\log(1/\delta)}$ so that
 $\log\dfrac{ n^{3/2} }{\delta \sqrt{p}} < \sqrt{np} $
  , with probability at least $1-\delta$, 
 $$ \sup_{\tau \in [-\frac{1}{2},\frac{1}{2}]}  \lvert \sum_{\ell = -n}^n  (B_\ell \sigma_\ell - p \sigma_\ell ) e^{2\pi j \ell \tau } \rvert \lesssim \lVert \sigma\rVert_\infty \sqrt{np    \log(\dfrac{n^{3/2}}{\delta \sqrt{p}})} \lesssim \lVert \sigma\rVert_\infty \sqrt{np\log(n)+np\log(1/\delta)} .   $$
 Then by the formula of $\sigma$ for squared Fejer kernel in \eqref{eq-particular-preconditioner}, we have $\lVert \sigma\rVert_{\infty} = O(1/n),$ thus the result holds when $q = 0.$ For $q = 1$ and $q = 2$, just notice that $\tilde{K}_n^{(q)}(\tau) = \sum_{\ell = -n}^n n^{q}B_\ell  (2\pi j \ell \sigma_\ell) e^{2\pi j \ell\tau} $ and applies above argument with $ 2\pi n\lVert \sigma\rVert_\infty $ and $4\pi^2 n^2 \lVert\sigma\rVert_\infty. $
\end{proof}

\subsection{Proof of Theorem~\ref{thm-subsampling-result}}
The proof of Theorem~\ref{thm-subsampling-result} relies on developing similar results as in section~\ref{sec-approximate-localization},\ref{sec-improved-estimation},\ref{sec-proof-sliding-lanscape} based on Lemma~\ref{lem-uniform-concentration}.

We summarize the extension of Lemma~\ref{lem-error-formula}, Proposition~\ref{prop-improved-localization} and Proposition~\ref{thm-sliding-regularity}  here and leave the proofs into Appendix~\ref{appendix-analysis-incomplete}.

\begin{property}\label{prop-subsampled-approximate-localization}  The claim in Lemma~\ref{lem-error-formula} holds for Bernoulli-$p$ subsampled observations with probability at least $1-1/n^2$ and $\mathcal {C}_t,\lambda_t$ replaced by \begin{align*}
	 \tilde{\mathcal C}_t 		 & = 0.3 - O\big( \dfrac{1}{(n\Delta)^4}+ s\sqrt{\log(n)/np}+\dfrac{n\varepsilon_{\x,t}}{\lvert x_{t+1}\rvert} \big)  \\
		 \tilde{\lambda}_1 & = 1- O \big(\dfrac{1}{(n\Delta)^4}+ s\sqrt{\log(n)/np} \big)    \\
		 \tilde{\lambda}_t 
		 & = 1-O\big(\dfrac{n\varepsilon_{\x,t}}{\lvert x_{t+1}\rvert} [\dfrac{1}{(n\Delta)^4}+ s\sqrt{\log(n)/np} ]  \big) , \quad t >1.    
\end{align*}
	
\end{property}

\begin{property}\label{prop-subsampled-improved-estimation}
	The claim in Proposition~\ref{prop-improved-localization} holds for Bernoulli-$p$ subsampled observations with probability at least $1-1/n^2$ and $\lambda_t$ replaced by $\tilde{\lambda}_t$ 
\end{property}

\begin{property}\label{prop-subsampled-sliding} 
	The claim in Propositoin~\ref{thm-sliding-regularity} and Proposition~\ref{thm-sliding-est-error} holds for symmetric Bernoulli-$p$ subsampled observations with probability at least $1-1/n^2$ and $\mathcal{R}_t ,\mathcal {CR}_t,\mathcal{E}_t$ replaced by \begin{align*}
		\tilde{\mathcal{R}}_t:=& \big[ \dfrac{1}{2} - O(\dfrac{1}{(n\Delta)^4}+ s\sqrt{\dfrac{\log n}{np}}) \big]  \cdot O\big( \dfrac{1}{(n\Delta)^4} + s \sqrt{\dfrac{\log n}{np}} \big)   ,\\
		\tilde{\mathcal{CR}}_t:= & \dfrac{1-18/\pi^4}{16} -  O( \dfrac{1}{(n\Delta)^4} + s \sqrt{\dfrac{\log n}{np}}+ n\varepsilon_t)^2  ,\\
		\tilde{\mathcal{E}}_t:=& \big(1 + [\dfrac{1}{2}-O(  \dfrac{1}{(n\Delta)^4} + s \sqrt{\dfrac{\log n}{np}}  ) ]^{-1}   \big)   \cdot O( \dfrac{1}{(n\Delta)^4} + s \sqrt{\dfrac{\log n}{np}}) .
	\end{align*} 
\end{property}

The claim then follows the same routine as in Theorem~\ref{thm-sliding-OMP-guarantee}.

\section{Results for Continuous OMP} \label{sec-byproduct-OMP}

In this section, we summarize the by-product results for continuous OMP.

\begin{theorem}\label{thm-omp-naive} There exists $c<c',C,C'$ so that
	as long as $$n\Delta >\zeta\cdot \dynx , $$
	for $\zeta\geq C,$ we have Algorithm~\ref{alg-COMP} with the preconditioner  and stopping threshold $c\cdot  \min_i\lvert x_i\rvert <\gamma <c' \min_i\lvert x_i\rvert $ will stop exactly after finding $s$ frequenceis $\omega_1,\dots,\omega_s$. Moreover, there exists unique $T(1),\dots, T(t) \in [s]$ so that $\tau_{T(i)} =\text{argmin}_{\tau \in \mT} \lvert \tau - \omega_i\rvert$ 
	and   \begin{align*}
		\lvert \omega_{i}-\tau_{T(i)}\rvert \leq \dfrac{C'}{n} \cdot(\dfrac{1}{\dynx \zeta})^2.
	\end{align*}
\end{theorem}
\begin{remark}
Our result implies that when $n\Delta \gtrsim \dynx,$ the OMP algorithm with suitable stopping threshold can find exactly $s$ frequencies $\bomega$ with $\lVert \bomega - \btau \rVert_\infty \leq \dfrac{1}{n(\dynx)^2}$ up to permutation over $\btau$. Such a result can be combined with the basin-of-attraction result for the loss \eqref{eq-global-loss-intro}  in \cite{Eftekhari2015, Traonmilin2020}  to get the performance guarantee of the two-staged algorithm with OMP initialization.
\end{remark}
\begin{remark}
While the Theorem is stated for pre-conditioned continuous OMP, it also holds for naive continuous OMP with possibly different absolute constants, worse dependency on $\dynx$, and an additional $\log s$ term in separation condition and the estimation error bound. 
\end{remark}
\begin{remark}
	Since the proof of Theorem~\ref{thm-omp-naive}  relies mostly on Proposition~\ref{prop-improved-localization}, whose sub-sampled analogue has been developed in section~\ref{sec-analysis-incomplete-sample}, the result of Theorem~\ref{thm-omp-naive} also holds with $O(s^2\log n)$ measurements. 
\end{remark}

\begin{theorem}[Impossible result for OMP]\label{thm-impossible-OMP}
	For any positive number $c>0,L>0$ there exists some constant $c_0>c$ and positive number $C_1 $ depends on $c$ so that:  there exists a $3$-spike instance with $ x_1>x_2>x_3>0$, $ \dfrac{x_1}{ x_3}>c,$ and $n\Delta = c_0$ such that when running OMP algorithm over such instance, we will get $\lvert \omega_3 - \tau_3\rvert \geq   {L}/{n} $. 
\end{theorem}

\section{Numerical Experiments}\label{sec-experiment}

\subsection{Implementation and Computational Complexity	}
We show that in the incomplete measurement setting with $O(s^2\log n)$ measurement number, the complexity of the Sliding-OMP algorithm and the two-stage OMP algorithm are $O\big(sn\log n + \text{poly}(s,\log n)\big)$.

\paragraph{Theoretical guarantee of grid-based implementation} While in the algorithm description and previous analysis, the correlation-maximization procedure is taken over the continuous region $\tau \in [0,1)$, it is non-trivial to solve this continuous-optimization problem due to the non-convexity. In our implementation, we consider discretizing $[0,1)$ into $N_{\text{grid}}$ grids uniformly and search the maximizer over the grids.  We would show that  $N_{\text{grid}} \asymp n$ is sufficient to provide the same guarantee as the continuous setting.

 Denoting $\hat\omega_{t,\text{Grid}}$ as the correlation maximizer at $t$-th step, we show the following lemma: \begin{lemma}\label{lem-grid-bound}
	Denote $\hat{\omega}_t$ the $t$-th step correlation-maximizer in Algorithm~\ref{alg-SOMP} and  $\hat\omega_{t,\text{grid}}$ the  $t$-th step correlation maximizer on the grid. We have $\hat\omega_{t,\text{grid}} \in \mathcal{S}(\tau_{T(t)})$ for the same $T(t)$ as in Lemma~\ref{lem-error-formula} and  \begin{align*}
	\lvert	\w^*(\hat\omega_{t,\text{grid}}) \bm r_{t-1} \rvert > \lvert	\w^*(\tau_{T(t)}) \bm r_{t-1} \rvert - O( \dfrac{n}{N_{\grid}}\big[\dfrac{n\varepsilon_{\x,t-1}}{n^4\Delta^4}+\lVert \x(\mT_{t-1}^c)\rVert_\infty \big]  ).
	\end{align*}
\end{lemma}
With Lemma~\ref{lem-grid-bound}, we can show that \begin{equation}\label{eq-grid-error} \lvert\hat{\omega}_{t,\text{grid}} - \tau_{T(t)} \rvert \lesssim  \big[ \dfrac{1}{n (n\Delta)^2 } + \dfrac{1}{\sqrt{nN_{\text{grid} } } } \big]\cdot \big[\sqrt{\dfrac{n\varepsilon_{\x,t-1}}{\lVert \x(\mT_t^c)\rVert(n\Delta)^2}}+ 1 \big]  \end{equation} following the same routine as Proposition~\ref{prop-improved-localization}.
 Now if we plug in the bound \eqref{eq-grid-error} into the proof of Theorem~\ref{thm-sliding-OMP-guarantee}, we can find that there exists some  $C$ such that as long as $(n\Delta)^2>C, N_{\text{grid}}> Cn,$ we have $
 	\lvert x_{T(t)} (\hat{\omega}_{t,\text{grid}} - \tau_{T(t)}) \rvert \leq C_{\text{radius}}.
$ The additional ${N_{\text{grid}}}$ term in \eqref{eq-grid-error} will not affect the proof of Theorem~\ref{thm-sliding-OMP-guarantee}, thus the same performance guarantee holds for $N_{\text{grid}} \gtrsim n.$ The same argument also holds for continuous OMP and the incomplete measurements.

\paragraph{Complexity of grid-based implementation} Firstly, for any $\bm r_t,$ we have $O(N_{\text{grid}} \log N_{\text{grid}})$ implementations for the correlation-maximization step $\max_{\tau \in \text{Grid}}\lvert \bm{f}(\tau)^* \bm r_t\rvert $ over grids via Fast-Fourier-Transformation.  On the other hand, both updating $\bm r_t$ and optimizing sliding loss are problems of scale $s^2\log n$, which have at most $\text{poly}(s,\log n)$ complexity. Finally, since the algorithm stops after $s$ iterations and $N_{\text{grid}}\asymp n$, the total complexity is $O(sn\log n + s^2\log n ).$

\subsection{Preconditioning}\label{sec-precondition-experiment}

In previous sections, we have presented our results with the preconditioner \eqref{eq-particular-preconditioner} , which corresponds to the squared Fej\'er kernel $K_{n,\sigma}$ \eqref{eq-squared-Fejer}. Our bounds developed for $K_{n,\sigma} $ in Proposition~\ref{prop-squared-fejer-kernel} can be generalized (with different  absolute constants) to trigonometric concentration kernels of type\begin{align}\label{eq-general-concentration}
	 K_{n,\sigma; \alpha}(t):= \big[\dfrac{\sin((\frac{2n}{\alpha})\pi t )}{(\frac{2n}{\alpha}+1)\sin (\pi t)}  \big]^{\alpha}
\end{align} for arbitrary $\alpha\in \mathbb{Z}_+$ with replacing $(n\Delta)^4$ by $(n\Delta)^{\alpha}$.  In particular, \eqref{eq-general-concentration}  is provable to be a 
polynomial of $\{\exp(2\pi j  kt)\}_{k = -n}^n$ when $n/\alpha\in \mathbb{Z}_+$ so there exists an corresponding preconditioner. In this section, we compare three normalized concentration kernels:  Dirichlet Kernel ($\alpha = 1$, no preconditioning), Fej\'er Kernel ($\alpha = 2$), Squared Fej\'er Kernel ($\alpha = 4$). Figure~\ref{fig-kernel} shows the behaviour of these kernels when $n = 100$ .  \begin{figure}[H]
 \centering
\begin{subfigure}[b]{0.3\textwidth}\label{fig-sliding-bias}
 \centering
 	\includegraphics[width=\textwidth]{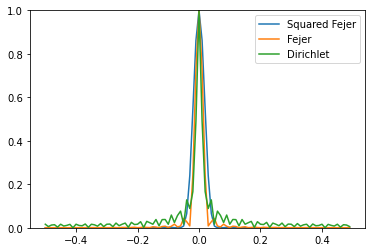}
 \caption{Various kernels}
 \end{subfigure}
  \begin{subfigure}[b]{0.3\textwidth}
 \centering
 \hspace{0cm}\includegraphics[width=\textwidth]{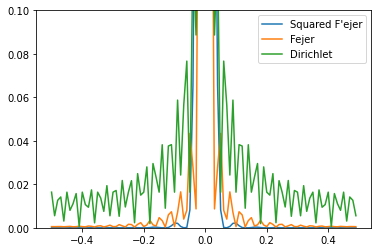}
 \caption{Tail decay rates}
 \end{subfigure}
   \begin{subfigure}[b]{0.3\textwidth}
 \centering
 \hspace{0cm}\vspace{0.1cm}	\includegraphics[width=\textwidth]{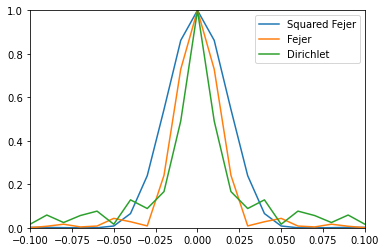}
 \caption{Concentration near the peak}
 \end{subfigure}
 \caption{Various Concentration Kernels}\label{fig-kernel}
 \end{figure}
\subsubsection*{Effect of preconditioning}
In this section, we would show the effect of the preconditioning for the OMP algorithm. While it has been shown in section~\ref{sec-byproduct-OMP} that the $n\Delta$ must depend on $\dynx$ to guarantee the performance of the OMP algorithm, we would show that preconditioning can help reduce such dependency. Figure~\ref{fig-omp-precond} shows reconstruction loss and recovery probability for the OMP algorithm with three kernels over 10 times of simulation:
\begin{figure}[H]
\centering
  \begin{subfigure}[b]{0.4\textwidth}
 \centering
 \hspace{0cm}\includegraphics[width=\textwidth]{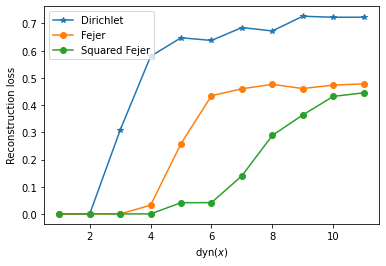}
 \caption{Reconstruction error}
 \end{subfigure}
 \qquad 
   \begin{subfigure}[b]{0.4\textwidth}
 \centering
 \hspace{0cm}\vspace{0.1cm}	\includegraphics[width=\textwidth]{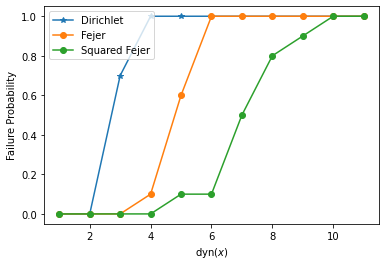}
 \caption{Probability of failure}
 \end{subfigure}
\caption{
 Performance of preconditioned OMP algorithms under different $\dynx$, with $180$ uniformly sub-sampled  measurements,  $n\Delta = 1.15$, 5 frequencies $\tau$ generated as $\tau_i = \frac{(1+i)\Delta }{n},  2n+1 = 789$ and 1800 grids. }
 \label{fig-omp-precond}
\end{figure} 

In the experiment , we draw the argument of $\bm x$ uniformly from $[0,2\pi)$, and set its amplitude as $$\lvert \bm x\rvert = [1,u,\max\{\frac{u}{2},1\}, 1,u ]^T, \quad 1\leq u \leq 8.$$  The algorithm output is said to recover $\btau$ if  $\max_{i} \lvert \omega_i-\tau_{T(i)}\rvert < 10^{-4}$.
As shown in Figure~\ref{fig-omp-precond}, the preconditioner with larger $\alpha$ has better tolerance on $\dynx$.  This result demonstrates our motivation for introducing precondition:    The preconditioned concentration kernel has a lighter tail, thus reducing the interaction between different spikes.

\subsubsection*{ The trade-off on $\alpha$ }

 While it can be observed that larger $\alpha$ encourages better decay on the tail of the kernel in Figure~\ref{fig-kernel}(b), Figure~\ref{fig-kernel}(c) shows that smaller $\alpha$ will lead to sharper concentration near its peak. Indeed, the phenomenon shown in Figure~\ref{fig-kernel}(c) prevents us from selecting $K_{n,\sigma;\alpha}$ with an arbitrary large $\alpha$, that provides a trade-off between the concentration near peak and the decay rate of the tail. 
 Theoretically balancing the trade-off on $\alpha$ when selecting kernels of type~\eqref{eq-general-concentration} or introducing other type trigonometric kernels to balance the trade-off between tail and peak is an interesting direction to explore.


\subsection{Effect of Sliding}

In this section, we show the necessity of sliding in the large $\dynx$ regime by comparing the Sliding-OMP algorithm with the OMP algorithm \cite{Benard2022} and the Greedy algorithm\cite{Eftekhari2015, traonmilin2020projected}. Since we focus on studying the effect of sliding operation, we fix the pre-condition kernel for all algorithms as Fej\'er kernel. 
  \begin{figure}[H]
 \centering
  \begin{subfigure}[b]{0.4\textwidth}
 \centering

 \hspace{0cm}\includegraphics[width=\textwidth]{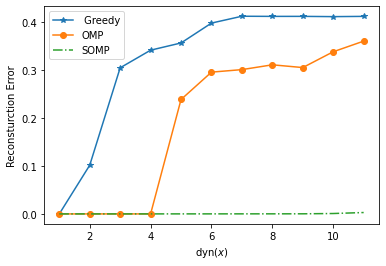} \caption{Reconstruction error}
 \end{subfigure}
 \qquad 
   \begin{subfigure}[b]{0.4\textwidth}
 \centering
 \hspace{0cm}\vspace{0.1cm}	\includegraphics[width=\textwidth]{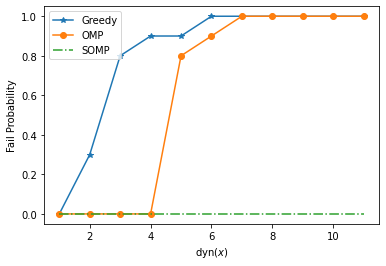}
 \caption{Probability of failure}
 \end{subfigure}
 \caption{
 Performance of algorithms under different $\dynx$, with the same setting as in Figure~\ref{fig-omp-precond}. }
 \label{fig-dynx-error}
 \end{figure} 
%
%
%
%

\subsection{Empirical Phase Transition of $n\Delta$}

In this section, we numerically evaluate our algorithms' dependency on $n\Delta$. We compare our algorithms with the greedy algorithm and the OMP algorithm in Figure~\ref{fig-dependency}. In the experiment, we set different $\dynx$ by sampling the angle of $x_i$ i.i.d. from $\text{Unif }[0,2\pi)$  and the amplitude of $x_i$ i.i.d. from $1+10^{\text{Unif}[0,v]}$ with different $v$. Larger $v$ corresponds to larger $\dynx$ in expectation. We denote the results of Sliding-OMP with Dirichlet kernel, Fej\'er kernel, Squared Fej\'er kernel   by D-SOMP, F-SOMP, SF-SOMP. 
  \begin{figure}[H]
 \centering
    \begin{subfigure}[b]{0.3\textwidth}
 \centering
 \hspace{0cm}\vspace{0.1cm}	\includegraphics[width=\textwidth]{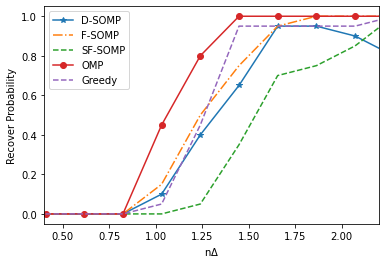}
 \caption{$v = 0.5$}
 \end{subfigure}
  \begin{subfigure}[b]{0.3\textwidth}
 \centering
 \hspace{0cm}\includegraphics[width = \textwidth]{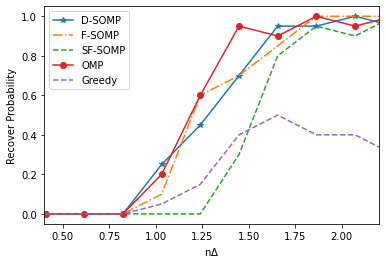}
 \caption{$v = 1$}
 \end{subfigure}
   \begin{subfigure}[b]{0.3\textwidth}
 \centering
 \hspace{0cm}\vspace{0.1cm}	\includegraphics[width=\textwidth]{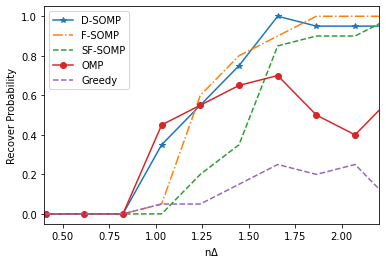}
 \caption{$v = 1.5$}
 \end{subfigure}
 \caption{
 Performance of algorithms under different $\dynx$,  with the same measurement number,  $n$,$\Delta$,$\tau$ and grid discretization in Figure~\ref{fig-omp-precond}.  }\label{fig-dependency}
 \end{figure}
 
 In the small $\dynx$ regime, the OMP and greedy algorithm have smaller phase-transition points, since the sliding guarantee may need a larger absolute constant in dependency on  $n\Delta$. When $\dynx$ is large, the Sliding-OMP algorithm is more stable than the OMP algorithm and the Greedy algorithm.

\section{Other Related Works \& Discussion}\label{sec-discuss}

\paragraph{Continuous OMP with refinements} Besides the Sliding-OMP algorithm in our work, there exist other works trying to add refinement steps to improve the performance of OMP over continuous dictionaries \cite{mamandipoor2016newtonized, marzi2016compressive, weiland2018omp, liang2021direction}. These works mainly design numerical experiments to show the efficiency of the algorithms and the paper \cite{mamandipoor2016newtonized} provides a rough convergence guarantee of their algorithm. Instead of iteratively refining in Sliding-OMP, other one-shot refinement with similar accuracy is to be explored.

\paragraph{Compressed Sensing Off-the-grid} Our result with incomplete measurements is related to the compressed sensing off-the-grid literature \cite{Tang2013}, which uses SDP-based method to estimate all coordinates of $\bm y$, under provable smaller separation gap. It would be much applicable to integrate these two methods by leveraging on their advatages.

There are also several other new topics to explore. One interesting direction is to pursue better trade-off mentioned in Section~\ref{sec-precondition-experiment} by extending  Proposition~\ref{prop-squared-fejer-kernel} to concentration kernel classes other than \eqref{eq-general-concentration}, e.g. the $G_\alpha$-cutoff kernel\cite{Tadmor2007}  or the DPSWF\cite{Eftekhari2015}. Another practical topic is to consider stable recovery in noisy case, which is much more challenging.

\bibliography{ref.bib}

\appendix
\section{Calculating the Gradient}
Denoting $\mv(\bomega): = \W(\bomega) \y$ and $A(\bomega) = \W(\bomega)^*\W(\bomega) ,$
we get now \begin{align*}
	 \partial_{\omega_m} 2\mathcal{L}_t (\bomega) & = -\partial_{\omega_m} \big( \mv(\bomega)^* A^{-1}(\bomega) \mv(\bomega) \big)\\
	 & = - 2\text{Re}\big[ \big(\partial_{\omega_m} \mv(\bomega) \big)^* A^{-1}(\bomega)\mv(\bomega)  \big] - \mv(\bomega)^* \partial_{\omega_m} (A^{-1}(\bomega) \big) \mv(\bomega) 
\end{align*}
now we introduce the new notation $ \bm{z}(\tau) = \dfrac{d}{d\tau} \w (\tau), \mZ(\bomega) = \big(\mz(\omega_1),\dots,\mz(\omega_t)\big) , \bm{W}(\bomega) = A^{-1}(\bomega)\mv(\bomega),$ then we have \begin{align*}
	\big[ \partial_{\omega_m}\mv(\bomega)\big]_{\ell} = \sum_{i=1} ^s \bar x_i \w^*(\tau_i) \partial_{\omega_m}  \w(\omega_\ell) = \delta_{lm} \sum_{i = 1}^s \bar x_i \w^*(\tau_i) \bm{z}(\omega_m).
\end{align*}
i.e. \begin{align*}
	 \partial_{\omega_m}\mv(\bomega) =  \sum_{i = 1}^s \bar{x}_i \w^*(\tau_i) \bm{z}(\omega_m) \cdot \mathbf{e}_m.
\end{align*}
And \begin{align*}
	 \partial_{\omega_m} (A^{-1}(\bomega)) = A^{-1} (\bomega)[ \partial_{\omega_m} A(\bomega) ] A^{-1}(\bomega)
\end{align*}
For $\partial_{\omega_m} A(\bomega),$ we have \begin{align*}
	 \partial_{\omega_m} A(\bomega) = \W(\bomega)^* \bm{z}(\omega_m) \otimes \e_m^*  + \bm{z}^*(\omega_m)\W(\bomega)  \otimes \e_m - 2\text{Re}[ \w (\omega_m)^*\bm{z}(\omega_m)] \e_m\e_m^*
\end{align*}
that leads to 
\begin{align*}
	-\mv(\bomega)^* \partial_{\omega_m} (A^{-1}(\bomega) \big)\mv(\bomega) &= \bm{W}(\bomega)^* \partial_{\omega_m} A(\bomega) \bm{W}(\bomega)\\
	& =  W(\bomega)^*[\W(\bomega)^* \mz(\omega_m) \otimes \e_m^*  + Z^*(\omega_m)\W(\bomega)  \otimes \e_m ] W(\bomega)  -2\text{Re}[ \w(\omega_m)^*\mz(\omega_m)] \lvert (W(\bomega))_m \rvert^2\\
	& = 2\text{Re}[\text{tr} \big( (\W(\bomega)^* \mz(\omega_m) \otimes \e_m^*  ) \underbrace{W(\bomega) W(\bomega)^* }_{= W(\bomega)^* \otimes W(\bomega) } \big) ]   -2\text{Re}[ \w(\omega_m)^*\mz(\omega_m)] \lvert (W(\bomega))_m \rvert^2 \\
	& = 2\text{Re}\big[ W(\bomega)^* \W(\bomega)^* \mz(\omega_m) \cdot (W(\bomega))_m \big]  -2\text{Re}[ \w(\omega_m)^*\mz(\omega_m)] \lvert (W(\bomega))_m \rvert^2.
\end{align*}
Combining these two terms lead to \begin{align*}
	\partial_{\omega_m} F(\bomega) & = 2\text{Re}\big[  [W(\bomega)^*\W(\bomega)^*\mz(\omega_m)-\y^*\mz(\omega_m) ]  W(\bomega)_m    -  \underbrace{\w(\omega_m)^*\mz(\omega_m)}_{ = 0} \lvert W(\bomega)_m\rvert^2\big] \\
	& =2 \text{Re}\big[ \y^*(\mP - \mI )\mz(\omega_m)\cdot (A^{-1}(\bomega) \W(\bomega)^*\y)_m  \big]
\end{align*}

In particular, we have $\hat{\x}(\bomega):=A(\bomega)^{-1} \W(\bomega)^*  \y$ is exactly the coefficient of best approximation of $\y$ at $\text{span}(\text{Col}(\W(\bomega)) ),$  thus we have \begin{align*}
	 \partial_{\omega_m} 2\mathcal{L}_t(\bomega) =  \text{Re}\big[  \y^*(\mP - \mI )\mZ \cdot \text{diag}(\hat \x)  \big]   
\end{align*}

\section{Proof of Results in Section~\ref{sec-approximate-localization}}
In the whole appendix, we set $c_0,c_1$ as following numerical constants and $c_2,c_3$ as following functions on $\varepsilon,\Delta$:
\begin{align*}c_0&=   \sum_{i=1}^\infty \dfrac{1}{(i-1/2)^4}+ \dfrac{1}{(i+1/2)^4} \leq 16.535, \\
 c_1  &= \sum_{i=1}^\infty \dfrac{2}{i^4}\leq 2.17, \\
c_2(\varepsilon,\Delta)&=  \sum_{i=1}^\infty  \dfrac{2}{(i- \varepsilon/\Delta)^4 }, \\
c_3(\varepsilon,\Delta)&=  \sum_{j = 1}^\infty \big[ \dfrac{1}{(j - \frac{2\varepsilon}{\Delta} )^4}+ \dfrac{1}{(j- \frac{\varepsilon}{\Delta})^4 }\big]\end{align*}
\subsection{Proof of Lemma~\ref{lem-J-bounds}}\label{appendix-proof-J-bound}
We first re-claim Lemma~2.3 with the exact formula of $\nu_1,\nu_2,\nu_3:$ 
 \begin{lemma}	As long as $n\Delta > 2$ and $\lvert \omega_i - \tau_{T(i)}\rvert \leq 1/2n,$ we have \begin{align}
	\inf_{\tau  \in \mS(\mT_t^c)}	\lvert J_{2,t}(\tau) \rvert &\geq   \max_{i\in \mT_t^c}  \lvert J_{2,t}(\tau) \rvert  \geq (1-\frac{c_1}{(n+2)^4\Delta^4} )\lvert x_{t+1
}\rvert, \\ 
	\sup_{\tau \notin \mS(\mT_t^c) }	\lvert J_{2,t}(\tau) \rvert &\leq  \lvert x_{t+1} \rvert \bigg(0.7+\underbrace{ \dfrac{c_0}{(n+2)^4\Delta^4}+\big( 1- \dfrac{c_3(\varepsilon_t,\Delta)}{n^4\Delta^4}\big)^{-1} \dfrac{c_0 c_1}{(n+2)^4\Delta^4 } }_{\nu_2}  \bigg),  \\
	\sup_{-1/2\leq \tau \leq 1/2 }	\lvert J_{1,t}(\tau) \rvert &\leq 	 \underbrace{\bigg[ 1+ (1-\dfrac{c_3(\varepsilon_t,\Delta)}{n^4\Delta^4})^{-1} (1+\dfrac{c_0}{n+2)^4\Delta^4}) \bigg](\dfrac{\pi^2}{3} +\pi^2 \dfrac{c_2(\varepsilon_t,\Delta)}{(n\Delta)^4})}_{\nu_1} \cdot n\varepsilon_{\x,t}. \\
	\sup_{-1/2\leq \tau \leq 1/2 }	\lvert J_{3,t}(\tau) \rvert &\leq \lvert x_{t+1}\rvert   \underbrace{\big( 1- \dfrac{c_3(\varepsilon_t,\Delta)}{n^4\Delta^4}\big)^{-1}   \bigg(    \dfrac{c_0+c_1}{(n+2)^4\Delta^4 }    \bigg) }_{\nu_3}	\end{align}
\end{lemma}  
\begin{proof}[Proof of the Lemma~\ref{lem-J-bounds}]

1. For $J_{2,t}(\tau)$ , we have by Lemma~5 in [Cai2011], \begin{align*}
	 \max_{i\in \mT_t^c}  \lvert J_{2,t}(\tau) \rvert &= 	\max_{i \in \mathcal{T}_t^c}\lvert  \w(\tau_i )^* (\bm{I}-\mP_t) \y \rvert\\
	  & =\lVert  \w(\mathcal{T}_t^c )^* (\bm{I}-\mP_t) \w(\mathcal{T}_t^c ) \bm x(\mathcal{T}_t^c) \rVert_\infty \\
	&\geq  \big[\big(\lVert \w(\mT_t^c)^*\w(\mT_t^c)\big]^{-1}  \rVert_{\infty,\infty}\big)^{-1}  \lVert  \bm{x} (\mT_t^c)  \rVert_\infty . 
	\end{align*}
Now noticing that for a diagonal-dominate symmetric matrix $A$ with $\Delta_i,$ we have $ \lVert A^{-1}\rVert_{\infty,\infty} \leq  \max_i {1}/{\Delta_i(A)}. $
 On the other hand, we have \begin{align*}
 	 \Delta_i( \w(\mT_t^c)^*\w(\mT_t^c)) = 1 - \sum_{k\in \mT_t^c} K_n(\tau_i - \tau_k) \geq 1-\dfrac{c_1}{(n+2)^4\Delta^4} \quad \forall i \in [t],
 \end{align*}
 thus  $ \max_{i\in \mT_t^c}  \lvert J_{2,t}(\tau) \rvert  \geq (1-\frac{c_1}{(n+2)^4\Delta^4} )\lvert x_{t+1
}\rvert.$\\  
For any $\tau \notin \mS(\mT_t^c), $ we have denoting $\mP_t\w(\tau) = \sum_{i\in \mT_t} b_i(\tau) \w(\tau_i),$ then \\
\textbf{1. If $T_{t+1}(i) \in \mT_{t}^c$}:  \begin{align*}
	\lvert  J_{2,t}(\tau)   \rvert & \leq \lvert \sum_{i\in \mT_t^c} x_i \w(\tau)^*\w(\tau_i)\rvert +  \lvert \sum_{i\in \mT_t^c} x_i \sum_{k\in \mT_t} b_k(\tau)\w(\tau_k)^*\w(\tau_i)\rvert \\
	&\leq \lvert x_{t+1} \rvert \lvert \sum_{i\in \mT_t^c} K_n(\tau-\tau_i)\rvert +  \big(\sum_{k\in \mT_t} \lvert b_k(\tau)   \rvert\big) \cdot \max_{k\in \mT_t} \lvert  \sum_{i\in \mT_t^c}x_i \w(\tau_k)^* \w(\tau_i)  \rvert\\
	&\leq \lvert x_{t+1} \rvert \bigg( 0.7 + \dfrac{c_0}{(n+2)^4\Delta^4 } \bigg)+ \lvert x_{t+1}\rvert \big( 1- \dfrac{c_3(\varepsilon_t, \Delta)}{n^4\Delta^4}\big)^{-1} \bigg(  \dfrac{c_1}{(n+2)^4\Delta^4}\bigg)\cdot \bigg( \dfrac{c_0}{(n+2)^4\Delta^4 } \bigg) .\end{align*} 
\textbf{2. Otherwise $T_{t+1}(i) \in \mT_t$:} 
 \begin{align*}
	\lvert  J_{2,t}(\tau)   \rvert & \leq \lvert \sum_{i\in \mT_t^c} x_i \w(\tau)^*\w(\tau_i)\rvert +  \lvert \sum_{i\in \mT_t^c} x_i \sum_{k\in \mT_t} b_k(\tau)\w(\tau_k)^*\w(\tau_i)\rvert \\
	&\leq \lvert x_{t+1} \rvert \lvert \sum_{i\in \mT_t^c} K_n(\tau-\tau_i)\rvert +  \big(\sum_{k\in \mT_t} \lvert b_k(\tau)   \rvert\big) \cdot  \max_{k\in \mT_t} \lvert  \sum_{i\in \mT_t^c}x_i \w(\tau_k)^* \w(\tau_i)  \rvert\\
	&\leq \lvert x_{t+1} \rvert \bigg( \dfrac{c_0}{(n+2)^4\Delta^4 } \bigg)+ \lvert x_{t+1}\rvert \big( 1- \dfrac{c_3(\varepsilon_t,\Delta)}{n^4\Delta^4}\big)^{-1} \bigg(  \dfrac{c_1}{(n+2)^4\Delta^4}\bigg)\cdot \bigg( \dfrac{c_0}{(n+2)^4\Delta^4 } + 0.7 \bigg) \\
	& \leq  \lvert x_{t+1} \rvert \bigg(0.7+\underbrace{ \dfrac{c_0}{(n+2)^4\Delta^4}+\big( 1- \dfrac{c_3(\varepsilon_t,\Delta)}{n^4\Delta^4}\big)^{-1} \dfrac{c_0 c_1}{(n+2)^4\Delta^4 } }_{\nu_2}  \bigg),  \end{align*} 
when $n^4\Delta^4 > \frac{10c_0}{3}$ . Where we have used the following lemma to bound the summation of $b_k$: \begin{lemma}\label{lem-b-bound-full}
	For $\bm{F}_t = [\w(\omega_1),\dots,\w(\omega_t) ]$ and $\bm{v}$ a vector in $\mC^n,$ we have $\bm{P}_t \bm{v} = \sum_{i=1}^t b_i  \w(\omega_i),$ with \begin{align*}
		 \lVert \bm{b}\rVert_\infty &\leq \big( 1- \dfrac{c_3(\varepsilon_t,\Delta)}{n^4\Delta^4}\big)^{-1} \lVert  \bm{F}_t \bm{v} \rVert_\infty \\
		 \lVert \bm{b}\rVert_1 &\leq \big( 1- \dfrac{c_3(\varepsilon_t,\Delta)}{n^4\Delta^4}\big)^{-1} \lVert  \bm{F}_t \bm{v} \rVert_1 , 
	\end{align*}
	 as long as $\varepsilon_t \leq \dfrac{1}{2n+4}.$
\end{lemma}
\noindent and the fact
 $$ \lvert \sum_{i\in \mT_t}b_i(\tau) \rvert \leq \big( 1- \dfrac{c_3(\varepsilon_t,\Delta)}{n^4\Delta^4}\big)^{-1} \lvert \sum_{i\in \mT_t} K_n(\tau-\tau_i)\rvert. $$ \\
Then the bound for $J_2$ holds by the bound when $T_{t+1}(i) \in \mT_t$ is larger than its bound when $T_{t+1}(i)\in \mT_t^c$.

\noindent 2. For $J_{3,t}(\tau),$ we have for any $\tau \in [0,1),$ denoting $\tilde{\mP}_t\w(\tau) = \sum_{i\in \mT_t} \tilde{b}_i(\tau)\w(\omega_i),$ then similar as in previous argument, we have \begin{align*}
	\lvert J_{3,t}(\tau)\rvert & \leq \lvert \sum_{i\in \mT_t^c} x_i \sum_{k\in \mT_t} \tilde{b}_k(\tau)\w(\omega_k)^*\w(\tau_i)\rvert + \lvert \sum_{i\in \mT_t^c} x_i \sum_{k\in \mT_t} b_k(\tau)\w(\tau_k)^*\w(\tau_i)\rvert\\
	&\leq \lvert x_{t+1}\rvert \max_{i\in \mT_t^c} \sum_{k\in \mT_t} \lvert\tilde{b}_k(\tau) K_n(\tau_i - \omega_k) \rvert + \lvert x_{t+1}\rvert \max_{i\in \mT_t^c} \sum_{k\in \mT_t} \lvert {b}_k(\tau) K_n(\tau_i - \tau_k) \rvert     \\
	&\leq \lvert x_{t+1}\rvert   \underbrace{\big( 1- \dfrac{c_3(\varepsilon_t,\Delta)}{n^4\Delta^4}\big)^{-1}   \bigg(    \dfrac{c_0+c_1}{(n+2)^4\Delta^4 }    \bigg) }_{\nu_3}
\end{align*} 
3.
For $J_{1,t}(\tau),$ for all $\tau \in [0,1),$ we have \begin{align*}
	\lvert J_{1,t}(\tau) \rvert &=\lvert \sum_{i\in \mT_t} x_i \w(\tau)^*(\mI-\tilde{\mP}_t) (\w(\tau_i) - \w(\omega_i)  ) \rvert \\
	&\leq  \big\lvert \sum_{i\in \mT_t} x_i [K_n(\tau-\tau_i)-K_n(\tau-\omega_i)]  \big \rvert  +\lvert \sum_{k\in\mT_t} \tilde{b}_k(\tau) \rvert \max_{k\in \mT_t} \lvert \sum_{i\in \mT_t} x_iK_n(\omega_k-\tau_i) - K_n(\omega_k - \omega_i)\rvert \\
	&\leq  \bigg(1+ \sum_{k \in \mT_t}\lvert \tilde{b}_k(\tau)\rvert \bigg) \sup_{\tau \in [0,1)} \big \lvert \sum_{i\in \mT_t}  x_i[ K_n(\tau-\tau_i) - K_n(\tau-\omega_i)  ] \big \rvert\\
	&\leq  \bigg(1+ \sum_{k \in \mT_t}\lvert \tilde{b}_k(\tau)\rvert \bigg) \sup_{\tau \in [0,1)}  \sum_{i\in \mT_t}  \big \lvert x_i(\tau_i - \omega_i) K'_n(\tau- \xi_i)   \big \rvert,
\end{align*}
Noticing that \begin{align*}
	 \sum_{k\in \mT_t} \lvert \tilde{b}_k(\tau) \rvert & \leq \big( 1- \dfrac{c_3(\varepsilon_t,\Delta)}{n^4\Delta^4}\big)^{-1}  \sum_{i\in \mT_t}  \lvert K_n(\tau-\omega_i)\rvert \\
	 &\leq  \big( 1- \dfrac{c_3(\varepsilon_t,\Delta)}{n^4\Delta^4}\big)^{-1} \sum_{i\in \mT_t} \dfrac{1}{\max\{1, (n+2)^4(\tau-\omega_i)^4 \}}\\
	 & \leq (1-\dfrac{c_3(\varepsilon_t,\Delta)}{n^4\Delta^4} )^{-1} \big(1+ \dfrac{c_0}{(n+2)^4\Delta^4} \big)
\end{align*}
On the other hand, we have denoting $\mT_t(\tau):=\{i\in \mT_t, \lvert \tau_i - \tau\rvert < \Delta \}$, then by $$\lvert \mT_t(\tau) \rvert\leq 2,\quad \lvert  K_n'(\tau)\rvert \leq \min\{ \dfrac{\pi^2}{3}(n+2)  , \cdot \dfrac{\pi^2}{(n+2)^3\lvert \tau \rvert ^4 } \}  $$ \begin{align*}
	 \lvert \sum_{i\in \mT_t} x_i(\tau_i-\omega_i) K_n'(\tau-\xi_i)\rvert &\leq   \max_{i\in \mT_t}\lvert x_i(\tau_i-\omega_i)\rvert \cdot \big( \sum_{i \in \mT_t(\tau)}  + \sum_{i\notin  \mT_t(\tau)} \lvert K_n'(\tau - \tau_i)\rvert    \big)\\
	 &\leq \max_{i\in \mT_t} \lvert x_i(\omega_i -\tau_i) \rvert \cdot \big( \dfrac{\pi^2}{3} (n+2)  + \pi^2 (n+2) (\sum_{i\neq i(\tau)}\dfrac{1}{n^4(\xi_i - \tau )^4}  )  \big) \\
	 &\leq \max_{i\in \mT_t} \lvert n x_i(\omega_i - \tau_i) \rvert \cdot \big( \dfrac{\pi^2}{3}+ \pi^2  \sum_{i=1}^\infty  \dfrac{2}{n^4(i\Delta- \varepsilon_t)^4 }   \big)\\
	 &\leq n{\varepsilon}_t \bigg(\frac{\pi^2}{3} +\pi^2 \dfrac{c_2(\varepsilon_t,\Delta)}{(n\Delta)^4}\bigg) .
\end{align*}  
That leads to \begin{align*}
	\lvert J_{1,t}(\tau)\rvert &\leq \underbrace{\bigg[ 1+ (1-\dfrac{c_3(\varepsilon_t,\Delta)}{n^4\Delta^4})^{-1} (1+\dfrac{c_0}{n+2)^4\Delta^4}) \bigg](\dfrac{\pi^2}{3} +\pi^2 \dfrac{c_2(\varepsilon_t,\Delta)}{(n\Delta)^4})}_{\nu_1} \cdot n\varepsilon_{\x,t}. 
\end{align*}
\end{proof}

\subsection{Proof of Lemma~\ref{lem-J-small-bound}}\label{appendix-J-small-bound}
\begin{lemma} We have \begin{align*}
	\sup_{\tau \in \mS(\mT_t^c)} \lvert J_{1,t}(\tau) \rvert &\leq n\varepsilon_{\x,t} \cdot  \underbrace{ \dfrac{c_2(\varepsilon_t+\frac{1}{2n},\Delta)}{(n+2)^4\Delta^4} \bigg[ \pi^2 + \big(\dfrac{\pi^2}{3} + \dfrac{\pi^2c_2(\varepsilon_t,\Delta)}{(n\Delta)^4} \big)\cdot \big( 1- \dfrac{c_3(\varepsilon_t,\Delta)}{n^4\Delta^4}\big)^{-1}  \big) \bigg] }_{\kappa_1} . \\
	\sup_{\tau \in \mS(\mT_t^c)} \lvert J_{3,t}(\tau) \rvert &\leq  \lvert x_{t+1}\rvert  \cdot \underbrace{\big( 1- \dfrac{c_3(\varepsilon_t ,\Delta)}{n^4\Delta^4}\big)^{-1}   \bigg(    \dfrac{c_2(\varepsilon_t,\Delta) c_2(\varepsilon_t+\frac{1}{2n},\Delta)}{(n+2)^8\Delta^8 } + \dfrac{c_2(\frac{1}{2n},\Delta) c_1}{(n+2)^8\Delta^8 }  \bigg)}_{\kappa_3} .
\end{align*}
	And \begin{align*}
		\sup_{\omega\in \mS(\tau_i)} \lvert  J_{2,t}(\omega) \rvert &\leq \lvert x_m \rvert + \lvert x_{t+1} \rvert \underbrace{\big(\dfrac{c_2(\frac{1}{2n},\Delta)}{(n+2)^4\Delta^4 } +\big(1-\dfrac{c_3(\varepsilon_t,\Delta)}{(n+2)^4\Delta^4}  \big)^{-1} \dfrac{c_1 c_2(\frac{1}{2n}+ \varepsilon_t,\Delta) }{(n+2)^8\Delta^8 }  \big)}_{\kappa_{21}}, \\ 
		 \lvert J_{2,t}(\tau_i) \rvert &\geq  \lvert x_{i}\rvert -\lvert x_{t+1}\rvert \underbrace{ \big[1+ \big( 1- \dfrac{c_3(\varepsilon_t,\Delta) }{n^4\Delta^4}\big)^{-1} \dfrac{c_2(\varepsilon_t,\Delta)}{(n+2)^4\Delta^4}   \big]  \dfrac{c_1}{(n+2)^4\Delta^4}}_{\kappa_{22}} .
	\end{align*}
	when $\tau_i \in  \mT_t^c.$
\end{lemma}
\begin{proof}
\textbf{1. $J_2$ lower bound:} When $\tau_m \in \mT_t^c,$ we have denoting $\mP_t\w(\tau) = \sum_{i\in \mT_t} b_i(\tau) \w(\tau_i),$ then 
  \begin{align*}
	\lvert  J_{2,t}(\tau_m)   \rvert &  \geq \lvert x_{m}\rvert -  \lvert \sum_{i\in \mT_t^c,i\neq m} x_i \w(\tau_m)^*\w(\tau_i)\rvert -  \lvert \sum_{i\in \mT_t^c} x_i \sum_{k\in \mT_t} b_k(\tau_m)\w(\tau_k)^*\w(\tau_i)\rvert\\
	&\geq\lvert x_m \rvert -\lvert x_{t+1}  \rvert   \sum_{i\neq m} \lvert K_n(\tau_m - \tau_i) \rvert  - \lvert \sum_{k\in \mT_t} b_k(\tau_m) \rvert  \cdot \lvert x_{t+1}\rvert \cdot \max_{k\in \mT_t} \lvert \sum_{i\in \mT_t^c} K_n(\tau_k-\tau_i) \rvert\\
	&\geq \lvert x_m \rvert - \lvert x_{t+1}\rvert  \underbrace{ \big[1+ \big( 1- \dfrac{c_3(\varepsilon_t,\Delta) }{n^4\Delta^4}\big)^{-1} \dfrac{c_2(\varepsilon_t,\Delta)}{(n+2)^4\Delta^4}   \big]  \dfrac{c_1}{(n+2)^4\Delta^4}}_{\kappa_{22}} .  
	\end{align*}
\textbf{2. $J_2$ upper bound: } When $\omega\in \mS(\tau_m), $ we have \begin{align*}
	\lvert J_{2,t}(\omega) \rvert & \leq \lvert x_{m}\rvert+  \lvert \sum_{i\in \mT_t^c,i\neq m} x_i \w(\omega)^*\w(\tau_i)\rvert +  \lvert \sum_{i\in \mT_t^c} x_i \sum_{k\in \mT_t} b_k(\omega)\w(\tau_k)^*\w(\tau_i)\rvert\\
	&\leq \lvert x_m \rvert + \lvert x_{t+1}\rvert \sum_{i\neq m} \lvert K_n(\tau_i - \omega) \rvert + \lvert \sum_{k\in \mT_t} b_k(\omega ) \rvert  \cdot \lvert x_{t+1}\rvert \cdot \max_{k\in \mT_t} \lvert \sum_{i\in \mT_t^c} K_n(\tau_k-\tau_i) \rvert \\
	&\leq \lvert x_m \rvert + \lvert x_{t+1} \rvert  \underbrace{\big(\dfrac{c_2(\frac{1}{2n},\Delta)}{(n+2)^4\Delta^4 } +\big(1-\dfrac{c_3(\varepsilon_t,\Delta)}{(n+2)^4\Delta^4}  \big)^{-1} \dfrac{c_1 c_2(\frac{1}{2n}+ \varepsilon_t,\Delta) }{(n+2)^8\Delta^8 }  \big)}_{\kappa_{21}}
\end{align*} 
\textbf{3. $J_3$ upper bound: } When $\tau \in \mS(\tau_m)$ for some $\tau_m$,    denoting $\tilde{\mP}_t\w(\tau) = \sum_{i\in \mT_t} \tilde{b}_i(\tau)\w(\omega_i),$ then similar as in previous argument, we have \begin{align*}
	\lvert J_{3,t}(\tau)\rvert & \leq \lvert \sum_{i\in \mT_t^c} x_i \sum_{k\in \mT_t} \tilde{b}_k(\tau)\w(\omega_k)^*\w(\tau_i)\rvert + \lvert \sum_{i\in \mT_t^c} x_i \sum_{k\in \mT_t} b_k(\tau)\w(\tau_k)^*\w(\tau_i)\rvert\\
	&\leq \sum_{k\in \mT_t}\lvert \tilde{b}_k(\tau) \rvert \cdot  \lvert x_{t+1}\rvert \max_{k\in \mT_t} \sum_{i\in \mT_t^c} \lvert K_n(\tau_i - \omega_k) \rvert + \sum_{k\in \mT_t}\lvert {b}_k(\tau) \rvert \cdot  \lvert x_{t+1}\rvert \max_{k\in \mT_t} \sum_{i\in \mT_t^c} \lvert K_n(\tau_i - \tau_k) \rvert   \\
	&\leq \lvert x_{t+1}\rvert  \underbrace{ \big( 1- \dfrac{c_3(\varepsilon_t ,\Delta)}{n^4\Delta^4}\big)^{-1}   \bigg(    \dfrac{c_2(\varepsilon_t,\Delta) c_2(\varepsilon_t+\frac{1}{2n},\Delta)}{(n+2)^8\Delta^8 } + \dfrac{c_2(\frac{1}{2n},\Delta) c_1}{(n+2)^8\Delta^8 }  \bigg) }_{\kappa_3}  .
\end{align*} 
\textbf{4. $J_1$ upper bound: }When $\tau \in \mS(\tau_m) $ for some $m \in \mT_t^c,$ we have \begin{align*}
	\lvert J_{1,t}(\tau) \rvert &=\lvert \sum_{i\in \mT_t} x_i \w(\tau)^*(\mI-\tilde{\mP}_t) (\w(\tau_i) - \w(\omega_i)  ) \rvert \\
	&\leq  \big\lvert \sum_{i\in \mT_t} x_i [K_n(\tau-\tau_i)-K_n(\tau-\omega_i)]  \big \rvert  +\lvert \sum_{k\in\mT_t} \tilde{b}_k(\tau) \rvert \max_{k\in \mT_t} \lvert \sum_{i\in \mT_t} x_iK_n(\omega_k-\tau_i) - K_n(\omega_k - \omega_i)\rvert \\
	&\leq \sum_{i\in \mT_t} \lvert  x_i(\tau_i-\omega_i) K_n'(\tau - \tilde{\xi}_{i}) \rvert    + \sum_{k \in \mT_t}\lvert \tilde{b}_k(\tau)\rvert \cdot \sup_{\omega \in [0,1)} \big \lvert \sum_{i\in \mT_t}  x_i[ K_n(\omega-\tau_i) - K_n(\omega-\omega_i)  ] \big \rvert\\
	&\leq   \sum_{i\in \mT_t} \lvert  x_i(\tau_i-\omega_i) K_n'(\tau - \tilde{\xi}_{i}) \rvert    + \sum_{k \in \mT_t}\lvert \tilde{b}_k(\tau)\rvert \cdot \sup_{\tau \in [0,1)}  \sum_{i\in \mT_t}  \big \lvert x_i(\tau_i - \omega_i) K'_n(\omega- \xi_i)   \big \rvert,
\end{align*}
Noticing that \begin{align*}
	 \sum_{k\in \mT_t} \lvert \tilde{b}_k(\tau) \rvert & \leq \big( 1- \dfrac{c_3(\varepsilon_t,\Delta)}{n^4\Delta^4}\big)^{-1}  \sum_{i\in \mT_t}  \lvert K_n(\tau-\omega_i)\rvert \\
	 &\leq  \big( 1- \dfrac{c_3(\varepsilon_t,\Delta)}{n^4\Delta^4}\big)^{-1} \dfrac{c_2(\varepsilon_t+\frac{1}{2n},\Delta) }{(n+2)^4\Delta^4}.\end{align*}
And  $$ \lvert \sum_{i\in \mT_t} x_i(\tau_i-\omega_i) K_n'(\omega-\xi_i)\rvert  \leq (n+2) \varepsilon_{\x,t} \big(\dfrac{\pi^2}{3} + \dfrac{\pi^2c_2(\varepsilon_t,\Delta)}{(n\Delta)^4} \big),  $$
we have the second term is bounded by \begin{align*}
	(n+2) \varepsilon_{\x,t}\big(\dfrac{\pi^2}{3} + \dfrac{\pi^2c_2(\varepsilon_t,\Delta)}{(n\Delta)^4} \big)\cdot \big( 1- \dfrac{c_3(\varepsilon_t,\Delta)}{n^4\Delta^4}\big)^{-1} \dfrac{c_2(\varepsilon_t+\frac{1}{2n},\Delta) }{(n+2)^4\Delta^4} .
\end{align*}
For the first term, we have \begin{align*}
	\sum_{i\in \mT_t} \lvert  x_i(\tau_i-\omega_i) K_n'(\tau - \tilde{\xi}_{i}) \rvert   \leq \varepsilon_{\x,t} \sum_{i\in \mT_t} \lvert K_n'(\tau - \tilde{\xi}_i)\rvert \leq  (n+2) \varepsilon_{\x,t}  \dfrac{\pi^2c_2(\varepsilon_t+\frac{1}{2n},\Delta)}{(n+2)^4\Delta^4}, 
\end{align*}
i.e. \begin{align*}
	\lvert J_{1,t}(\tau) \rvert & \leq (n+2)\varepsilon_{\x,t}  \underbrace{ \dfrac{c_2(\varepsilon_t+\frac{1}{2n},\Delta)}{(n+2)^4\Delta^4} \bigg[ \pi^2 + \big(\dfrac{\pi^2}{3} + \dfrac{\pi^2c_2(\varepsilon_t,\Delta)}{(n\Delta)^4} \big)\cdot \big( 1- \dfrac{c_3(\varepsilon_t,\Delta)}{n^4\Delta^4}\big)^{-1}  \big) \bigg] }_{\kappa_1}.
\end{align*}
	\end{proof}

\section{Proof of Results in Section~\ref{sec-proof-sliding-lanscape}}

\subsection{Proof or Proposition~\ref{thm-sliding-regularity}}
\begin{proof}
  We have \begin{align*}
	    \nabla \mathcal{L}_t(\bomega) & = \text{Re}\big[  \y^* (\mP(\bomega)-\bm{I})  \bm{Z}(\bomega)  \text{diag}(\hat\x ) \big]  .
	\end{align*}
	Thus we have the following upper bound and lower bound of $\lvert \nabla \mathcal{L}_{t,m} (\bomega)\rvert^2$ and $\nabla \mathcal{L}_{t,m} (\bomega)(\omega_m - \tau_m)$:
   \begin{align*}
   	  \lvert \dfrac{\nabla\mathcal{L}_{t,m} (\bomega)}{2} \rvert^2
	&\leq    \lvert \hat{x}_m \rvert^2\big[ \underbrace{\lvert \bar{x}_m \w(\tau_m)^*(\mP(\bomega)-\bm{I})\mz(\omega_m)\rvert^2}_{G_{1,m}^2} +  \underbrace{\lvert  \x_{-m}^*\w_{-m}^* (\mP(\bomega) - \bm{I} ) \mz(\omega_m)  \rvert^2}_{G_{2,m}^2}+ 2 G_{1,m}G_{2,m}\big] ,\\
	  \dfrac{ \nabla \mathcal{L}_{t,m}(\bomega)}{2}(\omega_m - \tau_m)
	 &\geq  \underbrace{\text{Re}\big(\hat{x}_m \bar{x}_m \w (\tau_m)^*(\mP_t-\mI)\mz(\omega_m) \big) (\omega_m - \tau_m)}_{G_{3,m}} -  \lvert \hat{x}_m\rvert \underbrace{\big \lvert  \x_{-m}^*\w_{-m}^* (\mP_t-\mI)\mz(\omega_m)\big) (\omega_m - \tau_m) \big\rvert}_{=G_{2,m}\cdot \lvert \omega_m - \tau_m\rvert }.
\end{align*}
and we aim to control the lower bound of $G_{3,m}$ and upper bounds of $G_{2,m},G_{1,m} $ separately.
Firstly, we have for $\hat{\bm x}= (\bm F^*(\bomega)\bm F(\bomega))^{-1} \bm F(\bomega)^* \bm y$, 
\begin{lemma}\label{lem-coefficient-bound}
	 Denoting  $\bm F_t = \bm F(\bomega)$ for simplicity,  we have then  $$\lvert [  (\bm{F}_t^* \bm{F}_t)^{-1}  \bm{F}_t^* \sum_{i=1}^n x_i \w(\tau_i)]_m - x_m \rvert \leq  \big(1- \dfrac{c_3(\varepsilon_t,\Delta)}{n^4\Delta^4})^{-1} \cdot\big((n+2) \varepsilon_{\x,t} (\dfrac{\pi^2}{3} +\dfrac{\pi^2c_2(\varepsilon_t,\Delta)}{(n\Delta)^4}) + \dfrac{\lVert \x(\mT_t^c)\rVert_\infty   c_2(\varepsilon_t,\Delta)}{(n\Delta)^4}   \big) $$  
	 as long as $ \varepsilon_t  	 \leq \dfrac{1 }{2n+4}  $.
	 \end{lemma} 
On the other hand, we have 
\begin{lemma} \label{lem-G-bound-squared} As long as $ \max_i \lvert  x_{i}(\omega_i - \tau_{i}) \rvert
	 \leq \dfrac{\lVert \x_{\geq t}\rVert_\infty }{2n+4}  $ , we have 
\begin{align*}
	G_{1,m}&\leq   {\varphi_1}   \cdot  n^2\lvert x_m(\omega_m-\tau_m) \rvert ,\\
	G_{2,m}&\leq  {\varphi_{21}}  \cdot n^2\varepsilon_{\x,t}  + n\lVert \x_{>t}\rVert_\infty  \cdot  {\varphi_{22}},\\
	G_{3,m}&\geq  {\varphi_3} n^2\lvert x_m\rvert^2 (\omega_m-\tau_m)^2  .
\end{align*}
with  $$\dfrac{\pi^2}{3} \leq  \varphi_{1} \leq \dfrac{\pi^2}{3}+\dfrac{C_1}{(n\Delta)^8}, \quad \varphi_{21},\varphi_{22} \leq \dfrac{C_2}{(n\Delta)^4}, \quad \dfrac{\pi^2}{3} \geq  \varphi_3 \geq  \dfrac{\pi^2}{3} - C_3[n^2\varepsilon_t^2+\dfrac{1}{(n\Delta)^8} ], $$
where $C_1,C_2,C_3>0$ are universal constants independent of $t$.
\end{lemma} 

Now denoting \begin{align*}
	 \rho(\varepsilon_t,\Delta, \zeta): = \big(1- \dfrac{c_3(\varepsilon_t,\Delta)}{n^4\Delta^4})^{-1} \cdot\big((n+2)\zeta (\dfrac{\pi^2}{3} +\dfrac{\pi^2c_2(\varepsilon_t,\Delta)}{(n\Delta)^4}) + \dfrac{2   c_2(\varepsilon_t,\Delta)}{(n\Delta)^4}   \big), 
\end{align*} and $\rho =\rho (\varepsilon_t,\Delta, \frac{\varepsilon_{\x,t}}{\min_{i\in \mT_t} \lvert x_i\rvert} ) $ for simplicity, we get by $\lVert \x_{>t}\rVert_\infty \leq 2\min_{i\in \mT_t}\lvert x_i \rvert,$ \begin{align*}
	\lvert \dfrac{\nabla\mathcal{L}_{t,m}(\bomega)}{2} \rvert_2^2 &\leq  2(1+\rho)^2 \lvert x_m\rvert^2 \big[ \varphi_1^2 n^4 \lvert x_m(\omega_m-\tau_m)\rvert^2 +2 \varphi_{21}^2  n^4 \varepsilon_{\x,t}^2 + 2n^2\lVert \x_{>t}\rVert_\infty^2\varphi_{22}^2   \big],\\
	{\nabla\mathcal{L}_{t,m}(\bomega)}(\omega_m-\tau_m) &\geq   2 \varphi_{3} n^2\lvert x_m\rvert^2(\omega_m-\tau_m)^2 -   (1+\rho) \cdot(\varphi_{21}n^2\varepsilon_{\x,t}+\varphi_{22} n\lVert \x_{>t}\rVert_\infty )\lvert x_m(\omega_m-\tau_m)\rvert
\end{align*} 
	Thus   \begin{align*}
	&2{ \dfrac{\nabla\mathcal{L}_{t,m}(\bomega)}{2n}(\omega_m-\tau_m) } - \tilde{\mu} \lvert \dfrac{\nabla\mathcal{L}_{t,m}(\bomega)}{2n} \rvert_2^2 \\
	\geq &  2\lvert x_m\rvert^2 \cdot \big[ \varphi_3n(\omega_m-\tau_m)^2 -  [1+\rho]^2 \tilde{\mu}\varphi_1^2 n^2\lvert x_m\rvert^2 (\omega_m-\tau_m)^2       \big]\\
	&  -\bigg( [1+\rho] \varphi_{21}n\varepsilon_{\x,t}\lvert x_m(\omega_m-\tau_m)\rvert +4[1+\rho]^2 \tilde{\mu}\lvert x_m\rvert^2\varphi_{21}^2 n^2\varepsilon_{\x,t}^2  \bigg)   \\
		& - \bigg( \big[1+\rho \big] \varphi_{22}\lVert \x_{>t}\rVert_\infty \lvert x_m(\omega_m - \tau_m) \rvert + 4[1+\rho]^2 \tilde{\mu} \lvert x_m\rvert^2 \lVert \x_{>t}\rVert_\infty^2 \varphi_{22}^2  \bigg),  
\end{align*}
 Thus as long as $\tilde{\mu} $ is small enough so that $$
	n\tilde{\mu}\varphi_1^2[1+\rho]^2    \lvert x_m\rvert^2   \leq  {\varphi_3}/{2},$$ and the above difference turns to be larger than  \begin{align*}
 	& \varphi_3n \lvert x_m \rvert^2  \lvert \omega_m - \tau_m \rvert^2 -\big[ 1+\rho \big]\big(  \varphi_{21}\varepsilon_{\x,t}+  \dfrac{\varphi_{22}}{n}\lVert \x_{>t}\rVert_\infty  \big)\lvert  n x_m(\omega_m-\tau_m)\rvert\\
 	& - \underbrace{4\tilde{\mu}\lvert x_m\rvert^2[1+\rho]^2 }_{\leq {2\varphi_3}/({\varphi_1^2 n)}}\big[ \varphi_{21}^2n^2\varepsilon^2_{\x,t}+\lVert \x_{>t}\rVert_\infty^2\varphi_{22}^2 \big]  ,
 	    	 \end{align*}
 	    	 which can be re-written as \begin{align*}
 	    	 	&\varphi_3 n \bigg( \lvert x_m(\omega_m-\tau_m)\rvert-\dfrac{1+\rho }{2\varphi_3} \cdot \big[ \varphi_{21}\varepsilon_{\x,t}+\dfrac{\varphi_{22}}{n}\lVert \x_{>t}\rVert_\infty \big]    \bigg)^2 -\dfrac{2\varphi_3}{\varphi_1^2n}\big[ \varphi_{21}^2n^2\varepsilon_{\x,t}^2+\lVert \x_{>t}\rVert_\infty^2\varphi_{22}^2 \big]\\
 	    	 	&- \dfrac{n[1+\rho]^2}{4\varphi_3}\big[ \varphi_{21}\varepsilon_{\x,t}+\dfrac{\varphi_{22}}{n}\lVert \x_{>t}\rVert_\infty \big]^2,
 	    	 \end{align*} 
 	  which is larger or equal than\begin{equation} \label{eq-diff-lower-bound} \begin{aligned}
 	  		&\varphi_3 n \bigg( \lvert x_m(\omega_m-\tau_m)\rvert-\dfrac{1+\rho }{2\varphi_3} \cdot \big[ \varphi_{21}\varepsilon_{\x,t}+\dfrac{\varphi_{22}}{n}\lVert \x_{>t}\rVert_\infty \big]    \bigg)^2 \\
 	    	 	&-\varphi_3 n\big[\dfrac{2}{\varphi_1^2}+ \dfrac{[1+\rho]^2}{4\varphi_3^2}\big] \cdot \big[ \varphi_{21}\varepsilon_{\x,t}+\dfrac{\varphi_{22}}{n}\lVert \x_{>t}\rVert_\infty \big]^2,
 	  \end{aligned}  	 \end{equation}
Now under the condition \begin{equation}\label{eq-condition-inequality}
	\big[\dfrac{1}{2}- \dfrac{1+\rho}{2\varphi_3}\varphi_{21} \big]\varepsilon_{\x,t}  > \dfrac{2\varphi_{22}}{n}\lVert \x_{>t}\rVert_\infty
\end{equation}
we have $\lvert x_m (\omega_m-\tau_m)\rvert> \dfrac{1}{2}\varepsilon_{\x,t}$
 implies  \eqref{eq-diff-lower-bound} is lower bounded by \begin{align*}
 	&\varphi_{3}n \bigg( \big[   \big(\dfrac{1}{2}- \dfrac{1+\rho}{2\varphi_3}\varphi_{21} \big)\varepsilon_{\x,t}  - \dfrac{\varphi_{22}}{n}\lVert \x_{>t}\rVert_\infty   \big]^2 -  \big[\dfrac{2}{\varphi_1^2}+ \dfrac{[1+\rho]^2}{4\varphi_3^2}\big] \cdot \big[ \varphi_{21}\varepsilon_{\x,t}+\dfrac{\varphi_{22}}{n}\lVert \x_{>t}\rVert_\infty \big]^2 \bigg) \\
 	\geq & \varphi_3 n \varepsilon_{\x,t}^2 \underbrace{\bigg(\big[\dfrac{1}{4}-\dfrac{1+\rho}{4\varphi_3} \varphi_{21}\big ]^2 - \big[ \dfrac{2}{\varphi_1^2}+ \dfrac{[1+\rho]^2}{4\varphi_3^2} \big] [\varphi_{21}+\dfrac{1}{4}-\dfrac{[1+\rho]}{2\varphi_3} \varphi_{21}  ]^2  \bigg)}_{\tilde{\nu} }. 
 \end{align*} 	    	 
 Thus when $\tilde{\nu}>0,$ we have the regularity condition  is satisfied by  $ \mu_m \nabla\mathcal{L}_{t,m} $ along $m$ with \begin{align*}
 	\mu_m =  \dfrac{\varphi_3}{2n[1+\rho]^2 \lvert x_m\rvert^2\varphi_1^2},\quad \lambda_m \geq \varphi_3 \mu_m n \lvert x_m\rvert^2 \tilde{\nu}\geq \dfrac{\varphi_3^2\tilde{\nu}  }{2\varphi_1^2[1+\rho]^2}   
 \end{align*}
In particular, $\tilde{\nu }$ is lower bounded by \begin{align*}
	 \mathcal{CR}_t =  \big[\dfrac{1}{4}-\dfrac{1+\rho}{4\varphi_3} \varphi_{21}\big ]^2    -   \big(\dfrac{18}{\pi^4}+ \dfrac{[ 1+\rho]^2}{4\varphi_3^2} \big) \ \cdot\big( \dfrac{1}{4}+ \varphi_{21}\big )^2 \big)
\end{align*} 
which is monotone increasing as $n\Delta$ increasing and $\varepsilon$ decreases.

On the other hand, when \eqref{eq-condition-inequality} holds,  we have for every $m$,
\begin{align*}
 		 \lvert {\mu}_m   x_m \dfrac{\nabla\mathcal{L}_{t,m} }{2n} (\bomega) \rvert & \leq \dfrac{\varphi_3}{2n^2 \varphi_1^2[1+\rho]^2} (G_{1,m}+G_{2,m} )  \\
 		 &\leq \dfrac{\varphi_3}{2n^2 \varphi_1^2 [1+\rho]^2}\big( [\varphi_1+\varphi_{21} ] n^2\varepsilon_{\x,t} + \varphi_{22}n \lVert \x_{>t}\rVert_\infty    \big)  \\
 		 & \leq  \dfrac{\varphi_3}{2\varphi_1^2[1+\rho]^2 } \big( \varphi_1+\varphi_{21}+ \big[\dfrac{1}{4}- \dfrac{1+\rho}{4\varphi_3}\varphi_{21} \big]   \big)  \varepsilon_{\x,t} . 	\end{align*} 
 		 which is smaller or equal than \begin{align*}
 		 	\tilde{\beta}:= \dfrac{\varphi_3}{2\varphi_1^2 [1+\rho]^2}\big( \dfrac{\varphi_1}{2}+\varphi_{21}+ \big[\dfrac{1}{4}- \dfrac{1+\rho}{4\varphi_3}\varphi_{21} \big]   \big)  \varepsilon_{\x,t} 
 		 \end{align*}
 when $\lvert x_m(\omega_m-\tau_m)\rvert < \dfrac{1}{2}\varepsilon_{\x,t}.$ Thus the $\text{WRC}_{ \x_{\leq t}} (\bm{\lambda} ,\dfrac{1}{2},\tilde{\beta})$ is satisfied by $\bm \mu \odot \dfrac{\nabla \mathcal{L}_t}{2n}(\cdot) $ over the region \begin{align*}
 	 	B_t: = \{ \underbrace{\big[\dfrac{1}{2}- \dfrac{1+\rho}{2\varphi_3}\varphi_{21} \big]^{-1}\cdot 2\varphi_{22} }_{:= \mathcal{R}_t}   \dfrac{\lVert \x_{>t}\rVert_\infty}{n}    \leq \lVert \bomega - \btau_{\leq t}\rVert_{t,\infty}\leq \dfrac{\varepsilon  \lvert x_t\rvert}{2n}  \}
 \end{align*}

While $\mu_m$ contains the orcale information $x_m$, we can choose $$  \hat{\mu}_m =   \dfrac{\varphi_3[1-\rho]^2 }{2n [1+\rho]^2  \lvert \hat{x}_m\rvert^2\varphi_1^2},$$, then by $[1-\rho] \lvert \x_{\leq t}\rvert \ge \lvert \hat{\x}_{\leq t} \rvert \geq [1+\rho]\lvert \x_{\leq t}\rvert$, we have $$   \dfrac{[1-\rho]^2}{[1+\rho]^2} \bm \mu  \leq  \hat{\bm\mu} \leq \bm \mu,$$ thus the regularity condition holds with $\bm g(\cdot) = \hat{\bm \mu} \odot \dfrac{\nabla \mathcal{L}_t}{2n}(\cdot)$ and some $\bm\lambda >0.$ 

Finally, by \begin{align*}
	 \lvert x_mg_m\rvert &\leq \lvert x_m\mu_m \dfrac{(\nabla \mL_t)_m}{2n}\rvert \\
	 & \leq \dfrac{\varphi_3}{2n^2 \varphi_1^2[1+\rho]^2} \lvert G_{1,m}+G_{2,m} \rvert \\
	 &\leq \underbrace{ \dfrac{\varphi_3}{2 \varphi_1^2 [1+\rho]^2} [\varphi_1+\varphi_{21} ]}_{\phi_1} \varepsilon_{\x,t} + \underbrace{\dfrac{\varphi_{22} \varphi_3}{2 \varphi_1^2 [1+\rho]^2}}_{\phi_2} \frac{\lVert \x_{>t}\rVert_\infty}{n} ,   \end{align*}
	  the claim holds.
\end{proof}

\subsection{Proof of Lemma~\ref{lem-G-bound-squared}}

We first reclaim the lemma with explicit formula of $\varphi_{1},\varphi_{21},\varphi_{22},\varphi_{3}. $ 

\begin{lemma} As long as $ \max_i \lvert  x_{i}(\omega_i - \tau_{i}) \rvert
	 \leq \dfrac{\lVert \x_{\geq t}\rVert_\infty }{2n}  $ , we have 
\begin{align*}
	G_{1,m}&\leq \underbrace{\bigg(\dfrac{\pi^2}{3}+ 2\pi^4  \big(1- \dfrac{c_3(\varepsilon_t,\Delta)}{n^4\Delta^4})^{-1} \dfrac{c_3^2(\varepsilon_t,\Delta)}{(n\Delta)^8} \bigg)}_{\varphi_1}   \cdot  n^2\lvert x_m(\omega_m-\tau_m) \rvert ,\\
	G_{2,m}&\leq  \underbrace{2\pi^4\dfrac{c_3(\varepsilon_t,\Delta)}{n^4\Delta^4}  \bigg[ 2 + (1-\dfrac{c_3(\varepsilon_t,\Delta)}{n^4\Delta^4})^{-1} \big(1+\dfrac{c_2(\varepsilon_t,\Delta)} {(n\Delta)^4} \big)   \bigg]}_{\varphi_{21}}  \cdot n^2\varepsilon_{\x,t} \\
&\quad + \lVert \x_{>t}\rVert_\infty  \cdot \underbrace{\pi^2\bigg[(1- \dfrac{c_3(\varepsilon_t,\Delta)}{n^4\Delta^4})^{-1} \dfrac{c_3^2(\varepsilon_t,\Delta)}{(n\Delta)^8} +  \dfrac{c_3(\varepsilon_t,\Delta) }{n^4\Delta^4}   \bigg]  }_{\varphi_{22}},\\
	G_{3,m}&\geq \underbrace{\bigg([1-\rho]\big[\dfrac{\pi^2}{3}- \dfrac{\pi^4}{18} (n+2)^2  \varepsilon_t^2\big]  - 2\pi^4  \big(1- \dfrac{c_3(\varepsilon_t,\Delta)}{n^4\Delta^4})^{-1} \dfrac{c_3^2(\varepsilon_t,\Delta)}{(n\Delta)^8} \bigg)}_{\varphi_3}n^2\lvert x_m\rvert^2 (\omega_m-\tau_m)^2  .
\end{align*}
\end{lemma}

\begin{proof}\quad\\
\textbf{1. Lower Bound of $G_{3,m}$:}\\	
 Noticing that \begin{align*}
	\bar{x}_m \w(\tau_m)^* (\mP_t - \bm{I}) \mz(\omega_m)   & = \bar{x}_m [ \w(\tau_m)^* \mP_t \mz(\omega_m)  -K_n' ( \omega_m - \tau_m)       ] ,\\
	\hat{x}_m \bar{x}_m &= [ (\bm{F}_t^* \bm{F}_t)^{-1} \bm{F}_t^*  \y ] _m \bar{x}_m \\
	&=  \sum_{i=1}^m x_i\bar{x}_m [  (\bm{F}_t^* \bm{F}_t^*)^{-1}  \bm{F}_t^* \w(\tau_i)    ]_m  \\
	& =  \bar{x}_m \e_m^* (\W^*_t\W_t)^{-1}  \bm{F}_t^* [\sum_{i\leq t}+ \sum_{i> t} x_i \w(\tau_i) ] .
\end{align*}
By Lemma~\ref{lem-coefficient-bound}, we have $	\big\lvert \hat{x}_m \bar{x}_m - \lvert x_m\rvert^2 \big\rvert  \leq  \rho , $
    thus $\text{Re}(\hat{x}_m \bar{x}_m) \geq\big(1-\rho \big) \lvert x_m \rvert^2 . $
On the other hand, denote $\bm{P}_t \bm{z}(\omega_m) = \sum_{i\neq m,i\in [t]} a_{i,m} \w(\omega_i),$ then \allowdisplaybreaks \begin{align*}
\lvert \w(\tau_m)^* \mP(\bomega) \mz(\omega_m) \rvert  & = \lvert [\w(\tau_m)-\w(\omega_m)] ^* \mP(\bomega) \mz(\omega_m) \rvert  \\
&\leq \lVert \bm{a}_{m,-m}\rVert_1 \cdot \max_{i\neq m}\lvert  [\w(\tau_m)-\w(\omega_m)] ^* \w(\omega_i)  \rvert \\
&   \leq \big(1- \dfrac{c_3(\varepsilon_t,\Delta)}{n^4\Delta^4})^{-1} \cdot \sum_{i\neq m}\lvert K_n'(\omega_m - \omega_i) \rvert \cdot \max_{i\neq m} \lvert (\omega_m-\tau_m) K_n'(\omega_i-\xi_m)\rvert \\
& \leq \big(1- \dfrac{c_3(\varepsilon_t,\Delta)}{n^4\Delta^4})^{-1} \cdot \sum_{i\neq m}\lvert K_n'(\omega_m - \omega_i) \rvert \cdot \max_{i\neq m} \lvert (\omega_m-\tau_m) K_n'(\omega_i-\xi_m)\rvert \\
&\leq 2\pi^4  \big(1- \dfrac{c_3(\varepsilon_t,\Delta)}{n^4\Delta^4})^{-1} \dfrac{c_3^2(\varepsilon_t,\Delta)}{(n\Delta)^8} \cdot n^2 \lvert \omega_m - \tau_m\rvert
\end{align*}
where we have used  $ \sum_{i\neq m }\lvert K_n'(\omega_m-\omega_i)\rvert \leq \pi^2n\cdot\dfrac{c_3(\varepsilon_t,\Delta)}{n^4\Delta^4}$ and $\lvert K_n'(\tau) \rvert \leq  \dfrac{\pi^2}{(n+2)^3\tau^4}  $ in the last line.  
Finally, we have \begin{align*}
	-(\omega_m - \tau_m)K_n'(\omega_m-\tau_m) & = -(\omega_m-\tau_m) \int_{0}^{\omega_m - \tau_m} K_n''(v) dv\\
	& \geq  (\omega_m-\tau_m)\int_0^{\omega_m-\tau_m} \dfrac{\pi^2}{3}n(n+4)-\dfrac{\pi^4}{6}(n+2)^4 v^2 dv\\
	&\geq \dfrac{\pi^2}{3}n(n+4)(\omega_m-\tau_m)^2-\dfrac{\pi^4}{18} (n+2)^4(\omega_m-\tau_m)^4\\
	&\geq  \big[\dfrac{\pi^2}{3}- \dfrac{\pi^4}{18} (n+2)^2  \varepsilon_t^2\big]  n^2 (\omega_m-\tau_m)^2 
\end{align*}

Combining these bounds, we get,
\begin{align*}
	G_{3,m}  &\geq  \text{Re}\big( \hat{x}_m\bar{x}_m  \big)\big ( - K_n'(\omega_m - \tau_m)\big ) (\omega_m - \tau_m) -  \lvert x_m^2( \omega_m-\tau_m)\rvert \cdot \lvert \w(\tau_m)^*\mP(\bomega)\mz(\omega_m)\rvert\\
	& \geq \underbrace{\bigg([1-\rho]\big[\dfrac{\pi^2}{3}- \dfrac{\pi^4}{18} (n+2)^2  \varepsilon_t^2\big]  - 2\pi^4  \big(1- \dfrac{c_3(\varepsilon_t,\Delta)}{n^4\Delta^4})^{-1} \dfrac{c_3^2(\varepsilon_t,\Delta)}{(n\Delta)^8} \bigg)}_{\varphi_3} n^2\lvert x_m\rvert^2 (\omega_m-\tau_m)^2   ,
\end{align*}
as desired.

\noindent\textbf{2. Upper bound of $G_{1,m}$:} By previous bounds, we have \begin{align*}
	G_{1,m} &\leq \lvert x_m K'_n(\omega_m - \tau_m) \rvert    +\lvert x_m \w(\tau_m)^* \mP(\bomega) \mz(\omega_m) \rvert\\
	& \leq \lvert x_m\rvert\bigg[\int_{0}^{\lvert \omega_m-\tau_m\rvert}\lvert K_n''(v)\rvert dv + 2\pi^4  \big(1- \dfrac{c_3(\varepsilon_t,\Delta)}{n^4\Delta^4})^{-1} \dfrac{c_3^2(\varepsilon_t,\Delta)}{(n\Delta)^8} \cdot n^2 \lvert \omega_m - \tau_m\rvert \bigg] \\
	&\leq \underbrace{\bigg(\dfrac{\pi^2}{3}+ 2\pi^4  \big(1- \dfrac{c_3(\varepsilon_t,\Delta)}{n^4\Delta^4})^{-1} \dfrac{c_3^2(\varepsilon_t,\Delta)}{(n\Delta)^8} \bigg)}_{\varphi_1}   \cdot  n^2\lvert x_m(\omega_m-\tau_m) \rvert 
\end{align*}

\noindent\textbf{3. Upper bound of $G_{2,m}$:}  
We would divide $G_{2,m}$ into two parts and bound them separately: \begin{align*}
	G_{2,m} &\leq   \lvert\sum_{i\leq t,i\neq m}  x_i \w (\tau_i)^* (\mP(\bomega)-\bm{I}) \mz(\omega_m) \rvert +   \lvert  \x_{>t}^* \w_{>t}^* (\mP(\bomega)-\bm{I}) \mz(\omega_m) \rvert 
\end{align*}
For the first part, we have
\begin{align*}
	  \lvert\sum_{i\leq t,i\neq m}  x_i \w (\tau_i)^* (\mP(\bomega)-\bm{I}) \mz(\omega_m) \rvert & = \lvert \sum_{i\leq t,i\neq m}  x_i [\w(\tau_i) - \w(\omega_i)]^*(\mP(\bomega)-\bm{I}) \mz(\omega_m) \rvert \\
	  & = \lvert \sum_{i\leq t,i\neq m}  x_i [\w(\tau_i) - \w(\omega_i)]^*(\mP(\bomega)-\bm{I}) \mz(\omega_m)  \rvert \\
	  &\leq \sum_{i\leq t,i\neq m} \lvert  x_i\rvert \big[ \sum_{k \leq t, k\neq m}a_{k,m}  \lvert( \w(\tau_i)-\w(\omega_i))^*\w(\omega_k) \rvert     + \lvert [\w(\tau_i) - \w(\omega_i)]^* \mz(\omega_m)  \rvert  \big]  \\
	  & \leq \lVert \bm{a}_{-m,m}\rVert_1   \max_{k\leq t,k\neq m} \big\lvert \sum_{i\leq t, i\neq m} \lvert x_i \rvert \cdot \lvert K_n(\tau_i-\omega_k) - K_n(\omega_i - \omega_k)\rvert \big\rvert\\
	  &\quad +\sum_{i<t,i\neq m} \lvert x_i [K_n'(\tau_i-\omega_m) - K_n'(\omega_i-\omega_m)]\rvert .
	 \end{align*} 
	 In particular, the first term can be bounded by $$  \big(1- \dfrac{c_3(\varepsilon_t,\Delta)}{n^4\Delta^4})^{-1} \dfrac{c_3(\varepsilon_t,\Delta)}{(n\Delta)^4} \cdot  2\pi^4 n^2\varepsilon_{\x,t}\big(1+\dfrac{c_2(\varepsilon_t,\Delta)} {(n\Delta)^4} \big) $$
	 as we do when bounding $G_{3,m}.$ For the second term, we have \begin{align*}
	 	\sum_{i<t,i\neq m} \lvert x_i[K_n'(\tau_i - \omega_m) - K_n'(\omega_i - \omega_m) ] \rvert & \leq \sum_{i<t,i\neq m} \lvert x_i(\tau_i -\omega_i) K_n''(\xi_i - \omega_m)\rvert \\
	 	&\leq \varepsilon_{\x,t} \sum_{i<t , i\neq m} \lvert K_n''(\xi_i -\omega_m) \rvert\\
	 	&\leq \varepsilon_{\x,t} \dfrac{4\pi^4n^2  c_{3}(\varepsilon_t,\Delta) }{(n\Delta)^4}
	 \end{align*}
Thus the first part  is upper bounded by \begin{align*}
	\underbrace{2\pi^4\dfrac{c_3(\varepsilon_t,\Delta)}{n^4\Delta^4}  \bigg[ 2 + (1-\dfrac{c_3(\varepsilon_t,\Delta)}{n^4\Delta^4})^{-1} \big(1+\dfrac{c_2(\varepsilon_t,\Delta)} {(n\Delta)^4} \big)   \bigg]}_{\varphi_{21}}  \cdot n^2\varepsilon_{\x,t}
\end{align*}

	 And for the remaining part,  we have 
\begin{align*}
    \lvert  \x_{>t}^* \w_{>t}^* (\mP(\bomega)-\bm{I}) \mz(\omega_m) \rvert & \leq \sum_{i>t} \lvert x_i \rvert\big[ \sum_{\ell\neq m ,\ell \leq t}\lvert a_{\ell,m} \w(\omega_i)^* \w(\omega_\ell)\rvert     +\lvert    \w(\omega_i)^* \mz(\omega_m) \rvert\big]\\
    &\leq \sum_{\ell\neq m,\ell\leq t} \lvert a_{\ell,m}\rvert\big[ \sum_{i>t }  \lvert x_i \rvert  \lvert \w(\omega_i)^* \w(\omega_\ell)\rvert \big] + \sum_{i>t} \lvert x_i \rvert \lvert \w(\omega_i)^*\mz(\omega_m) \rvert \\
    &\leq \lVert \bm{a}_{-m,m}\rVert_1 \lVert \x_{>t}\rVert_\infty \max_{\ell\leq t}\lvert \sum_{i>t} \w(\omega_i)^*\w(\omega_\ell)\rvert + \lVert \x_{>t}\rVert_\infty  \sum_{i>t} \lvert \w (\omega_i)^*\mz(\omega_m)\rvert  \\
    &\leq  n\lVert \x_{>t}\rVert_\infty  \cdot \underbrace{\pi^2\bigg[(1- \dfrac{c_3(\varepsilon_t,\Delta)}{n^4\Delta^4})^{-1} \dfrac{c_3^2(\varepsilon_t,\Delta)}{(n\Delta)^8} +  \dfrac{c_3(\varepsilon_t,\Delta) }{n^4\Delta^4}   \bigg]  }_{\varphi_{22}}
\end{align*}
\noindent Combining all bounds above  together, we get the desired bound.
\end{proof}

\section{Proof of Results in Section~\ref{sec-analysis-incomplete-sample} }\label{appendix-analysis-incomplete}
\subsection{Proof of Proposition~\ref{prop-subsampled-approximate-localization} }
We would first reclaim the formula of $\tilde{\mathcal{C}}_t$ and $\tilde{\lambda}_t:$ 
There exists $\tilde{\nu}_1,\tilde{\nu}_2,\tilde{\nu}_3$ and  $\tilde{\kappa}_{1},\tilde{\kappa}_{21},\tilde{\kappa}_{22},\tilde{\kappa}_{3}$ satisfying  \begin{align*}
	\tilde{\nu}_1 &= O( 1+ \dfrac{1}{(n\Delta)^4}+ s \sqrt{\log(n)/np} ),\\
	\tilde{\nu}_2 &= O( \dfrac{1}{(n\Delta)^4} +s\sqrt{\log(n)/np} ),\\
	\tilde{\nu}_3 &=\big(1 - ( \dfrac{1}{(n\Delta)^4} +s\sqrt{\log(n)/np} )\big)^{-1}  \cdot O( \dfrac{1}{(n\Delta)^4} +s\sqrt{\log(n)/np} ),\\
	\tilde{\kappa}_1 &=( \dfrac{1}{(n\Delta)^4} +s\sqrt{\log(n)/np} )\\
	\tilde{\kappa}_{21} &=( \dfrac{1}{(n\Delta)^4} +s\sqrt{\log(n)/np} )\\
	\tilde{\kappa}_{22} &=( \dfrac{1}{(n\Delta)^4} +s\sqrt{\log(n)/np} )\\
	\tilde{\kappa}_3 &=( \dfrac{1}{(n\Delta)^4} +s\sqrt{\log(n)/np} )
	\end{align*}
	so that \begin{align*}
		 \tilde{\mathcal C}_t &=  0.3 - (\tilde \nu_2 + 2\tilde\nu_3 + \dfrac{c_1}{n^4\Delta^4}+ cs\sqrt{\frac{\log n}{np}} +2 \dfrac{n\varepsilon_{\x, t}}{\lvert x_{t+1}\rvert} \tilde \nu_1),   \\
		 & = 0.3 - O\big( \dfrac{1}{(n\Delta)^4}+ s\sqrt{\log(n)/np}+\dfrac{n\varepsilon_{\x,t}}{\lvert x_{t+1}\rvert} \big)  \\
		 \tilde{\lambda}_1 & = 1- O \big(\dfrac{1}{(n\Delta)^4}+ s\sqrt{\log(n)/np} \big)    \\
		 \tilde{\lambda}_t &=   \bigg[1 - \big( 2 \dfrac{n\varepsilon_{\x,t}}{\lvert x_{t+1}\rvert} \tilde{\kappa}_1+\tilde{\kappa}_{21}+\tilde{\kappa}_{22}+2\tilde{\kappa}_3 \big) \bigg] \lvert x_{t+1}\rvert\\
		 & = \bigg[ 1-O\big(\dfrac{n\varepsilon_{\x,t}}{\lvert x_{t+1}\rvert} [\dfrac{1}{(n\Delta)^4}+ s\sqrt{\log(n)/np} ]  \big) \bigg], \quad t >1.    
	\end{align*}
\begin{proof}
\textbf{The first iteration:} When $t = 0$  W.L.O.G. assume that $\lvert x_1 \rvert = \lVert \x\rVert_\infty,$  we have if $\omega_{1} \notin \mS([s])$, then with probability at least $1-1/n^2,$  \begin{align*}
	\lvert \w(\omega_1)^*\y\rvert  &\leq \lvert  \sum_{i = 1}^s x_i \tilde{K}_n(\tau_i - \omega_1) \rvert\\
	&\leq 0.7 {\lvert x_{T(1)}\rvert}p(1 + \sqrt{c\log(n)/(np)}  ) + \sum_{i\neq T(1)}^s \lvert x_i\rvert p( \dfrac{1}{
	n^4\lvert \omega_1 - \tau_i\rvert^4}+ \sqrt{c\log(n)/(np)}  ) \\
	&\leq p \lvert x_1\rvert \bigg( 0.7 + \dfrac{c_1}{(n\Delta)^4 } + c \cdot s\sqrt{\log(n)/np} \bigg) \end{align*}
On the other hand, we have \begin{align*}
	\lvert \w(\tau_1)^* \y\rvert \geq  p\lvert x_1\rvert \bigg(1 - \dfrac{c_1}{(n\Delta)^4 } -  c\cdot s\sqrt{\log(n)/np} \bigg) ,\end{align*}
that leads to \begin{align*}
\dfrac{2c}{(2n)^2\Delta^2}+  c\cdot s\sqrt{\log(n)/np} < 0.3 \implies \lvert \w(\tau_1)^*\y \rvert > \lvert \w(\omega_1)^* \y\rvert,
\end{align*}
a contradiction.\\
Thus $\omega_1\in \mS([s]). $\\
Now suppose  $\mS(\tau_{T(1)}),$ then we have \begin{align*}
	 p \lvert x_{T(1)}\rvert(1+ s\sqrt{\log(n)/np} )   +p\lvert x_1\rvert (\dfrac{c_1}{(n\Delta )^4 } + cs\sqrt{\log(n)/np})\geq p\lvert x_1\rvert \bigg( 1- \dfrac{c_1}{(n\Delta)^4 } -  cs\sqrt{\log(n)/np}\bigg),
\end{align*} 
which leads to \begin{align*}
	\lvert x_{T(1)} \rvert  \geq \lvert x_1 \rvert  \bigg( 1-\underbrace{ \dfrac{2c_1}{(n\Delta)^4 } - c\cdot s\sqrt{\log(n)/np}}_{\tilde{\lambda}_1 } \bigg)  
\end{align*}
That finishes the proof when $t = 0 $. \\
\textbf{The $t$-th iteration:} For $t>0$, we have suppose W.L.O.G. that $x_{t+1} \in \mT_{t}^c$ and $\lvert x_{t+1}\rvert = \lVert \x(\mT_t^c)\rVert_\infty,$ then for any $\tau \in [0,1),$ \begin{align*}
	&\w(\tau)^* (\bm{I}- \tilde{\bm{P}}_t) \y \\
	 =& \w(\tau)^* (\mI - \tilde{\bm{P}}_t) [\w(\mT_t)\x(\mT_t)+\w(\mT_t^c)\x(\mT_t^c) ]\\
	 =& \underbrace{\w(\tau)^* (\mI - \tilde{\bm{P}}_t) \w(\mT_t)\x(\mT_t)}_{\tilde J_{1,t}(\tau)} + \underbrace{\w(\tau)^* (\mI - {\bm{P}}_t)\w(\mT_t^c)\x(\mT_t^c)}_{\tilde J_{2,t}(\tau)} +\underbrace{\w(\tau)^* ({\bm{P}}_t- \tilde{\mP}_t )\w(\mT_t^c)\x(\mT_t^c)}_{\tilde J_{3,t}(\tau)} .
	 \end{align*}
1. For $\tilde J_{2,t}(\tau)$ , we have by Lemma~5 in [Cai2011], \begin{align*}
	 \max_{i\in \mT_t^c}  \lvert J_{2,t}(\tau) \rvert &= 	\max_{i \in \mathcal{T}_t^c}\lvert  \w(\tau_i )^* (\bm{I}-\mP_t) \y \rvert\\
	  & =\lVert  \w(\mathcal{T}_t^c )^* (\bm{I}-\mP_t) \w(\mathcal{T}_t^c ) \bm x(\mathcal{T}_t^c) \rVert_\infty \\
	&\geq   \big[\big(\lVert \w(\mT_t^c)^*\w(\mT_t^c)\big]^{-1}  \rVert_{\infty,\infty}\big)^{-1}  \lVert  \bm{x} (\mT_t^c)  \rVert_\infty . 
	\end{align*}
Now noticing that for a diagonal-dominate symmetric matrix $A$ with $\Delta_i,$ we have $ \lVert A^{-1}\rVert_{\infty,\infty} \leq  \max_i {1}/{\Delta_i(A)}. $
Then by \begin{align*}
 	 \Delta_i( \w(\mT_t^c)^*\w(\mT_t^c)) = \tilde{K}_n(0)  - \sum_{k\in \mT_t^c} \tilde{K}_n(\tau_i - \tau_k) \geq p(1-\dfrac{c_1}{(n\Delta)^4} -  cs\sqrt{\log(n)/np}), \quad \forall i \in [t],
 \end{align*}
 we have \begin{align*}
 	 \min_{i \in \mT_t^c} \lvert J_{2,t}(\tau_i)\rvert \geq p \big( 1- \dfrac{c_1}{(n\Delta)^4} - cs\sqrt{\log (n)/np} \big) \lvert x_{t+1}\rvert.
 \end{align*}

 On the other hand, for any $\tau \notin \mS(\mT_t^c), $  denoting $\mP_t\w(\tau) = \sum_{i\in \mT_t} b_i(\tau) \w(\tau_i),$ then by  Lemma~\ref{lem-subsampled-coefficient-bound}(an analog of Lemma~\ref{lem-b-bound-full}), we can show that 
  \begin{align*}
	\lvert  J_{2,t}(\tau)   \rvert & \leq \lvert \sum_{i\in \mT_t^c} x_i \w(\tau)^*\w(\tau_i)\rvert +  \lvert \sum_{i\in \mT_t^c} x_i \sum_{k\in \mT_t} b_k(\tau)\w(\tau_k)^*\w(\tau_i)\rvert \\
	&\leq \lvert x_{t+1} \rvert \lvert \sum_{i\in \mT_t^c}  \tilde{K}_n(\tau-\tau_i)\rvert + \big(\sum_{k\in \mT_t} \lvert b_k(\tau)   \rvert\big) \cdot \max_{k\in \mT_t} \lvert  \sum_{i\in \mT_t^c}x_i \w(\tau_k)^* \w(\tau_i)  \rvert\\
	&\leq  p \lvert x_{t+1} \rvert \bigg[0.7 + \underbrace{ \dfrac{c_0}{(n\Delta)^4}+ cs\sqrt{\log(n)/np} )+ \dfrac{(\dfrac{c_1}{(n\Delta)^4}+ cs\sqrt{\log(n)/np}) ^2}{1- \dfrac{c_3}{n^4\Delta^4} - cs \sqrt{\log(n)/np}} }_{\tilde{\nu}_2} \bigg]
\end{align*} 
as previous arguments in Section~\ref{appendix-proof-J-bound}.

\noindent 2. For $J_{3,t}(\tau),$ we have for any $\tau \in [0,1),$ denoting $\tilde{\mP}_t\w(\tau) = \sum_{i\in \mT_t} \tilde{b}_i(\tau)\w(\omega_i),$ then similar as in previous argument, we have \begin{align*}
	\lvert J_{3,t}(\tau)\rvert & \leq \lvert \sum_{i\in \mT_t^c} x_i \sum_{k\in \mT_t} \tilde{b}_k(\tau)\w(\omega_k)^*\w(\tau_i)\rvert + \lvert \sum_{i\in \mT_t^c} x_i \sum_{k\in \mT_t} b_k(\tau)\w(\tau_k)^*\w(\tau_i)\rvert\\
	&\leq p \lvert x_{t+1}\rvert \underbrace{ \big(1- \dfrac{c_3}{(n\Delta)^4} - cs \sqrt{\log(n)/np} \big)^{-1}   (\dfrac{c_0+c_1}{n^4\Delta^4} + cs\sqrt{\log(n)/np})}_{\tilde{\nu}_3 }.  
\end{align*} 
3.
 For $J_{1,t}(\tau),$ for all $\tau \notin \mS(\mT_t),$ we have similar to the argument in section~\ref{appendix-proof-J-bound},  \begin{align*}
	\lvert J_{1,t}(\tau) \rvert &=\lvert \sum_{i\in \mT_t} x_i \w(\tau)^*(\mI-\tilde{\mP}_t) (\w(\tau_i) - \w(\omega_i)  ) \rvert \\
	&\leq \big\lvert \sum_{i\in \mT_t} x_i [\tilde{K}_n(\tau-\tau_i)-\tilde{K}_n(\tau-\omega_i)]  \big \rvert + \lvert \sum_{i\in \mT_t } x_i \sum_{k\in \mT_t} \tilde{b}_k(\tau) \w(\omega_k)^*(\w(\tau_i) - \w (\omega_i)) \rvert \\
	&\leq  p \underbrace{\pi^2\bigg[ 1+ \dfrac{1 - \dfrac{c_3}{n^4\Delta^4} - c s \sqrt{ \log(n)/np} )} {1+\dfrac{c_0}{n^4\Delta^4}+ cs\sqrt{\log(n)/np}} \bigg] \cdot( \dfrac{1}{3} + \dfrac{c_2}{n^4\Delta^4}+cs \sqrt{ \log(n)/np}) }_{\tilde{\nu}_1}\cdot n \varepsilon_{\x,t}.
	\end{align*}
Thus if we define $$\tilde{\mathcal C}_t =p \cdot \bigg(0.3 - (\tilde \nu_2 + 2\tilde\nu_3 + \dfrac{c_1}{n^4\Delta^4}+ cs\sqrt{\frac{\log n}{np}} +2 \dfrac{n\varepsilon_{\x,t}}{\lvert x_{t+1}\rvert} \tilde \nu_1)\bigg) , $$
then we have $\mathcal{C}_t>0 \implies \omega_{t+1}\in \mS(\mT_t^c).$	 \\
Now  suppose $\omega_{t+1} \in \mS(x_{T(t+1)})$ for some $T(t+1) \in \mT_t^c,$ we have then \begin{align*}
		\lvert \w(\omega_{t+1})^*(\mI-\tilde{\mP} ) \y \rvert  > \lvert \w(\tau_{t+1})^*(\mI-\tilde{\mP} ) \y \rvert 
\end{align*} 
implies \begin{align*}
	\lvert  \tilde J_{2,t} (\omega_{t+1}) \rvert - \lvert \tilde J_{2,t}(\tau_{t+1}) \rvert \geq  2\sup_{\tau \in \mS(\mT^c_t)} \big( \lvert \tilde J_{1,t}(\tau )\rvert + \lvert \tilde J_{3,t}(\tau) \rvert \big)
\end{align*}
Moreover, similar as in Section~\ref{appendix-J-small-bound}, we have

\textbf{1. $J_2$ lower bound:} When $\tau_m \in \mT_t^c,$ we have denoting $\mP_t\w(\tau) = \sum_{i\in \mT_t} b_i(\tau) \w(\tau_i),$ then 
  \begin{align*}
	\lvert  J_{2,t}(\tau_m)   \rvert  &\geq p\bigg( \lvert x_m \rvert - \lvert x_{t+1}\rvert \underbrace{\bigg[1+  \dfrac{\dfrac{c_2(\varepsilon_t,\Delta)}{(n+2)^4\Delta^4}+cs \sqrt{\log(n)/np}}{  1- \dfrac{c_3(\varepsilon_t,\Delta)}{n^4\Delta^4} - cs \sqrt{\log(n)/np} }   \bigg]\cdot \bigg[  \dfrac{c_1} {(n+2)^4\Delta^4}+ cs\sqrt{\frac{\log(n)}{np}}  \bigg]}_{\tilde{\kappa}_{22} } \bigg) .  
	\end{align*}
\textbf{2. $J_2$ upper bound: } When $\omega\in \mS(\tau_m), $ we have \begin{align*}
	\lvert J_{2,t}(\omega) \rvert 
	& \leq p \bigg(\lvert x_m \rvert + \lvert x_{t+1}\rvert \underbrace{\bigg[ \dfrac{c_2(\frac{1}{2n},\Delta)}{n^4\Delta^4}+ cs\sqrt{\frac{\log n}{np}} + \dfrac{\big( c_1 c_2(\frac{1}{2n}+\varepsilon_t,\Delta)/(n\Delta)^8+ cs\sqrt{\frac{\log n}{np}}\big)   }{1 - \frac{c_3}{n^4\Delta^4} - cs \sqrt{\frac{\log n}{np}}}  \bigg] }_{\tilde\kappa_{21} }  \bigg) 
\end{align*}  
\textbf{3. $J_3$ upper bound: } When $\tau \in \mS(\tau_m)$ for some $\tau_m$,    denoting $\tilde{\mP}_t\w(\tau) = \sum_{i\in \mT_t} \tilde{b}_i(\tau)\w(\omega_i),$ then similar as in previous argument, we have \begin{align*}
	\lvert J_{3,t}(\tau)\rvert&\leq p\lvert x_{t+1}\rvert  \underbrace{ \dfrac{       \dfrac{c_2(\varepsilon_t,\Delta) c_2(\varepsilon_t+\frac{1}{2n},\Delta)}{(n+2)^8\Delta^8 } + \dfrac{c_2(\frac{1}{2n},\Delta) c_1}{(n+2)^8\Delta^8 } +cs\sqrt{\log(n)/np}   }{ 1- \dfrac{c_3(\varepsilon_t ,\Delta)}{n^4\Delta^4} - cs \sqrt{\log(n)/np} }}_{\tilde{\kappa}_3} .
\end{align*} 
\textbf{4. $J_1$ upper bound: }When $\tau \in \mS(\tau_m) $ for some $m \in \mT_t^c,$ we have  \begin{align*}
	\lvert J_{1,t}(\tau) \rvert & \leq p\cdot n\varepsilon_{\x,t}   \underbrace{2\pi^2 \bigg[ \dfrac{c_2(\varepsilon_t+\frac{1}{2n},\Delta)}{(n+2)^4\Delta^4} + cs\sqrt{\log (n)/np}\bigg] \bigg[ 1+\dfrac { \frac{1}{3}+ \dfrac{c_2(\varepsilon_t,\Delta)}{(n\Delta)^4} + cs \sqrt{\log(n)/np} }{ 1- \dfrac{c_3(\varepsilon_t,\Delta)}{n^4\Delta^4} - cs \sqrt{\log(n)/np} }  \big) \bigg]}_{\tilde{\kappa}_{1}}
\end{align*}

Then the claim holds by letting \begin{align*}
	 \tilde{\lambda}_t = \bigg[1 - \big( 2 \dfrac{n\varepsilon_{\x,t}}{\lvert x_{t+1}\rvert}\tilde{\kappa}_1+\tilde{\kappa}_{21}+\tilde{\kappa}_{22}+2\tilde{\kappa}_3 \big) \bigg].  
\end{align*}

\end{proof}

\subsection{Proof of Proposition~\ref{prop-subsampled-improved-estimation}}

\begin{proof}\quad \\
\textbf{The first iteration: } For the first step, suppose W.L.O.G. $\lvert x_1 \rvert = \max_{i\in [s]}\lvert x_i\rvert,$  we have by $\omega_1 \in \mS(\tau_{T(1)}),$  \begin{align*}
	&\lvert  \w(\omega_{1})^*\y \rvert \geq \lvert \w(\tau_{T(1)})^*\y \rvert \\
	\implies & \lvert \sum_{i=1}^s x_i  \tilde{K}_n(\omega_1 -\tau_i)\rvert \geq  \lvert \sum_{i=1}^s x_i \tilde{K} _n(\tau_{T(1)} -\tau_i)\rvert\\
	\implies & \lvert x_{T(1)}\rvert (p- p K_n(\omega_1-\tau_{T(1)}))  \leq p\lvert x_1 \rvert \bigg[\dfrac{2 c_1}{(n+2)^4\Delta^4} +cs\sqrt{\frac{\log n}{np}}\bigg] \\
	\implies &1- {K}_n  (\omega_1-\tau_{T(1)}) \leq  \lvert x_1\rvert \bigg[  \dfrac{2c_1}{\lvert x_{T(1)}\rvert (n+2)^4\Delta^4} + cs \sqrt{\dfrac{\log n}{np}} \bigg]  \\
	\implies& -\int_{0}^{\omega_1-\tau_{T(1)}} \int_0^u K_n''(v)  dvdu \leq   \lvert x_1\rvert \bigg[  \dfrac{2c_1}{\lvert x_{T(1)}\rvert (n+2)^4\Delta^4} + cs \sqrt{\dfrac{\log n}{np}} \bigg]\end{align*}
Then by   \begin{align*}
	\int_{0}^{\omega_1-\tau_{T(1)}} \int_0^u K_n''(v)  dvdu 	&\leq -1.19 n^2\lvert \omega_1-\tau_{T(1)}\rvert^2,
\end{align*}
we get \begin{align*}
	\lvert \omega_1-\tau_{T(1)}\rvert \leq \dfrac{1}{\sqrt{1.19}n} \big[ \dfrac{ \frac{2c_1}{n^4\Delta^4} +cs\sqrt{\frac{\log n}{np }}}{{1-\frac{2c_1}{n^4\Delta^4} - cs\sqrt{\frac{\log n}{np}}}  } \big]^{1/2}
\end{align*}
Thus the claim holds for the first iteration with $\tilde{\lambda}_1$ defined above.
\\
\textbf{The $t+1 $-th iteration: } Suppose W.L.O.G.$\lvert x_{t+1}\rvert = \lVert \x(\mT_{t}^c)\rVert_\infty$, then by $  \omega_{t+1} \in \mS(\tau_{T(t+1)}) $ and  \begin{align}\label{eq-OMP-stept-basic-inequality}
	 &\lvert \w(\omega_{t+1})^*(\mI - \tilde{\mP}_t) \y \rvert \geq \lvert \w(\tau_{T(t+1)})^* (\mI - \tilde{\mP}_t)\y\rvert . 	
\end{align}
Then   \begin{align*}
	 p\lvert x_{T(t+1)}\rvert (1- K_n(\omega_{t+1}-\tau_{T(t+1)})) \leq p  \lvert x_{t+1}\rvert  \underbrace{\big(2\dfrac{n\varepsilon_{\x,t}}{\lvert x_{t+1}\rvert}    \tilde\kappa_1 + \tilde\kappa_{21}+ \tilde\kappa_{22}+ 2\tilde \kappa_{3}  \big)} _{\tilde \lambda_{t+1}},
\end{align*} 
now as argued in step one, we get \begin{align*}
	\lvert \omega_{t+1}-\tau_{T(t+1)}\rvert \leq \dfrac{1}{\sqrt{1.19}n} \big[(1-\lambda_{t+1})^{-1}\lambda_{t+1} \big]^{1/2},
\end{align*}
as desired. 
\end{proof}

\subsection{Proof of Proposition~\ref{prop-subsampled-sliding}} 

\begin{proof} 
  Our proof will be divided into two steps. \\
  \textbf{First Step:}
  	Denoting $\mg(\bomega)$ the gradient of loss at $\bomega$, we have \begin{align*}
	    \mg(\bomega) & = 2 \text{Re}\big[  \y^* (\mP(\bomega)-\bm{I})  \bm{Z}(\bomega)  \text{diag}(\hat\x ) \big]  .
	\end{align*}
	Thus we have the following upper bound and lower bound of $\lvert g_m(\bomega)\rvert^2$ and $g_m(\bomega)(\omega_m - \tau_m)$:
   \begin{align*}
   	  \lvert \dfrac{g_m(\bomega)}{2} \rvert^2
	&\leq    \lvert \hat{x}_m \rvert^2\big[ \underbrace{\lvert \bar{x}_m \w(\tau_m)^*(\mP(\bomega)-\bm{I})\mz(\omega_m)\rvert^2}_{G_{1,m}^2} +  \underbrace{\lvert  \x_{-m}^*\w_{-m}^* (\mP(\bomega) - \bm{I} ) \mz(\omega_m)  \rvert^2}_{G_{2,m}^2}+ 2 G_{1,m}G_{2,m}\big] ,\\
	  \dfrac{g_m(\bomega)}{2}(\omega_m - \tau_m)
	 &\geq  \underbrace{\text{Re}\big(\hat{x}_m \bar{x}_m \w (\tau_m)^*(\mP_t-\mI)\mz(\omega_m) \big) (\omega_m - \tau_m)}_{G_{3,m}} -  \lvert \hat{x}_m\rvert \underbrace{\big \lvert  \x_{-m}^*\w_{-m}^* (\mP_t-\mI)\mz(\omega_m)\big) (\omega_m - \tau_m) \big\rvert}_{=G_{2,m}\cdot \lvert \omega_m - \tau_m\rvert }.
\end{align*}
and we aim to control the lower bound of $G_{3,m}$ and upper bounds of $G_{2,m},G_{1,m} $ separately.
As shown in proof of Proposition~\ref{thm-sliding-regularity},  we have the proof relies on the bound on $G_{1,m},G_{2,m},G_{3,m},$  and we provide an analogue of Lemma~\ref{lem-G-bound-squared} as following: 
\begin{lemma} \label{lem-subsampled-G-bound} As long as $ \max_i \lvert  x_{i}(\omega_i - \tau_{i}) \rvert
	 \leq \dfrac{\lVert \x_{\geq t}\rVert_\infty }{2n}  $ , we have 
\begin{align*}
	G_{1,m}&\leq   p\cdot \tilde{\varphi}_1  \cdot  n^2\lvert x_m(\omega_m-\tau_m) \rvert ,\\
	G_{2,m}&\leq p\cdot \tilde{\varphi}_{21}  \cdot n^2\varepsilon_{\x,t}  + np\lVert \x_{>t}\rVert_\infty  \cdot  \tilde {\varphi}_{22},\\
	G_{3,m}&\geq p\cdot \tilde \varphi_3 \cdot n^2\lvert x_m\rvert^2 (\omega_m-\tau_m)^2  .
\end{align*}
with \begin{align*}
&\dfrac{\pi^2}{3}  \leq \tilde \varphi_{1} \leq \dfrac{\pi^2}{3}+O \big(\dfrac{1}{(n\Delta)^8}+ s\sqrt {\frac{\log n}{np} } \big) , \\
 &\tilde\varphi_{21},\tilde \varphi_{22} = O(\dfrac{1}{(n\Delta)^4} +s \sqrt{\dfrac{\log n}{np} }), \\
   &\dfrac{\pi^2}{3} \geq \tilde \varphi_3 \geq  \dfrac{\pi^2}{3} - O\big(n^2\varepsilon_t^2+\dfrac{1}{(n\Delta)^8} - s\sqrt{\frac{\log n}{np}} \big) .  \end{align*}
\end{lemma} 
Then the claim holds by the same argument as proof of Proposition~\ref{thm-sliding-regularity} with  \begin{align*}
\tilde{\mathcal R}_t &= \big[\dfrac{1}{2} - \dfrac{1+\tilde{\rho} }{2\tilde{\varphi}_3} \tilde{\varphi}_{21} \big]^{-1} \cdot 2 \tilde{\varphi}_{22} = \big[ \dfrac{1}{2} - O(\dfrac{1}{(n\Delta)^4}+ s\sqrt{\dfrac{\log n}{np}}) \big]  \cdot O\big( \dfrac{1}{(n\Delta)^4} + s \sqrt{\dfrac{\log n}{np}} \big)   \\
	\tilde{\mathcal {CR}} _t &= \big[ \dfrac{1}{4} - \dfrac{1+\tilde{\rho}}{4\tilde{\varphi}_3}\tilde{\varphi}_{21} \big]^2  - \big( \dfrac{18}{\pi^4} + \dfrac{(1+\tilde\rho )^2}{4\tilde{\varphi}_3}   \big) \cdot \big(\dfrac{1}{4}+\tilde{\varphi}_{21}\big)^2= \dfrac{1-18/\pi^4}{16} -  O( \dfrac{1}{(n\Delta)^4} + s \sqrt{\dfrac{\log n}{np}}+ n\varepsilon_t)^2   \\
	\tilde{\mathcal E}_t &=    \dfrac{3}{2\pi^2}\bigg( \big[ \dfrac{1}{2} - \dfrac{1+\tilde\rho}{2\tilde\varphi_3}\tilde{\varphi}_{21} \big]^{-1}2(\tilde{\varphi}_1+\tilde{\varphi}_{21})+1  \bigg)\tilde{\varphi}_{22} = \big(1 + [\dfrac{1}{2}-O(  \dfrac{1}{(n\Delta)^4} + s \sqrt{\dfrac{\log n}{np}}  ) ]^{-1}   \big)   \cdot O( \dfrac{1}{(n\Delta)^4} + s \sqrt{\dfrac{\log n}{np}}).
\end{align*}
 
\end{proof}

\begin{proof}[Proof of Lemma~\ref{lem-subsampled-G-bound}]
	\quad \\
	\textbf{1. Lower Bound of $G_{3,m}$:}\\	
 Noticing that \begin{align*}
	\bar{x}_m \w(\tau_m)^* (\mP_t - \bm{I}) \mz(\omega_m)   & = \bar{x}_m [ \w(\tau_m)^* \mP_t \mz(\omega_m)  -K_n' ( \omega_m - \tau_m)       ] ,\\
	\hat{x}_m \bar{x}_m &= [ (\bm{F}_t^* \bm{F}_t)^{-1} \bm{F}_t^*  \y ] _m \bar{x}_m \\
	&=  \sum_{i=1}^m x_i\bar{x}_m [  (\bm{F}_t^* \bm{F}_t^*)^{-1}  \bm{F}_t^* \w(\tau_i)    ]_m  \\
	& =  \bar{x}_m \e_m^* (\W^*_t\W_t)^{-1}  \bm{F}_t^* [\sum_{i\leq t}+ \sum_{i> t} x_i \w(\tau_i) ] .
\end{align*}
By Lemma~\ref{lem-subsampled-xhat-bound} , we have for $$
	 \rho(\varepsilon_t,\Delta, \zeta): = \big(1- \dfrac{c_3(\varepsilon_t,\Delta)}{n^4\Delta^4})^{-1} \cdot\big((n+2)\zeta (\dfrac{\pi^2}{3} +\dfrac{\pi^2c_2(\varepsilon_t,\Delta)}{(n\Delta)^4}) + \dfrac{2   c_2(\varepsilon_t,\Delta)}{(n\Delta)^4}   \big), $$
	 if we denote $\rho  = \rho(\varepsilon_t,\Delta,\frac{\varepsilon_{\x,t}}{\min_{i\in {\mT_t} \lvert x_i \rvert }}),$  then
  \begin{displaymath}	\big\lvert \hat{x}_m \bar{x}_m - \lvert x_m\rvert^2 \big\rvert  \leq  \rho ,\end{displaymath}
   thus $\text{Re}(\hat{x}_m \bar{x}_m) \geq\big(1-\rho \big) \lvert x_m \rvert^2 . $
On the other hand, denote $\bm{P}_t \bm{z}(\omega_m) = \sum_{i\neq m,i\in [t]} a_{i,m} \w(\omega_i),$ then \allowdisplaybreaks \begin{align*}
\lvert \w(\tau_m)^* \mP(\bomega) \mz(\omega_m) \rvert  & = \lvert [\w(\tau_m)-\w(\omega_m)] ^* \mP(\bomega) \mz(\omega_m) \rvert  \\
&\leq \lVert \bm{a}_{m,-m}\rVert_1 \cdot \max_{i\neq m}\lvert  [\w(\tau_m)-\w(\omega_m)] ^* \w(\omega_i)  \rvert \\
&   \leq p^{-1} \big(1- \dfrac{c_3(\varepsilon_t,\Delta)}{n^4\Delta^4} - cs \sqrt{\log(n)/np})^{-1} \cdot \sum_{i\neq m}\lvert \tilde K_n'(\omega_m - \omega_i) \rvert \cdot \max_{i\neq m} \lvert (\omega_m-\tau_m) \tilde K_n'(\omega_i-\xi_m)\rvert \\&\leq 2\pi^4 \dfrac{[c_3/(n\Delta)^4+ cs\sqrt{\log (n)/np}]^2}{1- \dfrac{c_3}{n^4\Delta^4} - cs\sqrt{\log (n)/np} } \cdot n^2 p \lvert \omega_m -  \tau_m\rvert  .
\end{align*}
where we have used  $$ \sum_{i\neq m }\lvert \tilde K_n'(\omega_m-\omega_i)\rvert \leq 3np\cdot [\dfrac{c_3(\varepsilon_t,\Delta)}{n^4\Delta^4}+cs \sqrt{\log(n)/np ]}, \quad \lvert \tilde K_n'(\omega_i - \xi_m) \rvert \leq 3n \cdot [\dfrac{c_3(\varepsilon_t,\Delta)}{n^4\Delta^4} + cs\sqrt{\log(n)/np}],  $$ 
Finally, we have  by $\tilde{K}_n'(0) = 0,$ \begin{align*}
	-(\omega_m - \tau_m) \tilde{K} _n'(\omega_m-\tau_m) & \geq  -p (\omega_m-\tau_m) \int_{0}^{\omega_m - \tau_m} {K}_n''(v) dv - c n^2p\lvert \omega_m-\tau_m \rvert^2  cs\sqrt{\log(n)/np} \\
	& \geq  p (\omega_m-\tau_m)\int_0^{\omega_m-\tau_m} \dfrac{\pi^2}{3}n(n+4)-\dfrac{\pi^4}{6}(n+2)^4 v^2 dv  -  n^2p\lvert \omega_m-\tau_m \rvert^2  cs\sqrt{\log(n)/np} \\
	&\geq \big[\dfrac{\pi^2}{3}- \dfrac{\pi^4}{18} (n+2)^2  \varepsilon_t^2 - cs\sqrt{\dfrac{\log(n)}{np} } \big]  n^2p (\omega_m-\tau_m)^2 .
\end{align*}

Combining these bounds, we get 
\begin{align*}
G_{3,m}  &\geq  \text{Re}\big( \hat{x}_m\bar{x}_m  \big)\big ( - K_n'(\omega_m - \tau_m)\big ) (\omega_m - \tau_m) -  \lvert x_m^2( \omega_m-\tau_m)\rvert \cdot \lvert \w(\tau_m)^*\mP(\bomega)\mz(\omega_m)\rvert\\
	& \geq  p \lvert x_m\rvert^2 \underbrace{ \bigg([1-\tilde\rho_t]\big[\dfrac{\pi^2}{3}- \dfrac{\pi^4}{18} (n+2)^2  \varepsilon_t^2 - cs\sqrt{\dfrac{\log(n)}{np} }\big]  - 2\pi^4 \dfrac{[c_3/(n\Delta)^4+ cs\sqrt{\log (n)/np}]^2}{1- \dfrac{c_3}{n^4\Delta^4} - cs\sqrt{\log (n)/np} } \bigg)}_{\tilde{\varphi}_3}  n^2 (\omega_m-\tau_m)^2   ,
\end{align*}
as desired.

\noindent\textbf{2. Upper bound of $G_{1,m}$:} By previous bounds, we have \begin{align*}
	G_{1,m} &\leq \lvert x_m \tilde{K} '_n(\omega_m - \tau_m) \rvert    +\lvert x_m \w(\tau_m)^* \mP(\bomega) \mz(\omega_m) \rvert\\
	& \leq p\lvert x_m\rvert\bigg[\int_{0}^{\lvert \omega_m-\tau_m\rvert}\lvert K_n''(v)\rvert dv + 9   \dfrac{[c_3+cs\sqrt{\log(n)/np}]^2 /{(n\Delta)^8}}{1- \dfrac{c_3(\varepsilon_t,\Delta)}{n^4\Delta^4} - cs \sqrt{\log(n)/np}} \cdot  \big( n^2 + cs n^2\sqrt{\frac{\log(n)}{np}} \big) \lvert \omega_m - \tau_m\rvert \bigg] \\
	&\leq p \cdot \underbrace{\bigg(\dfrac{\pi^2}{3}+ 2\pi^4  \dfrac{[c_3+cs\sqrt{\log(n)/np}]^2 /{(n\Delta)^8}}{1- \dfrac{c_3(\varepsilon_t,\Delta)}{n^4\Delta^4} - cs \sqrt{\log(n)/np}}  + cs  \sqrt{ \dfrac{\log(n)}{np}} \bigg) }_{\tilde{\varphi}_{1}}  \cdot  n^2\lvert x_m(\omega_m-\tau_m) \rvert 
\end{align*}

\noindent\textbf{3. Upper bound of $G_{2,m}$:}   
We would divide $G_{2,m}$ into two parts and bound them separately: \begin{align*}
	G_{2,m} &\leq   \lvert\sum_{i\leq t,i\neq m}  x_i \w (\tau_i)^* (\mP(\bomega)-\bm{I}) \mz(\omega_m) \rvert +   \lvert  \x_{>t}^* \w_{>t}^* (\mP(\bomega)-\bm{I}) \mz(\omega_m) \rvert 
\end{align*}
For the first part, we have
\begin{align*}
	  \lvert\sum_{i\leq t,i\neq m}  x_i \w (\tau_i)^* (\mP(\bomega)-\bm{I}) \mz(\omega_m) \rvert & = \lvert \sum_{i\leq t,i\neq m}  x_i [\w(\tau_i) - \w(\omega_i)]^*(\mP(\bomega)-\bm{I}) \mz(\omega_m) \rvert \\
	  & = \lvert \sum_{i\leq t,i\neq m}  x_i [\w(\tau_i) - \w(\omega_i)]^*(\mP(\bomega)-\bm{I}) \mz(\omega_m)  \rvert \\
	  &\leq \sum_{i\leq t,i\neq m} \lvert  x_i\rvert \big[ \sum_{k \leq t, k\neq m}a_{k,m}  \lvert( \w(\tau_i)-\w(\omega_i))^*\w(\omega_k) \rvert     + \lvert [\w(\tau_i) - \w(\omega_i)]^* \mz(\omega_m)  \rvert  \big]  \\
	  & \leq \lVert \bm{a}_{-m,m}\rVert_1   \max_{k\leq t,k\neq m} \big\lvert \sum_{i\leq t, i\neq m} \lvert x_i \rvert \cdot \lvert \tilde K_n(\tau_i-\omega_k) - \tilde K_n(\omega_i - \omega_k)\rvert \big\rvert\\
	  &\quad +\sum_{i<t,i\neq m} \lvert x_i [\tilde K_n'(\tau_i-\omega_m) -\tilde  K_n'(\omega_i-\omega_m)]\rvert  \\
	 \end{align*} 
	 In particular, the first term can be bounded by 
	 $$np\cdot  2\pi^4 \dfrac{[c_3/(n\Delta)^4+ cs\sqrt{\log (n)/np}]^2}{1- \dfrac{c_3}{n^4\Delta^4} - cs\sqrt{\log (n)/np} } \cdot \varepsilon_{\x,t}   [n  +  csn \sqrt{\log(n)/np}]  $$
	 as we do when bounding $G_{3,m}.$ For the second term, we have \begin{align*}
	 	\sum_{i<t,i\neq m} \lvert x_i[\tilde K_n'(\tau_i - \omega_m) - \tilde K_n'(\omega_i - \omega_m) ] \rvert & \leq \sum_{i<t,i\neq m} \lvert x_i(\tau_i -\omega_i) \tilde K_n''(\xi_i - \omega_m)\rvert \\
	 	&\leq \varepsilon_{\x,t} \sum_{i<t , i\neq m} \lvert \tilde K_n''(\xi_i -\omega_m) \rvert\\
	 	&\leq p\varepsilon_{\x,t}4\pi^4\big[ \dfrac{n^2c_{3}(\varepsilon_t,\Delta) }{(n\Delta)^4} + n^2cs \sqrt{ \log (n)/np} \big]
	 \end{align*}
Thus the first part  is upper bounded by \begin{align*}
	2\pi^4 n^2p  \varepsilon_{\x,t}  [\dfrac{c_3(\varepsilon_t,\Delta)}{n^4\Delta^4}+ cs\sqrt{\frac{\log(n)}{np}} ]  \bigg[ 2 + \dfrac{1+\frac{c_2(\varepsilon_t,\Delta)}{(n\Delta)^4} +cs\sqrt{\log(n)/np}}{1-\frac{c_3(\varepsilon_t,\Delta)}{(n\Delta)^4} -cs\sqrt{\log(n)/np}}  \bigg]  \cdot [1+cs \sqrt{\log(n)/np}] 
\end{align*}

	 And for the remaining part,  we have 
\begin{align*}
    \lvert  \x_{>t}^* \w_{>t}^* (\mP(\bomega)-\bm{I}) \mz(\omega_m) \rvert & \leq \sum_{i>t} \lvert x_i \rvert\big[ \sum_{\ell\neq m ,\ell \leq t}\lvert a_{\ell,m} \w(\omega_i)^* \w(\omega_\ell)\rvert     +\lvert    \w(\omega_i)^* \mz(\omega_m) \rvert\big]\\
    &\leq \sum_{\ell\neq m,\ell\leq t} \lvert a_{\ell,m}\rvert\big[ \sum_{i>t }  \lvert x_i \rvert  \lvert \w(\omega_i)^* \w(\omega_\ell)\rvert \big] + \sum_{i>t} \lvert x_i \rvert \lvert \w(\omega_i)^*\mz(\omega_m) \rvert \\
    &\leq \lVert \bm{a}_{-m,m}\rVert_1 \lVert \x_{>t}\rVert_\infty \max_{\ell\leq t}\lvert \sum_{i>t} \w(\omega_i)^*\w(\omega_\ell)\rvert + \lVert \x_{>t}\rVert_\infty  \sum_{i>t} \lvert \w (\omega_i)^*\mz(\omega_m)\rvert  \\
    &\leq  np\lVert \x_{>t}\rVert_\infty  \cdot \pi^2 \bigg[\dfrac{ [\dfrac{c_3 }{n^4\Delta^4} + cs \sqrt{\log(n)/np}]^2} { 1- \dfrac{c_3(\varepsilon_t,\Delta)}{n^4\Delta^4} -cs \sqrt{\log(n)/np} } +  \dfrac{c_3 }{n^4\Delta^4}+ cs \sqrt{\log(n)/np}   \bigg]  
\end{align*}
\noindent Combining all bounds above  together, we get \begin{align*}
G_{2,m}&\leq  n^2p  \varepsilon_{\x,t}    \underbrace{2\pi^4 [\dfrac{c_3(\varepsilon_t,\Delta)}{n^4\Delta^4}+ cs\sqrt{\frac{\log(n)}{np}} ]  \bigg[ 4 + 3\dfrac{1+\frac{c_2(\varepsilon_t,\Delta)}{(n\Delta)^4} +cs\sqrt{\log(n)/np}}{1-\frac{c_3(\varepsilon_t,\Delta)}{(n\Delta)^4} -cs\sqrt{\log(n)/np}}  \bigg]  \cdot [1+cs \sqrt{\log(n)/np}]}_{\tilde\varphi_{21}} \\
&+ np\lVert \x_{>t}\rVert_\infty  \cdot \underbrace{ \pi^2 \bigg[\dfrac{ [\dfrac{c_3 }{n^4\Delta^4} + cs \sqrt{\log(n)/np}]^2} { 1- \dfrac{c_3(\varepsilon_t,\Delta)}{n^4\Delta^4} -cs \sqrt{\log(n)/np} } +  \dfrac{c_3 }{n^4\Delta^4}+ cs \sqrt{\log(n)/np}   \bigg]}_{\tilde \varphi_{22}}  .
\end{align*}
as desired.

\end{proof}

\section{Proof of Results in section~\ref{sec-byproduct-OMP} }

\subsection{Proof of Theorem~\ref{thm-omp-naive}}

\textbf{At the first round: } Since the first step of OMP and Sliding-OMP has no difference, we have it held directly by Proposition~\ref{prop-improved-localization} that \begin{align*}
	\lvert \omega_1 - \tau_{T(1)}\rvert \lesssim \dfrac{1}{n\dynx^2 \zeta^2}
\end{align*}
by $n\Delta > \zeta \dynx. $\\
\textbf{At the $t+1$-th round:} Supposing $t+1 \leq s$ and there exists some $c\geq 1$  independent of $t$ so that \begin{align*}
	\lvert \omega_k - \tau_{T(k)}\rvert \leq \dfrac{c}{n\cdot  \dynx^2 \zeta^2}
\end{align*}
for $1\leq k \leq t,$ we have then $\lambda_{t+1} \leq  \frac{C'c}{n^4\Delta^4\zeta^2} $and  \begin{align*}\mathcal{C}_{t+1} &= 0.3 - O(\dfrac{1}{(n\Delta)^4}+ \dfrac{n \varepsilon_{\x,t}}{\lVert \x(\mT_t^c)\rVert_\infty} )  \\
&> 0.3 - O \big(\dfrac{1}{\zeta^4} +  \dfrac{c}{\zeta} \big) \\
&> 0.3 - O(\dfrac{1}{C^4}+ \dfrac{c}{C})  > 0.
\end{align*}
thus by Proposition~\ref{lem-error-formula} we have there exists $T(t+1)$ so that $ \omega_{t+1}\in \mS (\tau_{T(t+1)})$ then by Proposition~\ref{prop-improved-localization}, we get \begin{align*}
	 \lvert \omega_{t+1} - \tau_{T(t+1)}\rvert  &\leq \dfrac{1}{n} \sqrt{\dfrac{\lambda_{t+1}}{1.19(1-\lambda_{t+1})}}\\
	 &\leq \dfrac{1}{n} \sqrt{ \dfrac{C'c}{1.19n^4\Delta^4 \zeta^2 (1- \dfrac{C'c}{ (n\Delta)^4\zeta^2 }  ) }  }\\
	 &\leq \dfrac{1}{n}\dfrac{C''c}{n^2\Delta^2 \zeta}\\
	 &\leq \dfrac{1}{n  \dynx^2 \zeta^2}\cdot \dfrac{C''c}{C}.
\end{align*}
Thus the claim at $t+1$-th step holds when  $C$ is large enough so that $C \geq C''.$ 

In conclusion, we show by induction that when $t\leq s$, we have \begin{align*}
	\lvert \omega_{k}-\tau_{T(k)}\rvert
 \leq \dfrac{c}{n \dynx^2\zeta^2}.\end{align*} 
Now we would show that there exists $c',c''$ so that the algorithm with threshold $c'\min_{i}\lvert x_i\rvert \leq \gamma \leq c'' \min_{i}\lvert x_i\rvert $ will stop after exactly $s$ steps: \\
When $t<s,$ we have $\mT_t^c \neq \emptyset$ and  \begin{align*}
	\max_{\tau \in [0,1)}  \lvert\w^*(\tau)  \bm r_{t} \rvert &\geq  \max_{i\in \mT_t^c} \lvert J_2(\tau) \rvert - \sup_{\tau\in [0,1)}  \lvert J_{1,t}(\tau)+J_{3,t}(\tau) \rvert \\
	&\geq \big(1- \dfrac{c_1}{n^4\Delta^4}\lVert \x(\mT_t^c)\rVert_\infty \big) -\nu_1  n\varepsilon_{\x,t} - \nu_3\lVert \x(\mT_t^c)\rVert_\infty \\
	&\geq \bigg(1-\dfrac{c_1}{n^4\Delta^4} - \nu_1 \dfrac{n\varepsilon_{\x,t} }{\lVert \x(\mT_t^c)\rVert_\infty } - \nu_3 \bigg) \lVert \x(\mT_t^c) \rVert_\infty \\
	& \geq \underbrace{\bigg(1-O (\dfrac{c}{C^2 } + \dfrac{1}{n^4\Delta^4}) \bigg)}_{c''} \min_i\lvert x_i  \rvert.      
\end{align*}
Thus the algorithm with threshold $\gamma\leq c''\min_i \lvert x_i \rvert	$ will not stop before $s$-th iteration.\\
On the other hand, when $t  = s,$ we have \begin{align*}
\lvert	 \w(\tau)^*\bm r_{s}\rvert &= \lvert	 \w(\tau)^*\bm r_{s}\rvert \\
& = \lvert \w(\tau)^*(\bm I - \bm P_s) \bm y \rvert \\
&=\lvert J_{1,s}(\tau) \rvert\\
&\leq C_3 n\max_i\lvert x_i \rvert  \varepsilon_s  \\
&\leq  \underbrace{\dfrac{C_3}{\zeta^2}}_{c'}\min_i\lvert x_i \rvert 
\end{align*}
Thus the algorithm with $\gamma \geq c'\min_i \lvert x_i \rvert $ will stop before or on $s+1$-th step.

Finally, for large enough $C$ we have $c'< c'',$ thus  the interval $I =[c'\min_i\lvert x_i \rvert,  c''\min_i\lvert x_i \rvert]$ is non-empty and the algorithm with $\gamma \in I$ will stop exactly after $s$ iterations. 
\subsection{Proof of Theorem~\ref{thm-impossible-OMP} }

Consider the following instance: \begin{align*}
	\y  = x_1 \w(\tau_1)+x_2 \w(\tau_2)+x_3 \w(\tau_3)
\end{align*}
We would first claim several properties of the Dirichlet kernel, which can be proved following the same argument as in section~\ref{appendix-kernel-inequality}: \begin{align*}
	\lvert D_n(t) \rvert &\lesssim \dfrac{1}{nt}, t\in [-1/2,1/2]\\
	\lvert D_n'(t) \rvert &\lesssim \dfrac{1}{t}, t\in [-1/2,1/2]\\
	\lvert D_n''(t) \rvert &\lesssim \dfrac{n}{t}, t\in [-1/2,1/2] \\
	  D_n''(t)  &\asymp -n^2, t\in [-\frac{1}{2n+4},\frac{1}{2n+4} ]
\end{align*}
Firstly, notice that \begin{align*}
		D_n'(x) = \pi \bigg( \dfrac{\cos (n\pi x)}{\sin (\pi x)} - \dfrac{ \cos (\pi x) \sin (n\pi x) }{n\sin^2 (\pi x)}  \bigg)
\end{align*}
we have by Taylor expansion, \begin{align*}
	D_n'(\dfrac{k}{n}) &= \pi \dfrac{(-1)^{k} }{\sin (\pi k/n)} =  (-1)^k\dfrac{n}{k}+ O( \dfrac{k^3}{n^3}), \\
	D_n'(\dfrac{k+\varepsilon}{n})& = \pi \bigg( \dfrac{\cos ((k+\varepsilon)\pi )}{\sin ({\pi(k+\varepsilon)}/{n})} - \dfrac{\cos(\pi (k+\varepsilon)/n)\sin (\pi (k+\varepsilon))  }{n \sin^2 (\pi (k+\varepsilon)/n)}  \bigg) \\
	& = \pi \bigg( \dfrac{(-1)^k \big(1+O(\varepsilon^2)\big) }{\pi (k+\varepsilon)/n+O((k+\varepsilon)^3/n^3 ) } - \dfrac{\big(\varepsilon+ O(\frac{k\varepsilon+\varepsilon^2}{n})\big)   }{n\big(\frac{\pi^2 (k+\varepsilon)^2}{n^2}+ O(\frac{k^6+\varepsilon^6}{n^6} ))}  \bigg)\\
	& =  \bigg( 1+ O(\varepsilon+\frac{k}{n}) \bigg) (-1)^k \dfrac{k}{n}.
\end{align*}
Now let $\lvert \tau_2 - \tau_1\rvert :=\Delta_1 = \dfrac{\ell_1}{n} , \lvert \tau_3 - \tau_1\rvert =  \Delta_2: =  \dfrac{\ell_2}{n} ,\lvert \tau_2-\tau_3\rvert = \Delta_3:= \dfrac{\ell_3}{n}$ for  large enough positive integers $\ell_1,\ell_2,\ell_3$ to be determined, and $ x_1>2x_2>4x_3$, we get then by an analogue of Theorem~\ref{prop-improved-localization},   $\lvert \omega_1 - \tau_1\rvert < \frac{1}{\sqrt{\min (n\ell_1,n\ell_2 )}} . $ Now for \begin{align*}
	h(\tau): =  \w(\tau)^*\y =  x_1 D_n (\tau- \tau_1) +x_2 D_n (\tau- \tau_2) +x_3 D_n (\tau- \tau_3),
\end{align*}
we have 
\begin{align*}
	 &0=h'(\omega_1)  =  x_1 D_n'(\omega_1 - \tau_1)  + \bigg(1+O( \dfrac{1}{\sqrt{\min( n\ell_1,n\ell_2)}}+ \frac{k}{n}) \bigg)\bigg( (-1)^{\ell_1} x_2 \dfrac{n}{\ell_1}+(-1)^{\ell_2} x_3 \dfrac{n}{\ell_2} \bigg)
\end{align*}
which then leads to    \begin{align*}
x_1 D_n''(\xi_1) (\omega_1-\tau_1)	= \bigg(1+O( \dfrac{1}{\sqrt{\min( \ell_1,\ell_2)}}+ \frac{k}{n}) \bigg)\bigg( (-1)^{\ell_1} x_2 \dfrac{n}{\ell_1}+(-1)^{\ell_2} x_3 \dfrac{n}{\ell_2} \bigg). 
\end{align*}
for some $\xi_1 \in \mS(\tau_1).$
In particular, for any fixed $k$ and  sufficiently large $\ell_1,\ell_2$ we have $$\frac{1}{2} < \bigg(1+O( \dfrac{1}{\sqrt{\min( \ell_1,\ell_2)}}+ \frac{k}{n}) < \frac{3}{2} $$  for sufficiently large $n$, that then leads to there exists some universal $c,C$ so that  \begin{align}\label{eq-lb-error-t1}
\dfrac{c}{n x_1}\big( \dfrac{x_2}{n\Delta_1}+ \dfrac{x_3}{n\Delta_2} \big) < 
 	\lvert \omega_1-\tau_1\rvert  < \dfrac{C}{n x_1}\big( \dfrac{x_2}{n\Delta_1}+ \dfrac{x_3}{n\Delta_2} \big)  . 
\end{align}

For $\omega_2,$ we have by Lemma~\ref{lem-error-formula} $\omega_2 \in \mS(\tau_2)$, thus for \begin{align*}
	h_1(\tau): = \w(\tau)^*(\bm I- \w(\omega_1)\w(\omega_1)^* )\y = h(\tau) - h (\omega_1) D_n(\tau-\omega_1), 
\end{align*} 
we have $h_1'(\omega_2) = 0$ implies $h'(\omega_2) = h(\omega_1)D_n'(\omega_2  - \omega_1)$, i.e. \begin{align*}
 		x_2 D_n'(\omega_2 - \tau_2)  + \bigg(1+O( \dfrac{1}{\sqrt{\min( n\ell_1,n\ell_2)}}+ \frac{k}{n}) \bigg)\bigg( (-1)^{\ell_1} x_1 \dfrac{n}{\ell_1}+(-1)^{\ell_3} x_3 \dfrac{n}{\ell_3} \bigg) = h(\omega_1)D_n'(\omega_2-\omega_1).
\end{align*}
Now noticing $h(\omega_1)D_n'(\omega_2-\omega_1) = O(\frac{\lvert x_1\rvert}{\Delta_3}  ) $, we get there exists some $C'$ so that
\begin{equation}\label{eq-lb-error-t2} \lvert \omega_2 - \tau_2\rvert \leq C'\dfrac{x_1}{n x_2}\big( \dfrac{1}{n\Delta_1}+\dfrac{1}{n\Delta_3} \big) \end{equation} 
Now for $\omega_3:$  \begin{align*}
	\w(\tau)' (I-P_3) \y 
	&= f(\tau) - D_n(\tau - \omega_1) f(\omega_1)- \dfrac{ [D_n(\tau - \omega_2) - D_n(\omega_2-\omega_1) D_n(\tau - \omega_1)][f(\omega_2) - D_n(\omega_2 - \omega_1) f(\omega_1)  ]}{\big(1 -  D_n(\omega_2 - \omega_1)\big)^2} \\
	& = f(\tau) -  D_n(\tau-\omega_1)\big( f(\omega_1)- \dfrac{D_n(\omega_2 - \omega_1)f(\omega_2) - D_n^2(\omega_2-\omega_1) f(\omega_1)}{(1-D_n(\omega_2-\omega_1))^2} \big) \\
	&\quad - D_n(\tau - \omega_2)  \dfrac{ f(\omega_2)-D_n(\omega_2-\omega_1) f(\omega_1)}{(1-D_n(\omega_2-\omega_1))^2}  
\end{align*}
We can rewrite $h_2$ as \begin{align*}
h_2(\tau):= \w(\tau)' (I-P_3) \y = H_1 x_1+H_2x_2+H_3x_3, 
 \end{align*}
 with \allowdisplaybreaks\begin{align*}
 	H_1 &= D_n(\tau - \tau_1) - D_n(\tau - \omega_1)\big[ D_n(\omega_1 - \tau_1)  - \dfrac{D_n(\omega_2-\omega_1) D_n(\omega_2-\tau_1) - D_n^2(\omega_2-\omega_1) D_n (\omega_1-\tau_1)}{(1-D_n(\omega_2-\omega_1))^2}\big]\\
 	&\quad  - D_n(\tau - \omega_2)\dfrac{ D_n(\tau_1-\omega_2) - D_n(\omega_2-\omega_1)D_n(\tau_1 - \omega_1)}{(1-D_n(\omega_2-\omega_1))^2} ,\\
 	H_2 &= D_n(\tau - \tau_2) - D_n(\tau - \omega_1)\big[  D_n(\omega_1 - \tau_2)  - \dfrac{D_n(\omega_2-\omega_1) D_n(\omega_2-\tau_2) - D_n^2(\omega_2-\omega_1) D_n (\omega_1-\tau_2)}{(1-D_n(\omega_2-\omega_1))^2}\big] \\
 	&\quad  - D_n(\tau - \omega_2)\dfrac{ D_n(\tau_2-\omega_2) - D_n(\omega_2-\omega_1)D_n(\tau_2 - \omega_1)}{(1-D_n(\omega_2-\omega_1))^2} ,\\
 		H_3 &= D_n(\tau - \tau_3) - D_n(\tau - \omega_1)\big[  D_n(\omega_1 - \tau_3)  - \dfrac{D_n(\omega_2-\omega_1) D_n(\omega_2-\tau_3) - D_n^2(\omega_2-\omega_1) D_n (\omega_1-\tau_3)}{(1-D_n(\omega_2-\omega_1))^2}\big]\\
 	&\quad  - D_n(\tau - \omega_2)\dfrac{ D_n(\tau_3-\omega_2) - D_n(\omega_2-\omega_1)D_n(\tau_3 - \omega_1)}{(1-D_n(\omega_2-\omega_1))^2} ,	
 \end{align*}
 In particular, for $\tau \geq \tau_3 - \Delta_2/2$, we have \begin{align*}
 	h_2(\tau_1) - h_2(\tau) = \sum_{i=1}^3 {\big( H_i(\tau_1) - H_i(\tau) \big)} x_i.
 \end{align*} 

\noindent \textbf{I.1 Lower bound on $\lvert H_1(\tau_1)\rvert $:}   
\begin{align*}
	H_1(\tau_1) &= 1 - D_n(\tau_1- \omega_1)^2 + \dfrac{ - D_n^2(\tau_1-\omega_1) D_n^2(\omega_2-\omega_1) + 2 D_n(\tau_1-\omega_1)D_n(  \omega_2-\omega_1)D_n(\omega_2-\tau_1)-D_n^2(\tau_1-\omega_2) }{(1-D_n(\omega_2-\omega_1))^2}\\
	& = 1 - D_n(\tau_1-\omega_1)^2- \dfrac{\big(D_n(\omega_2-\omega_1) D_n(\tau_1-\omega_1) - D_n(\tau_1-\omega_2)\big)^2 }{(1-D_n(\omega_2-\omega_1))^2}
\end{align*}
denoting $\epsilon_1 = \lvert \omega_1-\tau_1\rvert,\epsilon_2 = \lvert \omega_2-\tau_2\rvert,$ we have \begin{align*}
	\lvert D_n(\omega_2-\omega_1)D_n(\tau_1-\omega_1) -D_n(\tau_1-\omega_2)\rvert &= \lvert D_n(\omega_2-\omega_1)[D_n(\tau_1-\omega_1)-1] + D_n(\omega_2-\omega_1)-D_n(\tau_1-\omega_2) \rvert\\
	& \leq \dfrac{1}{2}\lvert D_n''(\xi_1) \epsilon_1^2\rvert + \dfrac{\epsilon_1}{\Delta_1 - \epsilon_1-\epsilon_2}.
\end{align*}
i.e. the second term is of order $O\bigg(  \big(\lvert D_n''(\xi_1) \epsilon_1^2\rvert + \dfrac{\epsilon_1}{\Delta_1 - \epsilon_1-\epsilon_2}\big)^2 \bigg)$. \\
For the first term, we have \begin{align*}
	\lvert 1 - D_n(\tau_1- \omega_1)^2\rvert = \lvert 1+D_n(\tau_1-\omega_1)\rvert \cdot \dfrac{1}{2} D_n''(\xi_1) \epsilon_1^2
\end{align*}
i.e. \begin{align*}
	\lvert H_1(\tau_1)\rvert \geq   \lvert 1+D_n(\tau_1-\omega_1)\rvert \cdot \dfrac{1}{2} D_n''(\xi_1) \epsilon_1^2 - O\bigg(  \big(\lvert D_n''(\xi_1) \epsilon_1^2\rvert + \dfrac{\epsilon_1}{\Delta_1 - \epsilon_1-\epsilon_2}\big)^2 \bigg)
\end{align*}
\noindent \textbf{I.2 Upper bound on $\lvert H_2(\tau_1)\rvert $:}   
\\
\begin{align*}
	&H_2(\tau_1) = \underbrace{D_n(\tau_1-\tau_2) - D_n(\tau_1-\omega_1) D_n(\omega_1-\tau_2)}_{J_1}+{\dfrac{D_n(\tau_1-\omega_1)[ D_n(\omega_2-\omega_1) D_n(\omega_2-\tau_2) - D_n^2(\omega_2-\omega_1)D_n(\omega_1-\tau_2)]}{(1-D_n(\omega_2-\omega_1))^2}}\\
	& {-\dfrac{  D_n(\tau_1-\omega_2) D_n(\tau_2-\omega_2)-D_n(\omega_2-\omega_1)D_n(\tau_2-\omega_1)D_n(\tau_1-\omega_2)}{(1-D_n(\omega_2-\omega_1))^2}}
\end{align*}
notice after multiplying $(1-D_n(\omega_2-\omega_1))^2,$ the terms except $J_1$ can be arranged as\begin{align*}
&\underbrace{D_n(\omega_2-\omega_1)D_n(\tau_2-\omega_1)D_n(\tau_1-\omega_2)  -D_n(\tau_1-\omega_1)  D_n^2(\omega_2-\omega_1)D_n(\omega_1-\tau_2) }_{J_{21}}\\
&+ \underbrace{[D_n(\tau_1-\omega_1)D_n(\omega_2-\omega_1)-D_n(\tau_1-\omega_2) ]D_n(\omega_2-\tau_2)}_{J_{22}}
\end{align*}
and $J_{21} =  O (\dfrac{1}{n(\Delta_1 -\epsilon_1-\epsilon_2)^3})$. \begin{align*}
	  J_{22} &= [D_n(\tau_1-\omega_1)D_n(\omega_2-\omega_1)-D_n(\tau_1-\omega_2) ]D_n(\omega_2-\tau_2)\\
	& =  (1-D_n(\tau_1-\omega_1) )  D_n(\omega_2-\omega_1)+ [ D_n(\omega_2-\omega_1)-D_n( \tau_1-\omega_2)] D_n(\omega_2-\tau_2) \\
	& = O( \dfrac{D_n''(\xi_1)\epsilon_1^2}{n(\Delta_1-\epsilon_1-\epsilon_2)} )+ [ D_n(\omega_2-\omega_1)-D_n( \tau_1-\omega_2)] D_n(\omega_2-\tau_2).
\end{align*}
we get $$H_2(\tau_1) = D_n(\tau_1-\omega_1) (D_n(\tau_1-\tau_2) - D_n(\omega_1-\tau_2))+ \dfrac{D_n(\omega_2-\tau_2)(D_n(\omega_2-\omega_1)-D_n(\tau_2-\omega_1))}{(1-D_n(\omega_2-\omega_1))^2}  +O(\dfrac{1}{n^3\Delta^3} )  $$
In particular, we have omitted the high-order terms,  \begin{align*}
	&\lvert H_1(\tau_1)\rvert > \lvert H_2(\tau_1) \rvert \iff  \\
 	&( 1+D_n(\tau_1-\omega_1)) ( 1-D_n(\tau_1-\omega_1)) > D_n(\tau_1-\omega_1) (D_n(\tau_1-\tau_2) - D_n(\omega_1-\tau_2))\\
 	 &+ {D_n(\omega_2-\tau_2)(D_n(\omega_2-\omega_1)-D_n(\tau_1-\omega_2))} \end{align*}
 	 notice now that \begin{align*}
 	 	\text{RHS} &= D_n(\tau_1-\tau_2)-D_n(\omega_1-\tau_2)+D_n(\omega_2-\omega_1)-D_n(\tau_1-\omega_2)   + O( \dfrac{1}{(n\Delta)^3} )\\
 	 	& = D_n'(\xi_{11}) (\tau_1-\omega_1) - D_n'(\xi_{22})(\tau_1-\omega_1) + O( \dfrac{1}{(n\Delta_1)^3} )  	 \end{align*}
where $\xi_{11} $ is between $\tau_1-\tau_2$ and $\omega_1-\tau_2,$ $\xi_{22} $ is between $\omega_1-\omega_2$ and $\tau_1-\omega_2$.  Thus we have \begin{align*}
	\lvert  \text{RHS} \rvert \leq  D_n''(\xi_4) (\tau_1-\omega_1)^2+ O(\dfrac{1}{(n\Delta_1)^3})
\end{align*}
Finally noticing that $\lvert \xi_4 \rvert \geq \dfrac{1}{\Delta_1 - \epsilon_1-\epsilon_2},$ we get then \begin{align*}
	D_n''(\xi_4)  \lesssim \dfrac{\pi n}{\Delta_1-\epsilon_1-\epsilon_2}.
\end{align*}
And our condition turns to \begin{align*}
	D''(\xi_1)\epsilon_1^2 \gtrsim  \dfrac{ n}{\Delta_1-\epsilon_1-\epsilon_2} \epsilon_1^2
\end{align*}
i.e. \begin{align*}
	D''(\xi_1) \gtrsim  \dfrac{ n^2}{n(\Delta_1-\epsilon_1-\epsilon_2)}
\end{align*}
which is satisfied for large enough $\ell_1$.\\

In conclusion, we get the following inequalites:\allowdisplaybreaks \begin{align*}
 &\lvert h_2(\tau_1)\rvert \gtrsim  \epsilon_1^2	\big ( D_n''(\xi_1) \lvert x_1\rvert   -  \dfrac{ n^2}{n(\Delta_1-\epsilon_1-\epsilon_2)} \lvert x_2\rvert\big) - \lvert x_3\rvert,\\
 &\sup_{\lvert \tau -\tau_3 \rvert\leq \Delta_2/2 } \lvert h_2  (\tau) \rvert \lesssim  \lvert x_3\rvert+ \lvert \dfrac{1}{n\Delta_2}\rvert^2  ( \lvert x_1\rvert+\lvert x_2\rvert )
\end{align*}
now by $\epsilon_1\gtrsim \dfrac{  \lvert x_2\rvert  }{\lvert x_1\rvert n(n\Delta_1 )},D_n''(\xi_1)\gtrsim{n^2}$ , we get then \begin{align*}
	[\dfrac{\lvert x_2\rvert}{\lvert x_1\rvert (n\Delta_1)^2} -  \dfrac{2\pi\lvert x_2\rvert^2 }{(n\Delta_1)^3 \lvert x_1\rvert^2 } ] \lvert x_2\rvert \gtrsim  \dfrac{\lvert x_1\rvert+\lvert x_2\rvert}{(n\Delta_2)^2 }  -\lvert x_3\rvert
\end{align*}
is sufficient to guarantee $\lvert h_2(\tau_1)\rvert > \sup_{\lvert \tau -\tau_3 \rvert\leq \Delta_2/2 }  \lvert h_2 (\tau) \rvert .$ In particular, letting $x_1 = 2x_2, $ and let $\Delta_2 = L\Delta_1$ for some absolute constant $L$ to be chosen, we get the condition turns to \begin{align}
{(n\Delta_1)^2}\big( 1  - O( \dfrac{1 }{n\Delta_1} + \dfrac{2}{L^2} )  	\big)\lvert x_2\rvert \gtrsim   \lvert x_3\rvert
\end{align}
That implies for $n\Delta_1, L, $ large enough so that $ 1  - O( \dfrac{1 }{n\Delta_1} + \dfrac{2}{L^2} ) >0$, we get \begin{align*}
	\dfrac{\lvert x_2\rvert}{\lvert x_1\rvert} \gtrsim  1  - O( \dfrac{1 }{n\Delta_1} + \dfrac{2}{L^2} ) \implies \lvert h_2(\tau_1)\rvert > \sup_{\lvert \tau -\tau_3 \rvert\leq \Delta_2/2 }  \lvert h_2 (\tau) \rvert,
\end{align*}
that leads to the desired claim.

\section{Auxiliary results}

\subsection{Proof of Lemma~\ref{lem-b-bound-full}}
\begin{proof}
	Noticing that $\lVert (\bm{F}_t^* \bm{F}_t )^{-1}\rVert_{\infty,\infty} = \lVert (\bm{F}_t^* \bm{F}_t )^{-1}\rVert_{1,1}$, and \begin{align*}
		\bm{P}_t \bm{v} = \bm{F}_t (\bm{F}_t^* \bm{F}_t )^{-1} \bm{F}_t^* \bm{v} = \sum_{i=1}^t  \underbrace{[(\bm{F}_t^* \bm{F}_t )^{-1} \bm{F}_t^* \bm{v}]_i}_{b_i} \w(\omega_i),
	\end{align*}
now by  \begin{align*}
		\lvert (\bm{F}_t^* \bm{F}_t )_{ii} \rvert - \sum_{j\neq i}  \lvert (\bm{F}_t^* \bm{F}_t )_{ij } \rvert &\geq 1 - \sum_{j\neq i} \lvert K_n(\omega_j - \omega_i)\rvert \\
		& \geq 1 - \sum_{j\neq i} \dfrac{1}{(n+2) ^4(\omega_j- \omega_i )^4} \\
		& \geq 1- \sum_{j = 1}^\infty [\dfrac{1}{(n+2)^4(j\Delta- 2\max_i  \lvert \omega_i - \tau_{T(i)} \rvert)^4     }  + \dfrac{1}{(n+2)^4(j\Delta- \max_i  \lvert \omega_i - \tau_{T(i)} \rvert)^4     } ] \\ 
		& \geq  1- \dfrac{1}{(n+2)^4 \Delta^4} \sum_{j = 1}^\infty \big[ \dfrac{1}{(j - \frac{2n\varepsilon_t}{n\Delta} )^4}+ \dfrac{1}{(j- \frac{n\varepsilon_t}{n\Delta})^4 }\big]   \\
    &\geq 1- \dfrac{c_3(\varepsilon_t,\Delta)}{n^4\Delta^4} \quad \forall i \in [t],
	\end{align*}
		
	we get then \begin{align*}
		\lVert (\bm{F}_t^* \bm{F}_t )^{-1} \rVert_{\infty,\infty} \leq  \big( 1- \dfrac{c_3(\varepsilon_t,\Delta)}{n^4\Delta^4}\big)^{-1}.	\end{align*}
	Thus  $\lVert \bm b\rVert_\infty \leq   \big( 1- \dfrac{c_3(\varepsilon_t,\Delta)}{n^4\Delta^4}\big)^{-1} \lVert \bm{F}_t^* \bm{v}\rVert_\infty.$
	
\end{proof}

\subsection{Proof of Lemma~\ref{lem-coefficient-bound}}

\begin{proof} 
W.L.O.G. supposing $T(i)=i,$	
 we have firstly, for $i\leq t,$ \begin{align*}
	[  (\bm{F}_t^* \bm{F}_t)^{-1}  \bm{F}_t^* \w(\tau_i)    ]_m  &= \e_m^*  (\bm{F}_t^* \bm{F}_t)^{-1}  \bm{F}_t^*[ \w(\tau_i)- \w(\omega_i)] + \delta_{mi}\\
	& =  \sum_{k = 1}^t \w(\omega_k)^*[\w(\tau_i)- \w(\omega_i) ] \e_m^* (\W_t^*\W_t)^{-1} \e_k  + \delta_{mi}  
\end{align*}
As a result, we get \begin{align*}
	&[\sum_{i=1}^t x_i  (\bm{F}_t^* \bm{F}_t)^{-1}  \bm{F}_t^* \w(\tau_i)    ]_m \\
	= &\sum_{i=1}^t x_i \e_m^*  (\bm{F}_t^* \bm{F}_t)^{-1}  \bm{F}_t^* (\w(\tau_i)\pm \w(\omega_i))      \\
	 =&   \sum_{i=1}^t x_i\bigg( \sum_{k = 1}^t \w(\omega_k)^*[\w(\tau_i)- \w(\omega_i) ] \e_m^* (\W_t^*\W_t)^{-1} \e_k  + \delta_{mi} \bigg)\\
	 =& x_m+  \underbrace{\sum_{k=1}^t \sum_{i=1}^t x_i  [K_n(\tau_i-\omega_k) - K_n  (\omega_i-\omega_k) ]   \e_m^*(\W_t^*\W_t)^{-1}\e_k}_{U_1}.
	\end{align*}
	Now notice that \begin{align*}
	 \lvert U_1\rvert &\leq \max_{1\leq k\leq t} \sum_{i=1}^t   \lvert x_i(\tau_i - \omega_i) K_n'(\omega_k - \xi_i) \rvert \cdot \underbrace{\sum_{k=1}^t \lvert \e_m^*(\W_t^*\W_t)^{-1} \e_k\rvert}_{\leq  \lVert (\W_t^*\W_t)^{-1}\rVert_{\infty,\infty}}  \\
	 &\leq  n\varepsilon_{\x,t}\big(\dfrac{\pi^2}{3} +\pi^2\dfrac{c_2(\varepsilon_t,\Delta)} {(n\Delta)^4} \big)  \cdot \big(1- \dfrac{c_3(\varepsilon_t,\Delta)}{n^4\Delta^4})^{-1},
	\end{align*}
	On the other hand, for $i>t$, we have  \begin{align*}
	 \lvert  \sum_{i>t} x_i	[   (\bm{F}_t^* \bm{F}_t)^{-1}   \bm{F}_t^* \w(\tau_i)    ]_m  \rvert & =  \lvert   \sum_{k\leq t} \sum_{i>t} x_i \w(\omega_k)^* \w(\tau_i) \e_m^*     (\bm{F}_t^* \bm{F}_t)^{-1} \e_k   \rvert \\
	 &\leq \big\lvert \max_{k\leq t} \lvert \sum_{i> t} x_i K_n(\omega_k - \tau_i) \rvert \cdot \underbrace{\lvert \sum_{k\leq t} \e_m^* (\bm{F}_t^* \bm{F}_t)^{-1} \e_k \rvert }_{= O(\lVert (\W_t^*\W_t)^{-1}\rVert_{\infty,\infty}  ) }    \big\rvert\\
	 &=  \dfrac{c_2(\varepsilon_t,\Delta)  \lVert \x_{>t}\rVert_\infty }{n^4\Delta^4 } \cdot  \big(1- \dfrac{c_3(\varepsilon_t,\Delta)}{n^4\Delta^4})^{-1}
	\end{align*}  
	Thus the claim holds.
\end{proof}

\subsection{Coefficient Bounds: Sub-Sampled Version}

\begin{lemma}\label{lem-subsampled-coefficient-bound}
	For $\bm{F}_t = [\w(\omega_1),\dots,\w(\omega_t) ]$ and $\bm{v}$ a vector in $\R^n,$ we have $\bm{P}_t \bm{v} = \sum_{i=1}^t b_i  \w(\omega_i),$ with \begin{align*}
		 \lVert \bm{b}\rVert_\infty &\leq p^{-1}\big (1-  \dfrac{c_3(\varepsilon_t,\Delta)  }{n^4\Delta^4} - cs \sqrt{\log(n)/np} \big )^{-1}   \cdot \lVert  \bm{F}_t^* \bm{v} \rVert_\infty \\
		 \lVert \bm{b}\rVert_1 &\leq p^{-1}\big (1-  \dfrac{c_3(\varepsilon_t,\Delta)  }{n^4\Delta^4} - cs \sqrt{\log(n)/np} \big )^{-1} \cdot \lVert  \bm{F}_t^* \bm{v} \rVert_1 \\  
	\end{align*}
\end{lemma}

\begin{proof}
	Noticing that $\lVert (\bm{F}_t^* \bm{F}_t )^{-1}\rVert_{\infty,\infty} = \lVert (\bm{F}_t^* \bm{F}_t )^{-1}\rVert_{1,1}$, and \begin{align*}
		\bm{P}_t \bm{v} = \bm{F}_t (\bm{F}_t^* \bm{F}_t )^{-1} \bm{F}_t^* \bm{v} = \sum_{i=1}^t  \underbrace{[(\bm{F}_t^* \bm{F}_t )^{-1} \bm{F}_t^* \bm{v}]_i}_{b_i} \w(\omega_i),
	\end{align*}
	we need only to upper bound  $\lVert (\bm{F}_t^* \bm{F}_t )^{-1}\rVert_{\infty,\infty}$: By  \begin{align*}
		\lvert (\bm{F}_t^* \bm{F}_t )_{ii} \rvert - \sum_{j\neq i}  \lvert (\bm{F}_t^* \bm{F}_t )_{ij } \rvert \geq p(1-  \dfrac{c_3(\varepsilon_t,\Delta)  }{n^4\Delta^4} - cs \sqrt{\log(n)/np} ) \quad \forall i \in [t],
	\end{align*}
	we get then \begin{align*}
		\lVert (\bm{F}_t^* \bm{F}_t )^{-1} \rVert_{\infty,\infty} \leq  p^{-1}\big (1-  \dfrac{c_3(\varepsilon_t,\Delta)  }{n^4\Delta^4} - cs \sqrt{\log(n)/np} \big )^{-1} .
	\end{align*}
	Thus the claim holds.
\end{proof}

\begin{lemma}\label{lem-subsampled-xhat-bound}
	 We have  \begin{align*}\lvert [  (\bm{F}_t^* \bm{F}_t)^{-1}  \bm{F}_t^* \sum_{i=1}^n x_i \w(\tau_i)     ]_m - x_m \rvert \leq &  \big(1- \dfrac{c_3(\varepsilon_t,\Delta)}{n^4\Delta^4}-cs\sqrt{\frac{\log n}{np}})^{-1} \cdot\bigg(n \varepsilon_{\x,t} (\dfrac{\pi^2}{3} +\dfrac{\pi^2c_2(\varepsilon_t,\Delta)}{(n\Delta)^4}+cs\sqrt{\log(n)/np}) \\
	 & +  \lVert \x(\mT_t^c)\rVert_{\infty} \big[\dfrac{   c_2(\varepsilon_t,\Delta)}{(n\Delta)^4} + cs\sqrt{\dfrac{\log(n)}{np}} \big]  \bigg) 
	 \end{align*}  
	 as long as $\varepsilon_t \leq \frac{1}{2n+4}.$
\end{lemma}
\begin{proof} 
W.L.O.G. supposing $T(i)=i,$	
 we have firstly, for $i\leq t,$ \begin{align*}
	[  (\bm{F}_t^* \bm{F}_t)^{-1}  \bm{F}_t^* \w(\tau_i)    ]_m  &= \e_m^*  (\bm{F}_t^* \bm{F}_t)^{-1}  \bm{F}_t^*[ \w(\tau_i)- \w(\omega_i)] + \delta_{mi}\\
	& =  \sum_{k = 1}^t \w(\omega_k)^*[\w(\tau_i)- \w(\omega_i) ] \e_m^* (\W_t^*\W_t)^{-1} \e_k  + \delta_{mi}  
\end{align*}
As a result, we get \begin{align*}
	&[\sum_{i=1}^t x_i  (\bm{F}_t^* \bm{F}_t)^{-1}  \bm{F}_t^* \w(\tau_i)    ]_m \\
	= &\sum_{i=1}^t x_i \e_m^*  (\bm{F}_t^* \bm{F}_t)^{-1}  \bm{F}_t^* (\w(\tau_i)\pm \w(\omega_i))      \\
	 =&   \sum_{i=1}^t x_i\bigg( \sum_{k = 1}^t \w(\omega_k)^*[\w(\tau_i)- \w(\omega_i) ] \e_m^* (\W_t^*\W_t)^{-1} \e_k  + \delta_{mi} \bigg)\\
	 =& x_m+  \underbrace{\sum_{k=1}^t \sum_{i=1}^t x_i  [\tilde K_n(\tau_i-\omega_k) - \tilde K_n  (\omega_i-\omega_k) ]   \e_m^*(\W_t^*\W_t)^{-1}\e_k}_{U_1}.
	\end{align*}
	Now notice that \begin{align*}
	 \lvert U_1\rvert &\leq \max_{1\leq k\leq t} \sum_{i=1}^t   \lvert x_i(\tau_i - \omega_i) \tilde K_n'(\omega_k - \xi_i) \rvert \cdot \underbrace{\sum_{k=1}^t \lvert \e_m^*(\W_t^*\W_t)^{-1} \e_k\rvert}_{\leq  \lVert (\W_t^*\W_t)^{-1}\rVert_{\infty,\infty}}  \\
	 &\leq  n\varepsilon_{\x,t} \cdot\big(\dfrac{\pi^2}{3} +\dfrac{\pi^2c_2(\varepsilon_t,\Delta)} {(n\Delta)^4} +cs \sqrt{\frac{\log n}{np}} \big)  \cdot \big(1- \dfrac{c_3(\varepsilon_t,\Delta)}{n^4\Delta^4} - cs \sqrt{\frac{\log n}{np}}   )^{-1},
	\end{align*}
	On the other hand, for $i>t$, we have  \begin{align*}
	 \lvert  \sum_{i>t} x_i	[   (\bm{F}_t^* \bm{F}_t)^{-1}   \bm{F}_t^* \w(\tau_i)    ]_m  \rvert & =  \lvert   \sum_{k\leq t} \sum_{i>t} x_i \w(\omega_k)^* \w(\tau_i) \e_m^*     (\bm{F}_t^* \bm{F}_t)^{-1} \e_k   \rvert \\
	 &\leq \big\lvert \max_{k\leq t} \lvert \sum_{i> t} x_i K_n(\omega_k - \tau_i) \rvert \cdot \underbrace{\lvert \sum_{k\leq t} \e_m^* (\bm{F}_t^* \bm{F}_t)^{-1} \e_k \rvert }_{= O(\lVert (\W_t^*\W_t)^{-1}\rVert_{\infty,\infty}  ) }    \big\rvert\\
	 &= \lVert \x_{>t}\rVert_\infty   (\dfrac{c_2(\varepsilon_t,\Delta)  }{n^4\Delta^4 }+cs\sqrt{\frac{\log n}{np}}) \cdot  \big(1- \dfrac{c_3(\varepsilon_t,\Delta)}{n^4\Delta^4}-cs \sqrt{\frac{\log n}{np}} )^{-1}
	\end{align*}  
	Thus the claim holds by noticing $\lvert x_m \rvert \leq 2\lVert \x_{>t}\rVert_\infty $.
\end{proof}

\subsection{Proof of Proposition~\ref{prop-squared-fejer-kernel}}\label{appendix-kernel-inequality}
We develop the Proposition~\ref{prop-squared-fejer-kernel} based on the following Lemma in \cite{Candes2014}:
\begin{lemma}
	For all $t\in [-1/2,1/2]$, we have \begin{align}
 K(t) & \geq 1- \dfrac{\pi^2}{6} n(n+4)t^2,\\
 \lvert K'(t) \rvert &\leq \dfrac{\pi^2}{3}n(n+4) t,\\
 K''(t) &\leq -\dfrac{\pi^2}{3}n(n+4) +\dfrac{\pi^4}{6}(n+2)^4 t^2,\\
  	\lvert K''(t)\rvert &\leq \dfrac{\pi^2}{3} n(n+4).
 \end{align}

\end{lemma}

\begin{proof}
\textbf{Inequality1: }By $\lvert \sin(\pi x) \rvert\geq \lvert 2x \rvert$ when $-0.5<x <0.5$ we have $ \lvert K_{n}(\tau) \rvert \leq \dfrac{1}{( n+2)^4\tau^4}.$ On the other hand, when $ ((n +2)\lvert \tau\rvert)^4>\frac{10}{7}, $ the second bound is trivial; when $0.5<(n+2)\lvert \tau \rvert< (10/7)^{1/4}, $ we have  \begin{align*}
	\big(\dfrac{\sin\big((n/2+1 ) \pi t\big )}{(n/2+1) \sin(\pi t)}\big)^4 < 0.7 \iff \dfrac{\lvert  \sin \big((n/2+1) \pi t\big) \rvert}{\lvert \sin (\pi t) \rvert  }< 0.7^{1/4}  (n/2+1).
\end{align*}
By \begin{align*}
	\dfrac{d}{dt} \dfrac{\sin\big ((n/2+1)\pi t\big )}{\sin \big(\pi t \big)} &= \dfrac{\pi }{\sin^2\big(\pi t \big)} \cdot \bigg((n/2+1)\cos\big((n/2+1)\pi t\big) \sin (\pi t) - \cos(\pi t) \sin\big((n/2+1)\pi t \big) \bigg)
\end{align*}
and when $\frac{1}{2n+4}< t< \frac{0.76}{n+2} ,$  \begin{align*}
	 \cos \big((n/2+1)\pi t\big) \sin (\pi t) &\leq \cos \big( (\frac{n}{2}+1)   \dfrac{\pi }{2n+4}    \big)  \pi t = \cos (\frac{\pi }{4})\cdot  \pi t, \\
	 \cos(\pi t) \sin\big( (n/2+1)\pi t  \big) &\geq  \big(1-\frac{4\pi^2}{n^2} \big)  \big( \dfrac{1-\sin(\pi /4)}{1/4 }(n/2+1)t  + 2\sin(\frac{\pi }{4})-1     \big) .  
\end{align*}
where the second line is because when $\frac{1}{4} \leq t\leq \frac{1}{2},$ \begin{align*}
	 \sin (\pi t) &\geq 4(1-\sin(\pi/4))t + 2 \sin (\pi /4)-1. 	  \end{align*}
Then as a result we have \begin{align*}
	\dfrac{d}{dt} \dfrac{\sin\big ((n/2+1)\pi t\big )}{\sin \big(\pi t \big)} \leq \dfrac{\pi }{\sin^2(\pi t)} \bigg( \big[ \dfrac{\sqrt{2}\pi }{2} - (1-\frac{4\pi^2}{n^2})(4- 2\sqrt{2})   ) \big](n/2+1) t - (1-\frac{4\pi^2}{n^2})\big[\sqrt{2} - 1\big]  \bigg),\\
\end{align*}
then by when $n>20\pi , $ $$(n/2+1)t < 0.38 \implies \big[ \dfrac{\sqrt{2}\pi }{2} - (1-\frac{4\pi^2}{n^2})( 4- 2\sqrt{2})   \big](n/2+1) t - (1-\frac{4\pi^2}{n^2}) \big[\sqrt{2} - 1\big]  0,  $$
 which then implies $\frac{d}{dt} \frac{\sin\big ((n/2+1)\pi t\big )}{\sin \big(\pi t \big)} <0$ when $\frac{1}{2n}< t < \frac{0.76}{n+2}$. \\
When $\frac{0.76}{n+2} <t<\frac{1}{n+2}$, we have \begin{align*}
	 \cos \big((n/2+1)\pi t\big) \sin (\pi t) &\leq \pi \big(\frac{1}{2}-(n/2+1)t \big) \sin(\pi t) \leq 0.12\pi^2 t, \\
	 	 \cos(\pi t) \sin\big( (n/2+1)\pi t  \big) &\geq  \big(1-O(\frac{1}{n}) \big)  \big( \dfrac{1-\sin(\pi /4)}{1/4 }(n/2+1)t  + 2\sin(\frac{\pi }{4})-1     \big) .  
\end{align*} 
Thus $$  (n/2+1)t\leq 0.5 \implies  \big[0.12\pi^2-4+2\sqrt{2} \big] ( n/2+1) t  - \big[\sqrt{2}-1 \big] \leq -0.38<0 , $$
which then implies  $\frac{d}{dt} \frac{\sin\big ((n/2+1)\pi t\big )}{\sin \big(\pi t \big)} <0$ when $\frac{0.76}{n+2}< t < \frac{1}{n+2}$.\\
When $\frac{1}{n+2} <t< \frac{(10/7)^4}{n+2},$ we have $\cos\big( (n/2+1)t \big) <0,$ thus  $\frac{d}{dt} \frac{\sin\big ((n/2+1)\pi t\big )}{\sin \big(\pi t \big)} <0$.

In conclusion, we have 
$$ \dfrac{d}{dt} \dfrac{\sin\big ((n/2+1)\pi t\big )}{\sin \big(\pi t \big)} <0 \text{ \quad when }\dfrac{1}{2n+1} <t<\dfrac{(10/7)^{1/4}}{n+2}. $$
Thus \begin{align*}
	\max_{\frac{1}{2n+4}\leq  t\leq  \frac{(10/7)^{1/4}}{n+2}} \lvert  \frac{\sin\big ((n/2+1)\pi t\big )}{\sin \big(\pi t \big)}  \rvert = \max \{\dfrac{\sin (\pi /4)}{\sin(\pi /(2n+4))} , \dfrac{\sin ((\frac{10}{7})^{1/4} \pi /2  )}{\sin \big((\frac{10}{7})^{1/4} \pi/(n+2)\big) }   \}
\end{align*}
Then the claim holds by the right-hand side is smaller or equal to $0.7^{1/4}(n/2+1)$

\textbf{Inequlity2: } By (9), we have \begin{align*}
	 \lvert 1-K_n(\tau) \rvert&\leq \dfrac{\pi^2}{6}n(n+4)\tau^2 \leq 4(n+2)^2\tau^2, \quad \forall \lvert \tau \rvert \leq \frac{1}{2n+4}.
\end{align*}
By (11), we have \begin{align*}
	 1 - K_n(\tau) & = -\int_0^{\tau} K_n'(t) dt\\
	 & = -\int_0^\tau \int_0^t K_n''(s) ds dt \\
	 & \geq   \int_0^\tau\int_0^t \dfrac{\pi^2}{3}n(n+4)-\dfrac{\pi^4}{6}(n+2)^4s^2 ds  dt \\
	 &=  \int_0^\tau  \dfrac{\pi^2}{3}n(n+4) t - \dfrac{\pi^4}{18} (n+2)^4 t^3 dt\\
	 & = \dfrac{\pi^2}{6}n(n+4)\tau^2 - \dfrac{\pi^4}{64}(n+2)^4 \tau^4\\
	 & \geq (\dfrac{\pi^2}{6} - \dfrac{\pi^4}{256}) (n+2)^2\tau^2 - \frac{2\pi^2}{3}\tau^2\\
	 & \geq (n+2)^2  \tau^2 , 
\end{align*}
where the last line is by $(\dfrac{\pi^2}{6} - \dfrac{\pi^4}{256}) (n+2)^2 -\dfrac{2\pi^2}{3}> (n+2)^2$ when $n>10.$\\
\textbf{Inequality3: } By calculating the derivative explicitly, we get \begin{align*}
	K_n'(t) = 4\pi  \big[ \dfrac{\sin((\frac{n}{2}+1)\pi t )}{(\frac{n}{2}+1)\sin(\pi t)} \big]^3\cdot \bigg( \dfrac{ (\frac{n}{2}+1 )\cos ((\frac{n}{2}+1) \pi t ) \sin(\pi t) - \sin((\frac{n}{2}+1)\pi t ) \cos (\pi t)  }{(\frac{n}{2}+1) \sin^2(\pi t) }  \bigg). 
\end{align*}
Thus  when $t>\frac{1}{2n+4}$,\begin{align*}
	\lvert K_n'(t)\rvert &\leq 4\pi \dfrac{1}{(n+2)^3t^3} \cdot  \big[\dfrac{(n/2+1) \pi t}{2(n+2)t^2 } + \dfrac{1}{(n+2)^2t^2 } \big]\leq \dfrac{\pi^2+\frac{4\pi }{(n+2)t} } {(n+2)^3t^4}\leq \dfrac{4\pi^2}{(n+2)^3 t^4}
\end{align*}
On the other hand, we have when $0\leq t\leq \frac{1}{2n+4}, $ \begin{align*}
	\lvert K_n'(t) \rvert \leq \dfrac{\pi^2}{3}n(n+4)t\leq \dfrac{\pi^2}{6} (n+2).
\end{align*} 
\textbf{Inequality4: }  When $0<\tau\leq \frac{1}{2n+4} $ \begin{align*}
	-K_n'(\tau) &= -\int_{0}^\tau K_n''(t) dt\\
	& \geq  \int_{0}^\tau \dfrac{\pi^2}{3}n(n+4)- \dfrac{\pi^4}{6}(n+2)^4 t^2 dt \\
	& = \dfrac{\pi^2}{3}n(n+4) \tau -\dfrac{\pi^4}{18}(n+2)^4 \tau^3   \\
	& \geq \big(\dfrac{\pi^2}{3}- \dfrac{\pi^4}{72} - \dfrac{4\pi^2}{3(n+2)^2}  \big) (n+2)^2\tau  ,
\end{align*}
In particular, when $4\pi^2/3(n+2)^2<0.36,$
we get \begin{align*}
	-K_n'(\tau) \tau  \geq 1.9 (n+2)^2\tau^2.
\end{align*}
On the other hand, \begin{align*}
	-K_n'(\tau) \leq \int_0^\tau \dfrac{\pi^2}{3}n(n+4) dt \leq \dfrac{\pi^2}{3}n(n+4) \tau ,
\end{align*}
thus \begin{align*}
	-K_n'(\tau) \tau \leq \dfrac{\pi^2}{3}(n+2)^2\tau^2\leq 4 (n+2)^2\tau^2.
\end{align*}
\noindent
\textbf{Inequality5:} \begin{align*}
	\lvert K_n''(t) \rvert \leq \dfrac{4\pi^4}{(n+2)^2t^4}
\end{align*} \begin{align*}
	 K''_n(t) =& 12\pi^2  \big[ \dfrac{\sin((\frac{n}{2}+1)\pi t )}{(\frac{n}{2}+1)\sin(\pi t)} \big]^2\cdot \bigg( \dfrac{ (\frac{n}{2}+1 )\cos ((\frac{n}{2}+1) \pi t ) \sin(\pi t) - \sin((\frac{n}{2}+1)\pi t ) \cos (\pi t)  }{(\frac{n}{2}+1) \sin^2(\pi t) }  \bigg)^2\\
	 & + 4\pi  \big[ \dfrac{\sin((\frac{n}{2}+1)\pi t )}{(\frac{n}{2}+1)\sin(\pi t)} \big]^3\cdot \dfrac{d}{d\tau }\bigg( \dfrac{ (\frac{n}{2}+1 )\cos ((\frac{n}{2}+1) \pi t ) \sin(\pi t) - \sin((\frac{n}{2}+1)\pi t ) \cos (\pi t)  }{(\frac{n}{2}+1) \sin^2(\pi t) }  \bigg) 
\end{align*}
As in we bounding $K_n'(t), $  the first term can be bounded by \begin{align*}
	 12\pi^2 \dfrac{1}{(n+2)^2t^2} \cdot  \big[\dfrac{\pi }{4t } + \dfrac{1}{(n+2)^2t^2 } \big]^2\leq \dfrac{{2\pi^4} +\frac{24\pi^2}{(n+2)^4t^2}}{(n+2)^2 t^4} \leq \dfrac{2\pi^4+1}{(n+2)^2t ^4} \end{align*} 
	 when $t > \frac{1}{2n+4}$ and $n > 40$.
Denoting  $$h(t): =   (\frac{n}{2}+1 )\cos ((\frac{n}{2}+1) \pi t ) \sin(\pi t) - \sin((\frac{n}{2}+1)\pi t ) \cos (\pi t),  $$
then we have \begin{align*}
	h'(t)=& -(\dfrac{n}{2}+1)^2\pi \sin\big ((\frac{n}{2}+1)\pi t \big )\sin (\pi t) + (\dfrac{n}{2} +1)\pi \cos \big((\frac{n}{2}+1 )\pi t \big)\cos(\pi t) \\
	&+(\dfrac{n}{2}+1)^2\pi   \cos \big((\dfrac{n}{2}+1 )\pi t\big) \cos (\pi t) - \sin((\frac{n}{2}+1) \pi t)\sin (\pi t).
\end{align*}
Thus \begin{align*}
	\lvert \frac{h'(t)}{\sin (\pi t)} \rvert & \leq  (\dfrac{n}{2}+1)^2\pi +\dfrac{2(n/2+1)^2\pi } {2t}+1 ,\\
		\lvert \frac{h(t)}{\sin (\pi t)} \rvert & \leq  (\dfrac{n}{2}+1) +\dfrac{1} {2t}.
\end{align*}
As a result, \begin{align*}
\lvert \dfrac{d}{dt } \dfrac{h(t)}{\sin^2(\pi t)}  \rvert &= \lvert \dfrac{h'(t)\sin^2(\pi t) - 2\pi \sin (\pi t)\cos (\pi t)h(t) }{\sin^4(\pi t)}\rvert \\
&\leq \dfrac{1}{2t}   \lvert \dfrac{h'(t)}{\sin (\pi t)}\rvert+ \dfrac{\pi }{2t^2} \lvert \dfrac{h(t)}{\sin(\pi t)}\rvert\\
& \leq \dfrac{1}{2t} \cdot \big[(\dfrac{n}{2}+1)^2\pi +\dfrac{2(n/2+1)^2\pi } {2t}+1 \big] + \dfrac{\pi }{2t^2}\cdot \big[ \dfrac{n}{2}+1 +\dfrac{1 } {2t}\big]\\
&\leq \dfrac{1}{t} \big[(\frac{n}{2}+1)^2{\pi }\big]+ \dfrac{1}{t^2} \big[ (n/2+1){\pi}    \big] +  \frac{\pi}{4t^3}. 
\end{align*}
Thus when $t > \frac{1}{2n+4},$ the second term of $K_n''(t)$ is bounded by \begin{align*}
\dfrac{8\pi}{(n+2)^4t^3} \cdot \big[ \dfrac{(n/2+1)^2\pi }{t}+\dfrac{(n/2+1)\pi}{t^2}+\dfrac{\pi}{4t^3}  \big]  &\leq \dfrac{8\pi^2}{(n+2)^4t^4}\big[ \frac{(n+2)^2}{4}+\dfrac{n+2}{2t}+\frac{1}{4t^2}   \big]\\
&\leq \dfrac{18\pi^2}{(n+2)^2t^4}.
\end{align*}
Now the claim is followed by combining the bounds for the first term and second term together and the fact $18\pi^2+1\leq 2\pi^4$.
\end{proof}

\subsection{Proof of Lemma~\ref{lem-grid-bound}}

Denoting $\hat{\tau}_{t+1}$ the nearest grid to $\tau_{T(t+1)},$ we have then \begin{align*}
	\lvert \w^*(\hat{\omega}_{t+1,\text{grid}})\bm r_{t} \rvert \geq \lvert \w^*(\hat{\tau}_{t+1,\text{grid}})\bm r_{t} \rvert \geq  \lvert \w^*(\tau_{T(t+1)})\bm r_{t} \rvert -\lvert  [\w(\tau_{T(t+1)})-\w(\hat{\omega}_{t+1,\text{grid}}) ]^* \bm r_{t} \rvert ,
\end{align*}
thus it sufficient to control $\lvert  [\w(\tau_{T(t+1)})-\w(\hat{\omega}_{t+1,\text{grid}}) ]^* \bm r_{t} \rvert:$ \begin{align*}
	[\w(\tau_{T(t+1)})-\w(\hat{\omega}_{t+1,\text{grid}}) ]^* \bm r_{t} &= [\w(\tau_{T(t+1)})-\w(\hat{\omega}_{t+1,\text{grid}}) ]^* (\bm I - \bm P(\bomega_{t}) ) \bm y  \\ 
	& = \underbrace{[\w(\tau_{T(t+1)})-\w(\hat{\omega}_{t+1,\text{grid}}) ]^* (\bm I - \bm P(\bomega_{t}) ) \x (\mT_t) \w(\mT_t) }_{I_1}\\
	& \quad +\underbrace{ [\w(\tau_{T(t+1)})-\w(\hat{\omega}_{t+1,\text{grid}}) ]^* (\bm I - \bm P(\bomega_{t}) ) \x (\mT_t^c) \w(\mT_t^c)}_{I_2}
\end{align*}
W.L.O.G. assume $T(i) = i$ for $i\leq t$ and $\lvert x_{t+1} \rvert = \lVert \x(\mT_t^c)\rVert_\infty$ and denoting $$\bm P(\bomega_t) \mz(\tau) = \sum_{i=1}^t b_i(\tau) \w(\omega_i),$$then \\
\textbf{Bounding $I_1$:}  we have 
 there exists $\xi_{t+1}$ between $\tau_{T(t+1)}$ and $\hat\omega_{t+1,\text{grid}}$  so that \begin{align*}
	I_1 &= \sum_{i=1}^t x_i   [\w(\tau_{T(t+1)})-\w(\hat{\omega}_{t+1,\text{grid}}) ]^* [\bm I - \bm P(\bomega_{t}) ] [\w(\tau_i)   - \w(\omega_i) ] \\
	& = \sum_{i=1}^t x_i (\tau_{T(t+1)} - \hat{\omega}_{t+1,\grid} )\cdot \big[ \mz (\xi_{t+1})  ^*[\w (\tau_i ) - \w(\omega_i)] + \sum_{k=1}^t b_k(\xi_{t+1}) \w(\omega_k)^*[\w(\tau_i) - \w(\omega_i)] \big]
\end{align*}

Now by the similar argument as we bounding $J_1$ in Section~\ref{appendix-J-small-bound}, we get \begin{align}
	I_1 \lesssim \dfrac{n}{N_{\text{grid}}}\cdot \dfrac{n\varepsilon_{\x,t}}{n^4\Delta^4}.
\end{align}
\textbf{Bounding $I_2$:} we have 
 there exists $\xi_{t+1}$ between $\tau_{T(t+1)}$ and $\hat\omega_{t+1,\text{grid}}$  so that \begin{align*}
	I_1 
	& = \sum_{i>t} x_i (\tau_{T(t+1)} - \hat{\omega}_{t+1,\grid} )\cdot \big[ \mz (\xi_{t+1})  ^*[\w (\tau_i ) - \w(\omega_i)] + \sum_{k=1}^t b_k(\xi_{t+1}) \w(\omega_k)^*[\w(\tau_i) - \w(\omega_i)] \big]\\
	&\lesssim \lvert x_{t+1}\rvert \dfrac{n}{N_{\grid}} \big(1+ \dfrac{1}{(n\Delta)^4}  \big) \lesssim \lvert x_{t+1}\rvert \dfrac{n}{N_{\grid}}.
\end{align*}
where the second inequality is shown by the similar argument as we bounding $J_2, J_3$ in Section~\ref{appendix-J-small-bound}.
Now combining the bound on $I_1,I_2,$ we get  \begin{align*}
	 [\w(\tau_{T(t+1)})-\w(\hat{\omega}_{t+1,\text{grid}}) ]^* \bm r_{t} = O( \dfrac{n}{N_{\grid}}\big[\dfrac{n\varepsilon_{\x,t}}{n^4\Delta^4}+\lVert \x(\mT_t^c)\rVert_\infty \big]  ),
\end{align*}
that completes the proof.

\end{document}